\documentclass[11pt,reqno,a4paper,final]{amsart}
\usepackage[british]{babel}
\usepackage{amssymb,amstext,amsthm,eucal,bbm,amscd,bbold}
\usepackage[margin=25mm]{geometry}
\usepackage{array, collcell}
\usepackage[svgnames,table]{xcolor}
\usepackage[T1]{fontenc}
\usepackage[utf8]{inputenc}
\usepackage[small]{eulervm}
\usepackage{tgpagella}
\usepackage{mathrsfs}
\usepackage{verbatim}
\usepackage{pdfsync}
\usepackage[final]{microtype}
\usepackage[nosort]{cite}
\usepackage{adjustbox}
\usepackage{booktabs} 
\usepackage{multirow}
\usepackage{makecell}
\usepackage{caption}
\usepackage{appendix}
\usepackage{chngcntr}
\usepackage{apptools}
\usepackage{romanbar,braket}
\usepackage{enumitem}
\usepackage{slashed}
\usepackage[boxsize=2pt,aligntableaux=bottom]{ytableau}

\newcommand\AddLabel[1]{\refstepcounter{equation}(\theequation)\label{#1}} 
\newcolumntype{L}{>{\collectcell\AddLabel}r<{\endcollectcell}} 
\newcommand\ext[1]{\scalebox{.85}{$\bigwedge^{\!#1}$}}
\newcommand\sym[1]{\scalebox{.85}{$\bigodot^{\!#1}$}}
\usepackage{tikz,pgfplots}
\pgfplotsset{compat=1.18}
\usepgfplotslibrary{fillbetween}
\usetikzlibrary{calc,arrows,arrows.meta,cd,shapes.geometric,positioning,through,decorations.markings,decorations.text}
\tikzset{>=latex}
\usepackage[unicode,pagebackref=true]{hyperref}
\hypersetup{%
  pdftitle   = {The spectrum of the bosonic ambitwistor string revisited},
  pdfkeywords = {BMS, BRST, cohomology, ambitwistor, string},
  pdfauthor  = {José Figueroa-O'Farrill, Girish Vishwa},
  pdfcreator = {\LaTeX\ with package \flqq hyperref\frqq},
  linkcolor=NavyBlue,
  citecolor=ForestGreen,
  urlcolor=OrangeRed,
  anchorcolor=OrangeRed,
  colorlinks=true}
\renewcommand*{\backref}[1]{}
\renewcommand*{\backrefalt}[4]{%
  \ifcase #1%
  \or [Page~#2.]%
  \else [Pages~#2.]%
  \fi%
}
\theoremstyle{plain}
\newtheorem{lemma}{Lemma}
\newtheorem{proposition}[lemma]{Proposition}

\newtheorem{corollary}[lemma]{Corollary}
\newtheorem{conjecture}[lemma]{Conjecture}

\theoremstyle{definition}

\newtheorem*{remark}{Remark}
\newcommand{\g}{\mathfrak{g}}

\newcommand{\h}{\mathfrak{h}}
\newcommand{\fa}{\mathfrak{a}}
\newcommand{\fk}{\mathfrak{k}}

\renewcommand{\d}{\partial}

\newcommand{\so}{\mathfrak{so}}

\newcommand{\be}{\boldsymbol{e}}

\newcommand{\w}{\boldsymbol{w}}

\newcommand{\p}{\boldsymbol{p}}

\newcommand{\sB}{\mathsf{B}}
\newcommand{\sZ}{\mathsf{Z}}
\newcommand{\sC}{\mathsf{C}}
\newcommand{\sH}{\mathsf{H}}
\newcommand{\Crel}{\mathsf{C}_{\mathrm{rel}}} 
\newcommand{\Hrel}{\mathsf{H}_{\mathrm{rel}}}
\newcommand{\Brel}{\mathsf{B}_{\mathrm{rel}}}
\newcommand{\Zrel}{\mathsf{Z}_{\mathrm{rel}}}
\newcommand{\phisup}[1]{\phi^{(#1)}}
\newcommand{\Asup}[1]{A^{(#1)}}
\newcommand{\Tsup}[1]{T^{(#1)}}
\newcommand{\Ssup}[1]{S^{(#1)}}

\newcommand{\eO}{\mathcal{O}}

\newcommand{\Stab}{\operatorname{Stab}}
\newcommand{\stab}{\mathfrak{stab}}

\newcommand{\reg}{\operatorname{reg}}
\newcommand{\Tgh}{T^{\text{gh}}}
\newcommand{\Mgh}{M^{\text{gh}}}
\newcommand{\Ltot}{L^{\text{tot}}}
\newcommand{\Ttot}{T^{\text{tot}}}
\newcommand{\Mtot}{M^{\text{tot}}}

\newcommand{\vac}{\ket{p}}

\newcommand{\id}{\operatorname{id}}
\newcommand{\tr}{\operatorname{tr}}

\newcommand{\VV}{\mathbb{V}}

\newcommand{\RR}{\mathbb{R}}

\newcommand{\ZZ}{\mathbb{Z}}

\newcommand{\CC}{\mathbb{C}}

\newcommand{\ISO}{\operatorname{ISO}}

\newcommand{\SO}{\operatorname{SO}}

\newcommand{\End}{\operatorname{End}}
\newcommand{\Hom}{\operatorname{Hom}}

\newcommand{\im}{\operatorname{im}}

\definecolor{dkgr}{rgb}{0,0.6,0}
\definecolor{gris}{rgb}{0.5,0.5,0.5}

\allowdisplaybreaks[0]
\numberwithin{equation}{section}
\AtAppendix{\counterwithin{lemma}{section}}
\begin{document}

\title{The spectrum of the bosonic ambitwistor string revisited}
\author[Figueroa-O'Farrill]{José M Figueroa-O'Farrill}
\author[Vishwa]{Girish S Vishwa}
\address{Maxwell Institute and School of Mathematics, The University
  of Edinburgh, James Clerk Maxwell Building, Peter Guthrie Tait Road,
  Edinburgh EH9 3FD, Scotland, United Kingdom}
\email[JMF]{\href{mailto:j.m.figueroa@ed.ac.uk}{j.m.figueroa@ed.ac.uk}, ORCID: \href{https://orcid.org/0000-0002-9308-9360}{0000-0002-9308-9360}}
\email[GSV]{\href{mailto:}{G.S.Vishwa@sms.ed.ac.uk}, ORCID: \href{https://orcid.org/0000-0001-5867-7207}{0000-0001-5867-7207}}
\begin{abstract}
We revisit the calculation of the spectrum of the bosonic ambitwistor string, understood as the BRST cohomology or, equivalently, as the semi-infinite cohomology of the $\mathrm{BMS}_3$ Lie algebra relative to the centre with values in a particular module. We work in momentum space, which allows us to work algebraically and interpret the BRST cohomology as inducing representations of the Poincar\'e group. In agreement with the existing literature, we find that all the cohomology resides in the massless sector, but a careful representation-theoretic analysis of the spectrum reveals, in addition to the usual massless sector of the closed bosonic string (dilaton, metric and Kalb--Ramond field), also a massless vector. We devote a large part of the paper to describing the cohomology at a massless momentum $p$ as a module over the stabiliser $H$ of $p$ in the Lorentz group, a task which is made difficult due to $H$ not acting reducibly when $p\neq 0$. This allows us to conclude that the spectrum is not unitary, forbidding the interpretation of the extra massless vector as a Maxwell field.
\end{abstract}
\maketitle
\tableofcontents

\section{Introduction}
\label{sec:introduction}
Ambitwistor strings \cite{MasonSkinner, Adamo:2013tsa,
  Berkovits:2013xba} are a class of string theories introduced as a
generalisation of Witten's twistor string \cite{Witten:2003nn} (see
also an alternative formulation by Berkovits \cite{Berkovits:2004hg})
that brought unprecedented new insights into the study of scattering
amplitudes. Indeed, one of the primary successes of ambitwistor
strings which (rightfully) garnered them immense interest is how they
provide a framework in which the celebrated \emph{Cachazo--He--Yuan
  (CHY) formulae} \cite{Cachazo:2013gna, Cachazo:2013hca,
  Cachazo:2013iaa, Cachazo:2013iea} can be naturally extended to
one-loop \cite{Adamo:2013tsa, Adamo:2015hoa, Geyer:2015bja,
  Geyer:2015jch, Ohmori:2015sha} and curved backgrounds
\cite{Adamo:2014wea, Chandia:2015sfa, Chandia:2015xfa, Adamo:2017sze, Adamo:2018ege, Eberhardt:2020ewh, Roehrig:2020kck, Gomez:2021qfd, Adamo:2025bfr}.
Today, there is a huge body of literature\footnote{See the review
  article \cite{Geyer:2022cey} and references therein.} surrounding
the computation of scattering amplitudes that has been birthed by
ambitwistor strings.

Ambitwistor strings have also seen success under a different guise.
More specifically, Siegel introduced the \emph{left-handed} or
\emph{chiral} string in \cite{Siegel:2015axg} constructed via a
singular gauge choice\footnote{This gauge choice is known as
  \emph{Hohm--Siegel--Zwiebach (HSZ)} gauge, named after the authors
  of \cite{Hohm:2013jaa}, who put forth a similar model with manifest
  $T$-duality in the context of doubled conformal field theory.} on
the ordinary bosonic string worldsheet, which was shown to reproduce
the bosonic ambitwistor string action and amplitudes under the
tensionless limit. This string is now known as the \emph{twisted
  string}, as it comes from an alternative choice of vacuum for the
original closed bosonic string, where the anti-holomorphic oscillator
modes are ``flipped'' with respect to the holomorphic ones, thereby
leading to a twist that characterises the two sectors (see
\cite{Lee:2017utr} for more details). Once again, this led to more
work on scattering amplitudes (see for instance \cite{Siegel:2015axg,
  Huang:2016bdd, Leite:2016fno, Casali:2017mss}).

Concomitantly, the study of tensionless strings was gaining interest
(see, e.g., \cite{Bagchi:2013bga, Bagchi:2015nca, Bagchi:2016yyf}),
and in \cite{Bagchi:2020fpr}, it became evident that when quantised,
the closed bosonic tensionless string admitted three distinct spectra
coming from three\footnote{It was initially thought
  that there were only two different vacua, namely the ones put forth
  in \cite{Gamboa:1989px}.} different vacua. One of these three vacua
is indeed the ambitwistor vacuum, and the fact that this is not the
tensionless limit of ordinary closed string theory but rather the
twisted string was also reiterated (see also \cite{Casali:2016atr}
where this was first made precise). More recently, it was shown in
\cite{Adamo:2021zpw, Adamo:2022wjo} that the worldsheet operator
product expansions (OPEs) of vertex operators of four-dimensional
ambitwistor string theory generate the coefficients of celestial OPEs.
This demonstrates that ambitwistor strings provide a dynamical
principle for computing celestial OPEs, thereby making them central to
the celestial approach to flat space holography.

Setting their credentials aside, ambitwistor strings are, at the end
of the day, string theories which admit (at least in the RNS
formalism) quantum mechanical, chiral worldsheet CFT descriptions.
Thus, their spectra are also pivotal aspects of their study. Hence, in
this paper, we present a complete, self-contained computation of the
spectrum of the bosonic ambitwistor string in Minkowski spacetime. We
exploit the fact that the ambitwistor string worldsheet is a chiral
field theory with BMS$_3$ symmetry, making the computation of its
spectrum amenable to homological and vertex operator algebraic
techniques. We follow the conventions of
\cite{Figueroa-OFarrill:2024wgs}, where we show that this ambitwistor
string affords a realisation of the BMS$_3$ Lie algebra and hence is a possible module on which to compute its semi-infinite cohomology relative to the centre; that is, its BRST cohomology.

As already observed by Berkovits and Lize in
\cite{Berkovits:2018jvm}, the spectrum is non-unitary.  In our work,
we rederive this result by a careful representation-theoretic analysis of the spectrum which shows that it does not induce a unitary representation of the Poincaré group.  If we identify the physical spectrum as the BRST cohomology at ghost number 2, as is usually the case, the maximal unitary Poincaré subrepresentation corresponds to a Kalb--Ramond field, metric and dilaton (i.e., the massless sector of the relativistic closed bosonic string). Often, this is considered to be the bosonic ambitwistor string spectrum \cite{MasonSkinner, Huang:2016bdd, Casali:2016atr, Lee:2017utr} (see also \cite{Gamboa:1989zc, Bagchi:2020fpr} for this result obtained from the perspective of tensionless string theory), although it was shown in \cite{Berkovits:2018jvm} that this is not the case (see also \cite{LipinskiJusinskas:2019cej, Carabine:2023yxv} for the chiral/twisted string approach to this). 
Our work builds on this and shows representation-theoretically that the full ambitwistor spectrum is strictly larger than this: containing as well
what at first sight appears as an additional massless vector --- an
interpretation which can be discarded after a more careful analysis
which shows that the spectrum is a non-unitary representation of the
Poincaré group.

Our algebraic, representation-theoretic approach to the study of the
ambitwistor spectrum allows us to go even further than just the
computation of BRST cohomology at ghost number 2 which is normally
done in the literature \cite{MasonSkinner, Berkovits:2018jvm}. Thus,
we are able to show that the non-unitarity of the BRST cohomology at
ghost number 2 also manifests itself at ghost number 4. Consider for a
moment the case of ordinary bosonic string theory. Since its BRST
cohomology induces only unitary Poincaré representations, there is a
representation-theoretic isomorphism between cohomology at ghost
numbers 1 and 2 (which leads to an isomorphism between the Poincaré
representations which they induce) -- a consequence of Poincaré
duality. In the ambitwistor string, however, the non-unitarity of BRST
cohomology at ghost number 2 leads to a more general statement, namely
that the BRST cohomology at ghost number 2 and 4 induce mutually
\emph{dual} Poincaré representations. Furthermore, the maximal unitary
Poincaré subrepresentation induced at ghost number 4 corresponds to a
Kalb--Ramond field and a photon. It initially seems peculiar that the
BRST cohomology of a closed string contains what is normally
considered to be an open string state, but we argue that because we
are simply \emph{not} in the setting of ordinary bosonic string
theory, we cannot expect similar algebraic behaviour of the resulting
BRST cohomology. To reiterate, we regard the appearance of this
photon-like state at ghost number 4 in the ambitwistor BRST cohomology
as a manifestation of the non-unitarity of the cohomology at ghost
number 2, which we choose to interpret as the physical spectrum of the
ambitwistor string.

Our techniques are readily applicable not just to other ambitwistor
strings in the RNS formalism such as the heterotic
\cite{Adamo:2018hzd} and Type II ambitwistor \cite{Adamo:2014wea}
strings, but also to other non-Lorentzian strings -- the spectra of a
bosonic string reminiscent of the Gomis--Ooguri string that was
constructed as a gauged Wess--Zumino--Witten model
\cite{Figueroa-OFarrill:2025nmo} and the \emph{carrollian string}
\cite{Gomis:2023eav, Blair:2023noj, Chen:2025gaz,
  Figueroa-OFarrill:2025njv}, where both the worldsheet and target
space are carrollian, were both obtained via similar homological
techniques. The latter, in particular, is closely related to the
ambitwistor string as it too can be described by a worldsheet with
chiral BMS$_3$ symmetry. From this perspective, the computation of the
spectra of the carrollian and bosonic ambitwistor strings are both
computations of semi-infinite cohomology of the BMS$_3$ algebra
relative to the centre, but with values in different BMS$_3$ modules
specified by their respective matter sector contents. We discuss
further in Section \ref{sec:conclusions-outlook} and also refer the
reader to \cite{Figueroa-OFarrill:2025njv} for more details.

\subsection{Organisation and summary of main results}
\label{sec:organ-summ-main}

In order to help the reader navigate what is in parts a rather
technical paper, we will now give a detailed summary of the
organisation of this paper and its main results.

In Section~\ref{sec:preliminaries} we define the problem we set out to
solve and review some representation-theoretic results we will make
use of in the sequel.  In Section~\ref{sec:bosonic-ambi} we define the
bosonic ambitwistor string as a particular module of the BMS$_3$ Lie
algebra, which we express in the language of two-dimensional conformal
field theory and operator product expansions.  The BMS$_3$ Lie algebra
is generated by two fields $T(z)$ and $M(z)$, where $T(z)$ is the
field corresponding to a Virasoro element and $M(z)$ is an abelian
ideal with conformal weight $2$.  The matter sector is given by $26$
(for criticality) linear systems $(X^\mu, \Pi_\mu)$ of conformal
weights $(0,1)$.  They carry a representation of BMS$_3$ which was
shown in \cite{Figueroa-OFarrill:2024wgs} to allow for the
construction of a BRST complex as the semi-infinite complex of the
BMS$_3$ Lie algebra relative to the centre, with the square-zero BRST
operator being identified with the semi-infinite differential.  The
BRST complex admits a representation of BMS$_3$ with vanishing
central charges.  Moreover, the zero mode of the Virasoro element acts
diagonally in the complex and trivially on cohomology, so that the
BRST complex is quasi-isomorphic to the subcomplex consisting of
states with zero total conformal weight.  We work in momentum space,
so that the BRST differential acts algebraically (i.e., not involving
derivatives) and hence we have for each momentum $p$ a differential
complex $\left(\sC^\bullet(p), d\right)$ whose cohomology is the main
object of this paper.  We choose to break the computation of the
cohomology into two, by first computing the cohomology relative to the
one-dimensional subalgebra spanned by the Virasoro zero mode. The
resulting relative complex $\left(\Crel^\bullet(p),d\right)$ is
finite, containing states with ghost numbers $0$ to $5$, whereas the
absolute complex $\left( \sC^\bullet(p), d \right)$ contains states
with ghost numbers $0$ to $6$.  The zero mode of the weight-$2$ field
in the BMS$_3$ algebra in the BRST complex also acts trivially in
cohomology.  Its Jordan--Chevalley decomposition allows us to prove, 
without any explicit cohomology calculations, 
that the (relative) cohomology is trivial unless $p^2=0$. We show that
the relative and absolute complexes are related by a split short exact
sequence of complexes~\eqref{eq:ses} giving rise to a long exact
sequence in cohomology~\eqref{eq:les}, which we will exploit later.

In Section~\ref{sec:what-is-the-cohomology} we explain the different
ways to describe the BRST cohomology at momentum $p$: as a vector
space, as a representation of the stabiliser subgroup $\Stab(p)$ of
$p$ in the Lorentz group or as a representation of a maximal compact
subgroup $K \subset \Stab(p)$. As in the relativistic bosonic string,
we interpret the cohomology at momentum $p$ as inducing a
representation of the Poincaré group.  Wigner \cite{MR1503456}
famously showed that all unitary irreducible representations
(UIRs) of the four-dimensional Poincaré group are induced from UIRs of
$\Stab(p)$, a result vastly generalised by Mackey (see, e.g.,
\cite{MR0396826}) for (regular) semidirect products.  Because of this,
it is the description of the BRST cohomology $\sH^\bullet(p)$ as a
representation of $\Stab(p)$ which is our final aim. Nevertheless much
can be learned by viewing it as a vector space or as a representation
of $K$.

In Section~\ref{sec:stabp-reps} we study the group $\Stab(p)$ and some
of its representations.  We mostly work at the Lie algebraic level and
hence starting in that section we refer to representations as modules
(over the universal enveloping algebra of the relevant Lie algebra).
This section introduces some notation which is used throughout the
paper.  We take $(V,\eta)$ to denote the $26$-dimensional real vector
representation of the Lorentz group, with $\eta$ the invariant
lorentzian inner product.  The momentum $p$ is an element of the dual
space $V^*$, which may be identified with $V$ using the musical
isomorphisms \eqref{eq:musical} induced by $\eta$.  Because we work in
momentum space we are forced to complexify $V$ to
$\VV = \CC \otimes_\RR V$ and extend $\eta$ complex-bilinearly.  A
numerical calculation summarised in Table~\ref{tab:numerical}, shows
that the cohomology is zero unless $p^2 = 0$, which allows us to focus
on two cases: a nonzero massless momentum $p$ and the case of zero
momentum.  The latter has the whole Lorentz group as stabiliser and
being semisimple it acts fully reducibly.  The former, however, has 
stabiliser $H := \Stab(p) \cong \ISO(24)$, which does not act fully
reducibly.  Therefore,  all the representation-theoretic complications
arise for the case of a nonzero massless momentum.

As an $H$-module, the vector
representation $V$ is indecomposable. It has two proper submodules
$\ell_p \subset p^\circ \subset V$, where $\ell_p$ is the real span of
the dual vector $p^\sharp \in V$ to $p \in V^*$, and $p^\circ$ can be
viewed either as the annihilator of $p \in V^*$ or as the
perpendicular subspace to $p^\sharp$ relative to $\eta$. The quotient
module $p^\circ/\ell_p$ is denoted $V^\top$ (for ``transverse'') and
will play an important role in the paper. As an $H$-module, $V^\top$
is orthogonal relative to the $H$-invariant euclidean inner product
induced by $\eta$. Upon complexification, we have that $\VV$ admits
two proper submodules $L_p \subset p^\perp \subset \VV$, where $L_p$
and $p^\perp$ denote the complexifications of $\ell_p$ and $p^\circ$,
respectively. The quotient module $p^\perp/L_p$ is denoted $\VV^\top$
and is the complexification of $V^\top$. The euclidean inner product
on $V^\top$ induces a hermitian inner product on $\VV^\top$, making it
into a unitary $H$-module.

In Section~\ref{sec:lie-algebra-stabp} we determine the structure of
the Lie algebra $\h$ of $H$, which can be interpreted as the euclidean
Lie algebra of $V^\top$ relative to the induced euclidean inner
product.  We use throughout the natural isomorphism
$\ext{2}V \cong \so(V)$, which sends
$v \wedge w \mapsto v \curlywedge w$, where the action of
$v \curlywedge w$ on $V$ is defined in \eqref{eq:curlywedge}.  Using
this notation, the abelian ideal in $\h$ is denoted
$\fa:= \ell_p \curlywedge p^\circ$, whereas $\h$ itself is isomorphic
to $\ext{2}p^\circ$ and $\h/\fa \cong \so(V^\top)$.  We use throughout
the paper the well-known fact that in a finite-dimensional unitary
module of $\h$, the abelian ideal $\fa$ must act trivially, so that
the representation descends to a representation of $\so(V^\top)$.  We
will see later that this is \emph{not} the case for the BRST
cohomology.

The rest of Section~\ref{sec:stabp-reps} is devoted to studying the
$\h$-modules $V\otimes V$ and its symmetric $\sym{2}V$ and
skew-symmetric $\ext{2}V$ submodules, since as we will see these are
the other modules arising in the BRST cohomology.  In
Section~\ref{sec:some-submodules-v} we merely describe the submodules
of $V \otimes V$.  In Section~\ref{sec:submodules-skew-2-v} we study
$\ext{2}V$, identify its submodules and prove some isomorphisms
involving quotients of these submodules.  In
Section~\ref{sec:submodules-sym-2-v} we do the same for the more
involved $\sym{2}V$.  Table~\ref{tab:dimensions} summarises some of
the $H$-modules discussed in the section and lists its dimensions.
All the results in
Sections~\ref{sec:some-submodules-v},~\ref{sec:submodules-skew-2-v}
and \ref{sec:submodules-sym-2-v} have analogues upon
complexification, which although we do not list separately, can
be simply obtained by replacing $(\RR, V, \ell_p, p^\circ,V^\top)$ by
$(\CC,\VV,L_p, p^\perp,\VV^\top)$ everywhere.  It bears highlighting
that whereas we complexify the modules, the Lie algebra itself remains
real.

With these preliminaries behind us, we turn in
Section~\ref{sec:relative-cohomology} to the calculation of the
relative BRST cohomology.  The action of the differential on the
fields corresponding to the states in the relative subcomplex is
calculated from the first order pole in the operator product expansion
with the BRST current \eqref{eq:BRST current}.  The results of those
calculations are recorded in Appendix~\ref{app:calculation of d on
  Crel}. We first show that the cohomology is a $\Stab(p)$-module and
we present the results of numerical calculations which determine the
cohomology at momentum $p$ as a complex vector space. Complex vector
spaces have only one invariant: their dimension and that is what
Table~\ref{tab:numerical} lists. For a nonzero massless momentum, the
relative cohomology resides in ghost numbers $2$ and $3$, whose
calculation is presented in detail in
Appendices~\ref{app:H2rel-detailed-calc} and \ref{app:hrel3-details},
respectively. We give cocycle representatives for the relative
cohomology in terms of covariant fields satisfying $H$-invariant
conditions. At ghost number $2$, these cocycle representatives are
vertex operators \eqref{eq:Zrel2-vertex-operators} subject to the
cocycle conditions~\eqref{eq:Zrel2-conditions} and modulo the
coboundaries~\eqref{eq:Brel2-conditions}; whereas at ghost number $3$,
the vertex operators are \eqref{eq:Zrel3-vertex-operator-1} and
\eqref{eq:Zrel3-vertex-operator-2} subject to the cocycle conditions
\eqref{eq:3-relative-cocycles} and modulo the
coboundaries~\eqref{eq:3-relative-coboundaries}. There is another set
of cocycle representatives at ghost number $3$, which is perhaps more
natural: see equation~\eqref{eq:Zrel3-alt-symm-vertex-op}. The final
results in Section~\ref{sec:relative-cohomology} are summarised in
Section~\ref{sec:hrel-summary}. Proposition~\ref{prop:Hrel} exhibits
the relative BRST cohomology at zero momentum as a Lorentz-module and
that at a nonzero massless momentum as a $K$-module. In this latter
case, we already see that both at ghost number $2$ and ghost number
$3$ we have an additional vector ($\VV^\top$) which seems to have been
ignored in the literature until now.

In Section~\ref{sec:brst-cohomology} we exploit the long exact
sequence relating the absolute and relative BRST cohomologies to
determine the absolute cohomology.  This gives immediately that the
absolute BRST cohomology is zero unless $p^2 = 0$.  As shown in
Section~\ref{sec:nonzero-massless}, for nonzero massless momenta, the
long exact sequence is enough to determine (mostly) the absolute
cohomology.  The absolute cohomology resides at ghost numbers $2,3,4$
and we find that, as $H$-modules, $\sH^2(p) \cong \Hrel^2(p)$,
$\sH^4(p) \cong \Hrel^3(p)$ and $\sH^3(p)$ is an extension\footnote{We
  have not yet determined whether this extension is trivial, but in
  Appendix~\ref{app:splitting} we compute the relevant
  Chevalley--Eilenberg cohomology group where the obstruction lives
  and show that it is not zero, requiring the calculation of the
  actual extension class, a task we will leave for the future.} of
$\Hrel^2(p)$ by $\Hrel^3(p)$.  At zero momentum, as shown in
Section~\ref{sec:zero-momentum}, the long exact sequence is not enough
and we must calculate explicitly the connecting homomorphism
$\Hrel^{n-1}(p) \to \Hrel^{n+1}(p)$.  Assuming Poincaré duality of the
absolute cohomology, it is enough to compute $\sH^1(0)$ and $\sH^2(0)$
to determine the rest of the cohomology.  The conjectural result is
contained in Conjecture~\ref{sec:conj:zero-momentum}.  Regardless
whether the conjecture holds, we do find an enhancement of cohomology
at zero momentum, just as with the relativistic bosonic string.

In Section~\ref{sec:phys-interpr} we concentrate on understanding the
cohomology for a nonzero massless momentum as a module over the
stabiliser $H$.  In Section~\ref{sec:rep-theory-hrel2} we determine
the $H$-module structure of $\Hrel^2(p) \cong \sH^2(p)$ and in
particular we show in Proposition~\ref{prop:Hrel2_not_unitary} that it
is not a unitary $H$-module and we show that its maximal unitary
submodule is isomorphic to $\ext{2}\VV^\top \oplus \sym{2}\VV^\top$,
which are precisely the representations that induce the unitary Poincaré modules describing a Kalb--Ramond field and a ``graviton + dilaton''. We simply refer to these as \emph{inducing representations.}  In Section~\ref{sec:rep-theory-hrel3} we determine the $H$-module
structure of $\Hrel^3(p) \cong \sH^4(p)$.  We show in
Proposition~\ref{prop:Hrel3_not_unitary} that it again fails to be
unitary and we show that its maximal unitary submodule is now
isomorphic to $\ext{2}\VV^\top \oplus \VV^\top$, which are the
inducing representations for unitary Poincaré modules describing a
Kalb--Ramond field and a ``photon''.  In
Proposition~\ref{prop:poincare-duality} we prove that $\Hrel^2(p)$ and
$\Hrel^3(p)$ are dual $H$-modules.  In
Section~\ref{sec:phys-interpr-spectr} we discuss the physical
interpretation of our calculations and end with a conjecture
concerning $\sH^3(p)$.

The paper ends with Section~\ref{sec:conclusions-outlook}, where we
offer some conclusions and point to future work.

\section{Preliminaries} \label{sec:preliminaries}

\subsection{Bosonic ambitwistor string worldsheet}
\label{sec:bosonic-ambi}

We review the worldsheet description of the bosonic ambitwistor
string. In particular, we follow the conventions of
\cite{Figueroa-OFarrill:2024wgs} in which the worldsheet of the
ambitwistor string is shown to afford a realisation of the BMS$_3$ Lie
algebra, which means that its matter sector is a possible module on
which to compute its semi-infinite cohomology relative to the centre;
that is, its BRST cohomology.

The BMS$_3$ algebra $\g$ is generated by fields $T(z)$ and $M(z)$
subject to the following operator product expansions (OPEs)
\begin{equation}
  \begin{split}
    T(z)T(w) &= \frac{1}{2}\frac{c_L\mathbbm{1}}{(z-w)^4} + \frac{2T(w)}{(z-w)^2} + \frac{\partial T(w)}{z-w} + \reg.\\
    T(z)M(w) &= \frac{1}{2}\frac{c_M \mathbbm{1}}{(z-w)^4} + \frac{2 M(w)}{(z-w)^2} + \frac{\partial M(w)}{z-w} + \reg.\\
    M(z)M(w) &= \reg.
  \end{split}
\end{equation}
The matter sector of the ambitwistor string is given by a free-field
realisation of this algebra with $c_L = 2D$ and $c_M=0$ in terms of
$D$ weight-$(0,1)$ linear systems $(X^\mu, \Pi_\mu)$, where
$\mu\in\{0,1,...,D-1\}$ with operator product expansions
\begin{equation} \label{eq:X-Pi defining OPEs}
  \Pi_\mu(z) X^\nu (w) = \frac{\delta^\nu_\mu \mathbbm{1}}{z-w} + \reg.
\end{equation}
The fields $T(z)$ and $M(z)$ are given by
\begin{equation}
  \label{eq:T-and-M}
  \begin{split}
    T &= (\partial X^\mu \Pi_\mu)\\
    M &= - \tfrac12 \eta^{\mu\nu}(\Pi_\mu\Pi_\nu),
  \end{split}
\end{equation}
where $\eta^{\mu\nu}$ is the inverse of the Minkowski metric.  One can
construct a BRST complex provided that $D=26$, compatible with the
facts that its spectrum should emerge as the NR contraction of two
copies of the Virasoro algebra \cite{CampoleoniGonzalezOblakRiegler,
  Bagchi:2020fpr} and that its critical dimension is 26
\cite{MasonSkinner, Casali:2016atr, Bagchi:2020fpr}.  We shall take
$D=26$ from now on.

We may rewrite OPEs as follows.\footnote{We refer the reader to
  \cite[Section 2]{Figueroa-OFarrill:2024wgs} for more details
  regarding this notation.}  Let $\mathcal{V}$ denote the underlying
vector space of the vertex operator algebra.  For any state
$A \in \mathcal{V}$, we denote its corresponding field as
$A(z) \in \End \mathcal{V}[[z,z^{-1}]]$.  Then for any two states
$A,B\in \mathcal{V}$, we may write their OPE as
\begin{equation}
  A(z)B(w) = \sum_{n\in\ZZ} \frac{[A,B]_n(w)}{(z-w)^n},
\end{equation}
where $[A,B]_n(w)$ is the field in the $n^{\text{th}}$ order pole of the
OPE corresponding to the state $[A,B]_n \in \mathcal{V}$. In
particular, we denote the normal-ordered product of two fields as
$(AB)(z):= [A,B]_0(z)$.

We introduce two weight-$(2,-1)$ fermionic linear systems $(b,c)$ and
$(B,C)$ with standard operator product expansions
\begin{equation}
  b(z) c(w) = \frac{\mathbbm{1}}{z-w} + \reg. \qquad\text{and}\qquad
  B(z) C(w) = \frac{\mathbbm{1}}{z-w} + \reg.
\end{equation}
and define the BRST current
\begin{equation} \label{eq:BRST current}
  j_{\text{BRST}} = (c T) + (C M) + \tfrac12 (c \Tgh) + \tfrac12 (C \Mgh),
\end{equation}
where~\cite{Figueroa-OFarrill:2024wgs}
\begin{equation}
  \begin{split}
    \Tgh &= -2(b\partial c) - (\partial b c) - 2 (B\partial C) - (\partial B C)\\
    \Mgh &= -2 (B\partial c) - (\partial B c).
  \end{split}
\end{equation}
The BRST differential $d$ is defined as the zero-mode of
$j_{\text{BRST}}$ or, equivalently, in terms of fields as the
first-order pole with $j_{\text{BRST}}$.

Let $\Ttot = db$ and $\Mtot = dB$.  They are given explicitly by
$\Ttot = T + \Tgh$ and $\Mtot = M + \Mgh$ and they obey the centreless
BMS algebra:
\begin{equation}
  \begin{split}
    \Ttot(z)\Ttot(w) &= \frac{2\Ttot(w)}{(z-w)^2} + \frac{\partial \Ttot(w)}{z-w} + \reg.\\
    \Ttot(z)\Mtot(w) &= \frac{2\Mtot(w)}{(z-w)^2} + \frac{\partial \Mtot(w)}{z-w} + \reg.\\
    \Mtot(z)\Mtot(w) &= \reg.
  \end{split}
\end{equation}

Let us introduce modes via
\begin{multline}
  T(z) =\sum_{n \in \ZZ} L_n z^{-n-2},\quad M(z) = \sum_{n \in \ZZ}
  M_n z^{-n-2},\quad X^\mu(z) = \sum_{n\in \ZZ} x^\mu_n
  z^{-n},\quad \Pi^\mu(z) =\sum_{n\in \ZZ} \pi^\mu_n
  z^{-n-1},\\
  b(z) = \sum_{n\in \ZZ} b_n z^{-n-2}, \quad c(z) = \sum_{n\in \ZZ}
  c_n z^{-n+1}, \quad B(z) = \sum_{n\in \ZZ} B_n z^{-n-2}
  \quad\text{and}\quad C(z) = \sum_{n\in \ZZ} C_n z^{-n+1}.
\end{multline}

We define a vacuum vector $\ket{0}$ by
\begin{equation}
  \begin{split}
    c_n \ket{0} = C_n \ket{0} &= 0 \quad\text{for $n\geq 2$}\\
    x^\mu_n \ket{0} &= 0 \quad\text{for $n\geq 1$}\\
    \pi^\mu_n \ket{0} &= 0 \quad\text{for $n\geq 0$}\\
    b_n \ket{0} = B_n \ket{0} &= 0 \quad\text{for $n\geq -1$}.
  \end{split}
\end{equation}
Notice that $\Ltot_0 \ket{0} = \Mtot_0 \ket{0} = 0$.  Moreover, since
$\Ltot_0$ acts reducibly, we may restrict to the subcomplex $\ker
\Ltot_0$, since there is no cohomology at any other eigenspace of
$\Ltot_0$.  The BRST complex can be taken to be $\sC^\bullet := \ker
\Ltot_0$, graded as usual by ghost number: where $B,b$ have ghost
number $-1$; $c,C$ have ghost number $1$ and $X^\mu, \Pi_\mu$ have
ghost number $0$.  This argument does not work for $\Mtot_0$ since it
does not act reducibly, but nevertheless we will deduce an important
consequence from the fact that $\Mtot_0$ acts trivially in cohomology.

Since in the BRST current $X^\mu$ only appears via $\d X^\mu$, its
zero mode does not occur.  This suggests working in momentum space,
resulting in a purely algebraic differential; that is, involving no
derivatives.  To that end we define the vertex operator $W_p := \exp(i
p_\mu X^\mu)$ and hence
\begin{equation}
  \vac := \lim_{z\to 0} W_p(z) \ket{0} = \exp(i p_\mu x^\mu_0) \ket{0},
\end{equation}
so that now
\begin{equation}
  \label{eq:momentum-eigenstates}
  \pi^\mu_0 \vac = i p^\mu \vac.
\end{equation}
The BRST subcomplex $\sC^\bullet(p)$ at momentum $p$ is the kernel of
$\Ltot_0$ acting on the states made out of $\vac$ using the creation
operators
\begin{equation}
  x^\mu_{n\leq 0}, \quad   \pi^\mu_{n\leq 0}, \quad b_{n\leq -2},
  \quad B_{n\leq -2}, \quad c_{n\leq 1} \quad\text{and}\quad C_{n\leq
    1},
\end{equation}
where $\pi^\mu_0$ acts diagonally.  Since $[d,\pi^\mu_0]=0$,
it follows that $\sC^\bullet(p)$ is indeed
a subcomplex.  Because of the $i$ in
equation~\eqref{eq:momentum-eigenstates}, $\sC^\bullet(p)$ is a
complex vector space and since complex conjugation sends $W_p \mapsto
W_{-p}$, it relates $\sC^\bullet(p)$ and $\sC^\bullet(-p)$.

It is convenient to break the computation into two: we first compute
the cohomology of the relative subcomplex $\Crel^\bullet(p) :=
\sC^\bullet(p) \cap \ker b_0$ from which we then compute the BRST
cohomology by a mixture of homological algebra and some explicit
computation.

All creation operators have negative conformal weight except for $c_1,
C_1, c_0, C_0, x^\mu_0, \pi^\mu_0$.  Now, $\pi^\mu_0$ acts diagonally,
$x^\mu_0$ does not appear in our calculations and $c_0$ does not
appear in the relative subcomplex, leaving only $c_1, C_1, C_0$.  The
highest ghost number which can appear in $\Crel^\bullet(p)$ is
therefore $5$, corresponding to $C_0 c_1 C_1 c_{-1} C_{-1}\vac$,
whereas the highest ghost number which can appear in $\sC^\bullet(p)$
is $6$, corresponding to $c_0C_0 c_1 C_1 c_{-1} C_{-1}\vac$.  There
are no states of negative ghost number in either $\sC^\bullet(p)$ or
$\Crel^\bullet(p)$.

Therefore the relative subcomplex at momentum $p$ is given by
\begin{equation}
  \label{eq:Crel}
  \begin{tikzcd}
    \Crel^0(p) \arrow[r,"d"] & \Crel^1(p) \arrow[r,"d"] &  \Crel^2(p)
    \arrow[r,"d"] &  \Crel^3(p) \arrow[r,"d"] &  \Crel^4(p) \arrow[r,"d"] &
    \Crel^5(p) \arrow[r] & 0,
  \end{tikzcd}
\end{equation}
and hence for each $p$, this is a finite-dimensional differential
complex.

We can now use the fact that $\Mtot_0$ acts trivially in cohomology:
since $\Mtot_0 = [d,B_0]$, it follows that it sends cocycles to
coboundaries.  An explicit calculation shows that acting on $\Crel^\bullet(p)$
\begin{equation}
  \Mtot_0 = -\tfrac12 p^2 \id - \sum_{n\geq 1} \left( \eta^{\mu\nu}
    (\pi_\mu)_{-n} (\pi_\nu)_n - n \left( c_{-n} B_n + B_{-n} c_n \right)\right).
\end{equation}
Every endomorphism on a finite-dimensional vector space admits a
Jordan--Chevalley decomposition into a semisimple and a nilpotent
part.  It is clear from the above expression that the semisimple part
is $-\frac12 p^2 \id$ and the rest is the nilpotent part.  As proved
in \cite[§4.2]{MR323842}, for example, since $\Mtot_0$ sends
cocycles to coboundaries, the same holds separately for its semisimple
and nilpotent parts.  In particular $-\frac12 p^2 \id$ acts trivially
in cohomology.  This means that if $p^2 \neq 0$, the cohomology is
trivial.  In summary, we have proved the following:

\begin{proposition}\label{prop:on-shell}
  $\Hrel^\bullet(p) = 0$ unless $p^2 = 0$.
\end{proposition}

We remark that the same proof \emph{mutatis mutandis} serves to
establish the similar result for the carrollian string of
\cite{Figueroa-OFarrill:2025njv} and indeed any of the interpolating
BMS string theories discussed in \cite[§6]{Figueroa-OFarrill:2025njv}.

The relative and absolute complexes are related by a short exact
sequence of differential complexes.  Indeed, $\Crel^n(p) \subset
\sC^n(p)$ as the kernel of $b_0$, whereas the image of $b_0 \colon
\sC^n(p) \to \sC^{n-1}(p)$ lands in $\Crel^{n-1}(p)$, since $b_0^2 =
0$, giving
\begin{equation}
  \label{eq:ses}
  \begin{tikzcd}
    0 \arrow[r] & \Crel^n(p) \arrow[r] & \sC^n(p) \arrow[r,"b_0"] & \Crel^{n
      -1}(p) \arrow[r] \arrow[l, bend left=45, "c_0" pos=0.45] & 0,
  \end{tikzcd}
\end{equation}
where we have indicated by $c_0 \colon \Crel^{n-1}(p) \to \sC^n(p)$
the splitting of the sequence.  Indeed, for every $\Psi \in
\Crel^{n-1}(p)$, one has that $b_0 c_0 \Psi = [b_0,c_0] \Psi = \Psi$.
By standard homological algebra, the short exact sequence
\eqref{eq:ses} gives rise to a long exact sequence relating the
absolute and relative cohomologies
\begin{equation}
  \begin{tikzcd}
    \cdots \arrow[r] & \Hrel^{n-2}(p) \arrow[r] & \Hrel^n(p) \arrow[r] & \sH^n(p) \arrow[r] & \Hrel^{n-1}(p)
    \arrow[r] & \Hrel^{n+1}(p) \arrow[r] & \cdots
  \end{tikzcd}
\end{equation}
on which we elaborate in Section~\ref{sec:brst-cohomology}.

\subsection{How should we view the (relative) BRST cohomology? }
\label{sec:what-is-the-cohomology}

In this paper, we will view the BRST cohomology in at least three
different ways in increasing level of complexity. Firstly, from its
definition as the quotient of cocycles by coboundaries, it is a vector
space and hence determined by a unique invariant up to isomorphism,
namely, its dimension. This is easily determined by a numerical
calculation of the dimensions of the spaces of cocycles and using the
Rank Theorem to determine the dimensions of the spaces of coboundaries
and hence of the cohomology. These results are contained in
Table~\ref{tab:numerical}, from where we can already read off some
vanishing results. In particular, this already indicates that the 
ambitwistor string spectrum cannot just be the massless states of
the closed bosonic string; the dimension of the ambitwistor string 
BRST cohomology vector space is larger than that of the vector space of
closed string massless states.

To proceed further we need to add more structure to the cohomology.
It is clear from the expression~\eqref{eq:BRST current} for the
BRST current and the explicit expressions for $T$ and $M$ in
equation~\eqref{eq:T-and-M} that the BRST differential is Lorentz
invariant.  We may show this explicitly as follows.

The action of the Lie algebra of the Lorentz group on the BRST complex
can be described in terms of fields.  Let $\omega_{\mu\nu} = -
\omega_{\nu\mu}$ and let
\begin{equation}
  \Lambda(\omega) := \tfrac12 \omega_{\mu\nu} \left( X^\mu \Pi^\nu \right).
\end{equation}
This is a conformal field with conformal weight $1$ and ghost number
$0$, whose zero mode provides a representation of the Lorentz algebra
$\so(25,1)$.  Indeed, we have that
\begin{equation}
  \label{eq:lorentz}
  [\Lambda(\omega_1), \Lambda(\omega_2)]_1 = \Lambda([\omega_1,\omega_2]),
\end{equation}
where
\begin{equation}
  [\omega_1,\omega_2]_{\mu\nu} = (\omega_1)_{\mu\rho}
  \eta^{\rho\sigma} (\omega_2)_{\sigma\nu} - (\omega_2)_{\mu\rho}
  \eta^{\rho\sigma} (\omega_1)_{\sigma\nu},
\end{equation}
and hence it follows that the zero modes $Q(\omega):= [\Lambda(\omega),-]_1$ obey
\begin{equation}
  [Q(\omega_1),Q(\omega_2)] = Q([\omega_1,\omega_2])
\end{equation}
Under the action of the BRST differential,
\begin{equation}
  \label{eq:d-lambda-is-exact}
  d \Lambda(\omega) = [j_{\text{BRST}},\Lambda(\omega)]_1 = \partial\left(  c \Lambda(\omega) \right),
\end{equation}
so that the zero mode of $\Lambda(\omega)$ commutes with the BRST
differential:
\begin{align*}
    [d,Q(\omega)] &= [j_{\text{BRST}},[\Lambda(\omega), -]_1]_1 -  [\Lambda(\omega),[j_{\text{BRST}},-]_1]_1\\
                  &= [[j_{\text{BRST}}, \Lambda(\omega)]_1,-]_1 & \tag{by the Jacobi identity for $[-,-]_1$}\\
                  &= [\partial(c \Lambda(\omega)),-]_1 & \tag{by equation~\eqref{eq:d-lambda-is-exact}}\\
                  &= 0,
\end{align*}
since $[\partial A, -]_1 = 0$ for all $A$.

We will show in Section~\ref{sec:relative-subcomplex} that only (the
Lie algebra of) the stabiliser $\Stab(p)$ of the momentum $p$
preserves the complex $\sC^\bullet(p)$ at momentum $p$ and that indeed
it also preserves the relative complex $\Crel^\bullet(p)$. But the
Lorentz-equivariance of the BRST differential implies that the
(relative) cochains, cocycles and coboundaries at momentum $p$ and at
each ghost number are $\Stab(p)$-submodules and hence so is the
(relative) cohomology. However, based on what one would expect from
previous studies on the ambitwistor string and in anticipation of our
results, $\Stab(p)$ does not act reducibly on the complex when
$p^2=0,\ p\neq 0$. This will complicate the analysis of the cohomology
as a $\Stab(p)$-module. Nonetheless, this is necessary for our final
goal of deducing the Poincaré module structure of the cohomology, and
thus of the ambitwistor string spectrum, as this is what will
eventually allow us to provide any meaningful physical interpretations
of our results.

In between these two descriptions of the cohomology as a vector space
and as a $\Stab(p)$-module sits a third description where we focus on
how the cohomology decomposes as a representation of a maximal compact
subgroup of $\Stab(p)$, which acts reducibly. This gives additional
structure to the cohomology that is not present in the vector space
description, but is insufficient for the study of the cohomology as
representations of the Poincaré group, since these are induced from
representations of $\Stab(p)$.  Nevertheless, it does give some
intuition and is a useful middle ground to perform sanity checks on
our calculations. In particular, it lets us see in more detail what
the vector space treatment of BRST cohomology told us; the ambitwistor
string BRST cohomology cannot simply be the restriction of BRST
cohomology of the closed string to its massless sector.

To summarise, we will first describe the relative cohomology of the
ambitwistor string as a vector space, then as a module over a
maximal compact subgroup of $\Stab(p)$, and finally as a
$\Stab(p)$-module. We then compute the full BRST cohomology of the
ambitwistor string via homological algebraic techniques as a
$\Stab(p)$-module. By applying Mackey's method of inducing
representations to this, we obtain the description of the BRST
cohomology as a Poincaré module, which allows us to provide some
physical interpretation of our results.

\subsection{The $\Stab(p)$ subgroup of $\SO(25,1)$ and some of its modules}
\label{sec:stabp-reps}

In this section we will discuss the relevant representation theory of
the stabiliser subgroup $\Stab(p)$ in the Lorentz group of a momentum
$p$.  We will introduce the notation $V$ for the $26$-dimensional
representation of the Lorentz group $\SO(25,1)$.  It is a lorentzian
vector space with inner product $\eta$.  We will also denote by $\VV =
\CC \otimes_\RR V$ its complexification and we extend $\eta$ complex
linearly to a complex bilinear inner product on $\VV$.  The need for
complexification is due to the fact that we work in momentum space.
The complexes $\sC^\bullet(p)$ and $\Crel^\bullet(p)$ are also complex
vector spaces and complex conjugation relates $\sC^\bullet(p)$ and
$\sC^\bullet(-p)$ and similarly for the relative subcomplexes.  We
may impose reality of the spectrum by taking complex conjugate classes
in $\sH^\bullet(p)$ (resp. $\Hrel^\bullet(p)$) and $\sH^\bullet(-p)$
(resp. $\Hrel^\bullet(-p)$), as usual.

The momentum $p \in V^*$ defines an element in the dual of the real
representation $V$, but we may extend it complex linearly to
$p \in \VV^*$.  We let $\left<-,-\right>$ denote the dual pairing
between $\VV^*$ and $\VV$ so that if $v_1, v_2 \in V$ and
$v = v_1 + i v_2 \in \VV$, then
$\left<p,v\right> = \left<p,v_1\right> + i \left<p,v_2\right> \in
\CC$.

The inner product $\eta$ defines
Lorentz-equivariant musical isomorphisms
\begin{equation}
  \label{eq:musical}
  \begin{tikzcd}
    V \arrow[r, shift left=0.4em,"\flat"] & V^* \arrow[l, shift left=0.4em,"\sharp"]
  \end{tikzcd}
\end{equation}
under which $p \in V^*$ is sent to $p^\sharp \in V$.  We use the
notation $p^2 = \left<p, p^\sharp\right> = \eta(p^\sharp,p^\sharp)$,
which could be any real number.  Similarly, we also have musical
isomorphisms
\begin{equation}
  \label{eq:musical-complex}
  \begin{tikzcd}
    \VV \arrow[r, shift left=0.4em,"\flat"] & \VV^* \arrow[l, shift left=0.4em,"\sharp"]
  \end{tikzcd}
\end{equation}
induced by the complex bilinear $\eta$ on $\VV$.

The stabiliser subgroup $\Stab(p) \subset \SO(25,1)$ of a momentum
$p\in V^*$ (which is also the stabiliser subgroup of $p^\sharp \in V$)
is defined as
\begin{equation}
   \Stab(p) = \{\Lambda \in \SO(25,1) \mid \Lambda \cdot p = p \},
 \end{equation}
 where $\Lambda \cdot p = p \circ \Lambda^{-1}$ or, equivalently,
 $\left<\Lambda \cdot p, v\right> = \left<p, \Lambda^{-1} v\right>$
 for all $v \in V$.

 The subgroup  $\Stab(p)$ depends on $p$ as follows:
\begin{itemize}
\item If $p=0$, $\Stab(p) \cong SO(25,1)$;
\item If $p^2 < 0$, $\Stab(p) \cong SO(25)$;
\item If $p^2 > 0$, $\Stab(p) \cong SO(24,1)$;
\item If $p^2 = 0$ but $p \neq 0$, $\Stab(p) \cong \ISO(24)$, the
  $24$-dimensional euclidean group.
\end{itemize}
In all cases but the last, $\Stab(p)$ is semisimple and acts reducibly
on any finite-dimensional representation; particularly, in the
complexes $\Crel^\bullet(p)$ and $\sC^\bullet(p)$ and hence in their
cohomologies.  Alas, it is the last case that we will have to
concentrate on since, as we will show below, the BRST cohomology is
trivial unless $p^2 = 0$.  If $p=0$, then the complexes and the
cohomologies fully reduce into finite-dimensional $\SO(25,1)$-modules,
so we will concentrate on the more relevant, but also more involved,
case of $p\neq 0$ and $p^2 = 0$.

Fix a nonzero massless momentum $p \in V^*$ once and for all and let
$H:=\Stab(p) \subset \SO(25,1)$ denote the stabiliser subgroup of this momentum.  
As an $H$-module, $V$ is indecomposable with the following submodule structure:
\begin{equation}
  \label{eq:filtered-module}
  0 \subset \ell_p \subset p^\circ \subset V
\end{equation}
where here and in the sequel we use the shorthand $\ell_p = \RR
p^\sharp \subset V$ to denote the line spanned by $p^\sharp$ and
\begin{equation}
  \label{eq:p-perp}
  p^\circ = \left\{v \in V~\middle |~ \left<p,v\right> = 0\right\}
\end{equation}
denote the annihilator of $p\in V^*$ in $V = (V^*)^*$.  One can also
think of $p^\circ$ as the subspace perpendicular to $p^\sharp$.  Since
$p^2 = 0$, it follows that $\ell_p \subset p^\circ$.  Neither of the
$H$-submodules $\ell_p$ nor $p^\circ$ of $V$ have a complementary
submodule. If we view equation~\eqref{eq:filtered-module} as
describing a filtration of $V$, the associated graded module is
\begin{equation}
  \label{eq:assoc-graded-V}
  \ell_p \oplus (p^\circ/\ell_p) \oplus (V/p^\circ).
\end{equation}
Each of the above summands are $H$-modules which factor through a
maximal compact subgroup $K \subset H$.  Any such maximal compact
subgroup $K$ is isomorphic to $\SO(24)$.  In particular, $\ell_p$ and
$V/p^\circ$ are trivial one-dimensional $K$-modules, whereas
$V^\top := p^\circ/\ell_p$ is the standard (i.e., vector) module.
This module inherits a positive-definite inner product from $\eta$ as 
follows.  If we let $\bar v$ denote the image of $v \in p^\circ$ under
the canonical surjection $p^\circ \to V^\top$, we define the inner
product $\bar\eta$ on $V^\top$ by
\begin{equation}
  \label{eq:euclidean-IP-on-V-transverse}
  \bar\eta (\bar v, \bar w) = \eta(v,w),
\end{equation}
which is well-defined, $H$-invariant and positive-definite.
Because of this, we may also sometimes refer to $K$ as $\SO(V^\top)$.

We can make the decomposition in equation~\eqref{eq:assoc-graded-V}
more explicit by choosing a Witt frame $\be_+,\be_i,\be_-$, where
$i =1,\dots,24$, adapted to the filtration~\eqref{eq:filtered-module},
so that $p^\sharp$ is the span of $\be_+$ and $V/p^\circ$ is the span
of $\be_-$ and the $K$-module $V^\top$ is the span of the $\be_i$.  In
particular, $p^\circ$ is the span of the $\be_i$ and $\be_+$.

It is also convenient to complexify the above notions.  We will let
$L_p = \CC \otimes_\RR \ell_p = \CC p^\sharp \subset \VV$ and we will use
the notation $p^\perp = \CC \otimes_\RR p^\circ$.  Then we have a complex
flag
\begin{equation}
  \label{eq:complex-flag}
  0 \subset L_p \subset p^\perp \subset \VV.
\end{equation}
We will let $\VV^\top := p^\perp/L_p$ and with that notation, the
associated graded module to the filtration~\eqref{eq:complex-flag} is
\begin{equation}
  L_p \oplus \VV^\top \oplus \VV/p^\perp.
\end{equation}

\subsubsection{The Lie algebra of $\Stab(p)$}
\label{sec:lie-algebra-stabp}

Let us first describe the structure of $H$ or, in fact, its Lie
algebra. The Lie algebra $\so(V)$ is isomorphic (even as an
$\so(V)$-module) to $\ext{2}V$.  The isomorphism sends
$v \wedge w \in \ext{2}V$ to $v \curlywedge w \in \so(V)$, where
$v\curlywedge w$ is the endomorphism of $V$ defined by
\begin{equation}
  \label{eq:curlywedge}
  (v \curlywedge w) u := -\imath_u (v \wedge w) =\eta(w,u) v -  \eta(v,u) w,
\end{equation}
which is clearly skew-symmetric relative to $\eta$:
\begin{equation}
  \eta((v\curlywedge w)u, u)  = \eta(\eta(w,u) v -  \eta(v,u) w,u) =
  \eta(w,u)\eta(v,u) - \eta(v,u)\eta(w,u) = 0.
\end{equation}
The Lie algebra is obtained by extending the following Lie bracket
bilinearly:
\begin{equation}
  \label{eq:so-brackets}
  [u \curlywedge v, w \curlywedge x] = \eta(v,w) u \curlywedge x -
  \eta(u,w) v \curlywedge x - \eta(v,x) u \curlywedge w +
  \eta(u,x) v \curlywedge w.
\end{equation}
Let $\so(V) \xrightarrow{\Phi}\ext{2}V$ denote this isomorphism.  The image
under $\Phi$ of the Lie algebra $\h$ of the stabiliser of $p$ is given
by those $\Lambda \in \ext{2}V$ such that $\imath_p \Lambda = 0$.

\begin{lemma}
  \label{lem:stabp-in-sov}
  $\Phi$ restricts to an isomorphism $\h \xrightarrow{\cong} \ext{2}p^\circ$.
\end{lemma}

\begin{proof}
  Contraction with $p$ gives a map $\imath_p \colon \ext{2} V \to
  p^\circ$, since $\imath_p \circ \imath_p = 0$.  This map is moreover
  surjective.  Indeed, pick any $v \in V$ such that
  $\left<p,v\right> = 1$.  Then for any $w \in p^\circ$, $w = \imath_p
  (v \wedge \w)$.  By the Rank Theorem, $\dim \ker \imath_p = \dim
  \ext{2} V - \dim p^\circ = 325 -25 = 300$.  Now $\ext{2} p^\circ
  \subseteq \ker \imath_p$ and $\dim \ext{2} p^\circ = 300$, hence
  $\Phi(\h) = \ker \imath_p = \ext{2}p^\circ$.
\end{proof}

In terms of the explicit Witt frame, the image of $\stab(p)$ under
$\Phi$ is spanned by $\be_+ \wedge \be_i$ and $\be_i \wedge \be_j$,
with $i < j$.

\begin{lemma}
  \label{lem:iso-h2rel-skewsymmetric}
  The canonical surjection $\pi : p^\circ \to V^\top$ induces an
  isomorphism of $H$-modules
  \begin{equation}
    \dfrac{\ext{2}p^\circ}{\ell_p \wedge p^\circ} \cong \ext{2}V^\top
  \end{equation}
  and hence a short exact sequence
  \begin{equation}
    \begin{tikzcd}
      0 \arrow[r] & \ell_p \wedge p^\circ \arrow[r] & \ext{2}p^\circ \arrow[r] & \ext{2} V^\top \arrow[r] & 0.
    \end{tikzcd}
  \end{equation}
\end{lemma}

\begin{proof}
  Functorially, the canonical surjection $\pi$ induces a surjection
  $\ext{2}\pi \colon \ext{2}p^\circ \to \ext{2}V^\top$, whose kernel
  has dimension $\dim \ext{2}p^\circ - \dim \ext{2}V^\top = 300 - 276
  = 24$.  It is also clear that $\ell_p \wedge p^\circ \subset \ker
  \ext{2}\pi$ is also $24$-dimensional, hence $\ell_p \wedge p^\circ =
  \ker \ext{2}\pi$ and hence the isomorphism follows.  The short exact
  sequence is simply a restatement of the isomorphism.
\end{proof}

The short exact sequence in the Lemma can be rephrased in terms of Lie
algebras
\begin{equation}
  \label{eq:SES-stab-h-semidirect}
  \begin{tikzcd}
    0 \arrow[r] & \fa \arrow[r] & \h \arrow[r] & \so(V^\top) \arrow[r] & 0,
  \end{tikzcd}
\end{equation}
where we have introduced the notation $\fa := \ell_p \curlywedge
p^\circ$ for the abelian ideal.  Since $\so(V^\top) \cong \so(24)$ is
simple, this sequence splits and, together with the isomorphism $\fa \cong V^\top$
(see Proposition~\ref{prop:stabp-module-isos}~(1) below), this allows
us to identify $\h$ with the semidirect product $\so(V^\top)
\ltimes V^\top$; that is, the Lie algebra of euclidean transformations
of $V^\top$ relative to the positive-definite euclidean inner product
induced by $\eta$.

\subsubsection{$H$-submodules of $V \otimes V$}
\label{sec:some-submodules-v}

Foreshadowing the results of our calculations, we will first discuss
$V \otimes V = \sym{2} V \oplus \ext{2}V$ as an $H$-module. The
filtration~\eqref{eq:filtered-module} induces a Hasse diagram of
$H$-submodules of $V \otimes V$, where arrows denote inclusions:
\begin{equation}
  \label{eq:VtensorV-as-H-mod}
  \begin{tikzcd}
    & & \ell_p \otimes V \arrow[rd] & & \\
    & \ell_p \otimes p^\circ  \arrow[rd] \arrow[ru] & & p^\circ \otimes V \arrow[dr] & \\
    \ell_p \otimes \ell_p \arrow[ru] \arrow[rd] & & p^\circ \otimes p^\circ \arrow[ru] \arrow[rd] & & V \otimes V\\
    & p^\circ \otimes \ell_p \arrow[rd] \arrow[ru] & & V \otimes p^\circ \arrow[ru] & \\
    & & V \otimes \ell_p \arrow[ru] & & \\
  \end{tikzcd}
\end{equation}

\subsubsection{$H$-submodules of $\ext{2}V$}
\label{sec:submodules-skew-2-v}

Intersecting the Hasse diagram \eqref{eq:VtensorV-as-H-mod} with
$\ext{2}V$ we arrive at the following Hasse diagram of $H$-submodules
of $\ext{2}V$:
\begin{equation}
  \label{eq:VwedgeV-as-H-mod}
  \begin{tikzcd}
    & \ell_p \wedge V \arrow[rd] & & \\
    \ell_p \wedge p^\circ \arrow[rd] \arrow[ru] & & p^\circ \wedge V \arrow[r,equal] & \ext{2}V\\  
    & \ext{2}p^\circ \arrow[ru] & & \\
  \end{tikzcd}
\end{equation}
where
\begin{equation}
  \begin{split}
    \ell_p \wedge V &= \ker \varepsilon_p \colon \ext{2}V \to \ext{3}V\\
    \ext{2}p^\circ &= \ker \imath_p \colon \ext{2}V \to p^\circ\\
    \ell_p \wedge p^\circ &= \ker \imath_p \colon \ell_p \wedge V \to \ell_p
  \end{split}
\end{equation}
where $\imath_p$ and $\varepsilon_p$ are interior and exterior
products with $p$ and $p^\sharp$, respectively.  We can collapse this
Hasse diagram into the following filtration:
\begin{equation}
  \label{eq:VwedgeV-filtered}
  0 \subset  \ell_p \wedge p^\circ \subset \left( \ell_p \wedge V + \ext{2}p^\circ  \right) \subset \ext{2}V,
\end{equation}
whose associated graded module is
\begin{equation}
  \label{eq:VwedgeV-assoc-graded}
  \left( \ell_p \wedge p^\circ  \right)\oplus \dfrac{\ell_p \wedge V + \ext{2}p^\circ}{\ell_p \wedge p^\circ}\oplus \dfrac{\ext{2}V}{\ell_p \wedge V + \ext{2}p^\circ}.
\end{equation}
Each of these modules factors through a maximal compact subgroup, as
we now show.  Indeed it is clear that under the action of $\fa$,
\begin{equation}
  \begin{split}
    \ell_p \wedge p^\circ & \mapsto 0\\
    \ell_p \wedge V + \ext{2} p^\circ & \mapsto \ell_p \wedge p^\circ\\
    \ext{2} V & \mapsto \ell_p \wedge V + \ext{2} p^\circ,
  \end{split}
\end{equation}
so that the action is trivial on the above (quotient) modules.  We can
be more precise about the nature of those modules, after a couple of
preliminary results.

\begin{lemma}
  \label{lem:module-iso}
  The following isomorphism of $H$-modules holds:
  \begin{equation}
    \dfrac{\ell_p \wedge V + \ext{2} p^\circ}{\ell_p \wedge V} \cong \dfrac{\ext{2}p^\circ}{\ell_p \wedge p^\circ}.
  \end{equation}
\end{lemma}

\begin{proof}
  This follows from the Second Isomorphism Theorem for modules
  \cite[Theorem~4(2), Section~10.2]{DummitFoote}: if $E,F$ are
  submodules of a module $M$ over a ring $R$, then there is an
  isomorphism
  \begin{equation}
    \dfrac{E + F}{F} \cong \dfrac{E}{E \cap F}.
  \end{equation}
  We apply this to $E = \ext{2}p^\circ$ and $F = \ell_p \wedge V$, which are
  submodules of $\ext{2}V$ as a module over the universal enveloping
  algebra of $\stab(p)$, and we use that that $E \cap F =
  \ext{2}p^\circ \cap (\ell_p \wedge V) = \ell_p \wedge p^\circ$.
\end{proof}

\begin{lemma}
  \label{lem:another-module-iso}
  The following isomorphism of $H$-module holds:
  \begin{equation}
    \dfrac{\ell_p \wedge V}{\ell_p \wedge p^\circ} \cong \ell_p.
  \end{equation}
\end{lemma}

\begin{proof}
  Restricting the interior product by $p$ to $\ell_p \wedge V \subset
  \ext{2}V$ we arrive at the following short exact sequence
  \begin{equation}
    \begin{tikzcd}
      0 \arrow[r] & \ell_p \wedge p^\circ \arrow[r] & \ell_p \wedge V  \arrow[r,"\imath_p"] & \ell_p \arrow[r] & 0
    \end{tikzcd}
  \end{equation}
  from which the isomorphism follows.
\end{proof}

\begin{proposition}
  \label{prop:stabp-module-isos}
  There are $H$-module isomorphisms
  \begin{enumerate}
  \item $\ell_p \wedge p^\circ \cong V^\top$
  \item $(\ell_p \wedge V + \ext{2}p^\circ)/(\ell_p \wedge p^\circ) \cong  \ell_p  \oplus \ext{2}V^\top$
  \item $(\ext{2}V)/(\ell_p \wedge V + \ext{2}p^\circ) \cong V^\top$
  \end{enumerate}
  where all $H$-modules in the right-hand sides factor through
  a maximal compact subgroup $K$.
\end{proposition}

\begin{proof}
  \begin{enumerate}
  \item Exterior product with $p^\sharp$ gives a short exact sequence of
    $H$-modules and $H$-equivariant maps
    \begin{equation}
      \begin{tikzcd}
        0 \arrow[r] & \ell_p \arrow[r] & p^\circ \arrow[r,"p^\sharp\wedge -"] & \ell_p \wedge p^\circ \arrow[r] & 0,
      \end{tikzcd}
    \end{equation}
    from where the isomorphism follows.
  \item We use the Third Isomorphism Theorem \cite[Theorem~4(3), Section~10.2]{DummitFoote} for modules applied to
    $\ell_p \wedge p^\circ \subset \ell_p \wedge V \subset \ell_p \wedge V +
      \ext{2}p^\circ$ to exhibit the first of the following isomorphisms:
      \begin{equation}
        \dfrac{(\ell_p \wedge V + \ext{2}p^\circ)/(\ell_p \wedge p^\circ)}{(\ell_p
          \wedge V)/(\ell_p \wedge p^\circ)} \cong \dfrac{\ell_p \wedge V +
          \ext{2}p^\circ}{\ell_p \wedge V} \cong \dfrac{\ext{2}p^\circ}{\ell_p
          \wedge p^\circ} \cong \ext{2}V^\top,
      \end{equation}
      where the second isomorphism is Lemma~\ref{lem:module-iso} and the third is
      Lemma~\ref{lem:iso-h2rel-skewsymmetric}.   Using
      Lemma~\ref{lem:another-module-iso}, we have a short exact
      sequence of $H$-modules
      \begin{equation}
        \begin{tikzcd}
          0 \arrow[r] & \ell_p \arrow[r] &  \dfrac{\ell_p \wedge V + \ext{2}p^\circ}{\ell_p \wedge p^\circ}\arrow[r] & \ext{2}V^\top \arrow[r] & 0,
        \end{tikzcd}
      \end{equation}
      where the first and third $H$-modules factor through a
      maximal compact subgroup.  It is easy to see that the middle
      module also factors through a maximal compact subgroup.
      Indeed, it is easy to see that $\fa$ sends $\ell_p
      \wedge V + \ext{2}p^\circ$ to $\ell_p \wedge p^\circ$ and hence it
      acts trivially on the quotient.  Since any maximal compact
      subgroup acts reducibly, the above sequence splits and we may
      conclude that
      \begin{equation}
        \dfrac{\ell_p \wedge V + \ext{2}p^\circ}{\ell_p \wedge p^\circ} \cong \ell_p \oplus \ext{2}V^\top.
      \end{equation}
    \item Interior product with $p$ gives a map $\imath_p \colon
    \ext{2}V \to p^\circ$, since $\imath_p \circ \imath_p = 0$.
    Composing this with the projection $p^\circ \to V^\top$ gives a
    $H$-equivariant surjective linear map $\ext{2} V \to
    V^\top$, whose kernel is precisely $\ell_p \wedge V + \ext{2}p^\circ$.
  \end{enumerate}
\end{proof}

In summary, the associated graded
module~\eqref{eq:VwedgeV-assoc-graded} of the filtered
module~\eqref{eq:VwedgeV-filtered} is isomorphic to
\begin{equation}
  \ext{2}V^\top \oplus 2 V^\top \oplus \ell_p.
\end{equation}

\begin{remark}
  We could have obtained this result in a quicker, albeit less
  covariant, way by choosing a Witt frame $(\be_+,\be_i,\be_-)$
  adapted to the filtration of $V$ and then noticing that $\ext{2}V$
  has basis $\be_i \wedge \be_j$ ($i<j$), which spans $\ext{2}V^\top$,
  $\be_i \wedge \be_{\pm}$, which spans two copies of $V^\top$, and
  $\be_+ \wedge \be_-$, which spans a copy of the trivial module
  isomorphic to $\ell_p$.  A Witt frame involves a choice of $\be_-$,
  which is a section of $V \to V/p^\circ$, and one must then show that
  everything is independent of that choice.   It seems more
  transparent to work in a manifestly covariant way ab initio.
\end{remark}

\subsubsection{$H$-submodules of $\sym{2}V$}
\label{sec:submodules-sym-2-v}

Alternatively, we may intersect \eqref{eq:VtensorV-as-H-mod} with
$\sym{2} V$ to arrive at the following Hasse diagram of $H$-submodules
of $\sym{2} V$:
\begin{equation}
  \label{eq:VsymV-as-H-mod}
  \begin{tikzcd}
    & & \ell_p \odot V \arrow[rd] & &\\
    \ell_p \odot \ell_p \arrow[r]  & \ell_p \odot p^\circ \arrow[ru] \arrow[rd] & &  p^\circ \odot V \arrow[r] & \sym{2} V\\
    & & \sym{2} p^\circ \arrow[ru] & &\\
  \end{tikzcd}
\end{equation}
where
\begin{equation}
  \begin{split}
    \ell_p \odot p^\circ &= \ker \imath_p : \ell_p \odot V \to \ell_p\\
    \sym{2} p^\circ &= \ker i_p : \sym{2} V \to V\\
    p^\circ \odot V &= \ker i_p^2 : \sym{2} V \to \CC.
  \end{split}
\end{equation}
The $\SO(V)$-module $\sym{2} V$ is not irreducible: it decomposes into
the symmetric traceless $\sym{2}_0V$ and the one-dimensional trivial
trace modules.  For any submodule $M \subset \sym{2} V$, we will let 
$M_0$ denote the intersection $M \cap \sym{2}_0 V$ hereafter.  Using
that $(\ell_p\odot V)_0 = \ell_p \odot p^\circ$, the above Hasse
diagram restricted to the symmetric traceless $\sym{2}_0 V$ collapses
to the following filtration:
\begin{equation}
  \label{eq:VsymVtraceless-as-H-mod}
  0 \subset  \ell_p \odot \ell_p \subset \ell_p \odot p^\circ  \subset \sym{2}_0  p^\circ \subset (\ell_p \odot V + \sym{2}p^\circ)_0 \subset \left( p^\circ \odot V\right)_0 \subset \sym{2}_0 V.
\end{equation}
The associated graded module factors through a maximal compact
subgroup of $H$.  Indeed, it is not hard to see that under the
action of $\fa$,
\begin{equation}
  \label{eq:action-of-ideal-on-sym2}
  \begin{split}
    \ell_p \odot p & \mapsto 0\\
    \ell_p \odot p^\circ & \mapsto \ell_p \odot p\\
    \sym{2}_0 p^\circ& \mapsto \ell_p \odot p^\circ\\
    (\ell_p \odot V + \sym{2}p^\circ)_0 & \mapsto \ell_p \odot p^\circ \\
    (p^\circ \odot V)_0 & \mapsto  (\ell_p \odot V + \sym{2}p^\circ)_0\\
  \sym{2}_0 V & \mapsto (p^\circ \odot V)_0,
  \end{split}
\end{equation}
so that the action is trivial on the above (quotient) modules.  As
before, we can be more precise about the nature of the above modules,
but not before a preliminary result.

\begin{lemma}
  \label{lem:iso-h2rel-symmetric-submodule}
  The canonical surjection $\pi : p^\circ \to V^\top$ induces an isomorphism of $H$-modules
  \begin{equation}
    \dfrac{\sym{2}p^\circ}{\ell_p \odot p^\circ} \cong \sym{2} V^\top,
  \end{equation}
  which, since $\ell_p\odot p^\circ \subset \sym{2}_0 p^\circ$ is already
  traceless, restricts to an isomorphism
  \begin{equation}
    \dfrac{\sym{2}_0p^\circ}{\ell_p \odot p^\circ} \cong \sym{2}_0 V^\top.
  \end{equation}
\end{lemma}

\begin{proof}
  This is very similar to the proof of
  Lemma~\ref{lem:iso-h2rel-skewsymmetric}.  Functorially, we have a
  surjection
  $\sym{2}\pi \colon \sym{2} p^\circ\to \sym{2} V^\top$ whose
  kernel has dimension
  $\dim \sym{2} p^\circ - \dim \sym{2} V^\top = 325 - 300 = 25$,
  which is precisely the dimension of
  $\ell_p \odot p^\circ \subset \ker \sym{2} \pi$.  This shows that
  $\ker \sym{2} \pi = \ell_p \odot p^\circ$ and proves the first
  isomorphism.  The second follow after restricting to traceless
  tensors and noticing that $\ell_p \odot p^\circ \subset \sym{2}_0
  p^\circ$ is already traceless.
\end{proof}

\begin{proposition}
  \label{prop:assoc-graded-symmetric-modules}
  There are additional $H$-module isomorphisms
  \begin{enumerate}
  \item $\ell_p \odot \ell_p \cong \ell_p$
  \item $(\ell_p \odot p^\circ)/(\ell_p \odot \ell_p) \cong V^\top$
  \item $(\ell_p \odot V + \sym{2}p^\circ)_0/\sym{2}_0 p^\circ \cong p$
  \item $(p^\circ \odot V)_0/(\ell_p \odot V + \sym{2}p^\circ)_0 \cong V^\top$
  \item $\sym{2}_0 V/(p^\circ \odot V)_0 \cong V/p^\circ$
  \end{enumerate}
\end{proposition}

\begin{proof}
  \begin{enumerate}
  \item Since $p$ is invariant under $H$ by definition, the
    linear map sending $p^\sharp \odot p^\sharp \mapsto p^\sharp$ is
    an $H$-equivariant isomorphism $\ell_p \odot \ell_p \cong \ell_p$.
  \item This is induced by the $H$-equivariant projection
    $\pi \colon p^\circ \to V^\top$.  We define a $H$-equivariant
    projection $\ell_p \odot p^\circ \to V^\top$ which sends $p^\sharp \odot
    v \mapsto \pi(v)$, for $v \in p^\circ$.  This has kernel
    precisely $\ell_p \odot p$.
  \item Contraction with $p$ gives a $H$-equivariant surjection
    $\ell_p \odot V + \sym{2}p^\circ \to p$ whose kernel is precisely
    $\ell_p \odot p^\circ + \sym{2}p^\circ = \sym{2}p^\circ$.  Restricting to
    the traceless tensors, we still have a surjection $(\ell_p \odot V +
    \sym{2}p^\circ)_0 \to \ell_p$ whose kernel is now the traceless
    $\sym{2}_0 p^\circ$, hence the isomorphism.
  \item Contraction with $p$ gives a $H$-equivariant linear map
    $p^\circ \odot V \to p^\circ$ which we can compose with the
    $H$-equivariant projection $p^\circ \to V^\top$ to
    arrive at a $H$-equivariant surjection $(p^\circ \odot V)
    \to V^\top$.  Restricting to the traceless $(p^\circ \odot
    V)_0$ still surjects onto $V^\top$, since there is no
    codimension-one $\so(V^\top)$-submodule of $V^\top$.  The
    kernel of this map is precisely $(\ell_p \odot V + \sym{2} p^\circ)_0$,
    hence the isomorphism follows.
  \item Contraction with $p$ gives a $H$-equivariant linear map
    $\sym{2}V \to V$ which we can compose with the
    $H$-equivariant projection $V \to V/p^\circ$ to  arrive at
    a $H$-equivariant surjection $\sym{2}V \to V/p^\circ$,
    whose kernel is precisely $p^\circ \odot V$.  Restricting to
    traceless $\sym{2}_0V$, the map $\sym{2}_0 V \to V/p^\circ$ is
    still surjective and the kernel is now the traceless $(p^\circ
    \odot V)_0$, hence the isomorphism.
  \end{enumerate}
\end{proof}

In summary, the associated graded module
\begin{equation}
  (\ell_p \odot p) \oplus \dfrac{\ell_p\odot p^\circ}{\ell_p\odot \ell_p} \oplus
  \dfrac{\sym{2}_0 p^\circ}{\ell_p \odot p^\circ} \oplus \dfrac{(\ell_p \odot V
    +  \sym{2} p^\circ)_0}{\sym{2}_0 p^\circ} \oplus \dfrac{(p^\circ
    \odot V)_0}{(\ell_p \odot V +  \sym{2} p^\circ)_0} \oplus
  \dfrac{\sym{2}_0V}{(p^\circ \odot V)_0}
\end{equation}
of the filtered module~\eqref{eq:VsymVtraceless-as-H-mod} is
isomorphic to
\begin{equation}
  \sym{2}_0V^\top \oplus 2 V^\top \oplus 2\ell_p \oplus  V/p^\circ,
\end{equation}
where of course $V/p^\circ \cong \RR p \cong \ell_p$ as both are
trivial modules.

Finally, we record in Table~\ref{tab:dimensions} the dimensions of
some of the $H$-modules discussed in this section.

\begin{table}[h!]
  \centering
  \begin{tabular}{>{$}c<{$}|r}
    \multicolumn{1}{l|}{Module} & \multicolumn{1}{l}{$\dim$}\\\toprule
    \ell_p & 1\\
    V/p^\circ & 1\\
    V^\top & 24\\
    \ell_p \wedge p^\circ & 24\\
    p^\circ & 25\\
    \ell_p \wedge V & 25\\
    \ell_p \odot p^\circ & 25\\
    \left(\ell_p \odot V\right)_0 & 25\\\bottomrule
  \end{tabular}\qquad\qquad
  \begin{tabular}{>{$}c<{$}|r}
    \multicolumn{1}{l|}{Module} & \multicolumn{1}{l}{$\dim$}\\\toprule
    V/\ell_p & 25\\
    V & 26\\
    \ell_p \odot V & 26\\
    \ext{2} V^\top & 276\\
    \sym{2}_0V^\top & 299\\
    \sym{2}V^\top & 300\\
    \ext{2} p^\circ & 300\\
    \sym{2}_0 p^\circ & 324\\\bottomrule
  \end{tabular}\qquad\qquad
  \begin{tabular}{>{$}c<{$}|r}
    \multicolumn{1}{l|}{Module} & \multicolumn{1}{l}{$\dim$}\\\toprule
    \sym{2} p^\circ & 325\\
    \ext{2}V & 325\\
    \left(\ell_p \odot V + \sym{2}p^\circ\right)_0 & 325\\
    \ell_p \odot V + \sym{2}p^\circ & 326\\
    \left(p^\circ \odot V\right)_0 & 349\\
    p^\circ \odot V & 350\\
    \sym{2}_0 V & 350\\
    \sym{2}V & 351\\\bottomrule
  \end{tabular}
  \caption{Dimensions of some $H$-modules}
  \label{tab:dimensions}
\end{table}

All of the results in
Sections~\ref{sec:some-submodules-v},~\ref{sec:submodules-skew-2-v} and
\ref{sec:submodules-sym-2-v} survive upon complexifying $V$.  Indeed,
replacing $\RR$ with $\CC$, $(\ell_p, p^\circ,V,V^\top)$ with their
complexifications $(L_p, p^\perp, \VV,\VV^\top)$ and re-interpreting
real dimensions as complex dimensions, we obtain complex versions of
these results.  We will not write them again and simply assume tacitly
any reference to any result in
Sections~\ref{sec:some-submodules-v},~\ref{sec:submodules-skew-2-v}
and \ref{sec:submodules-sym-2-v} applies equally well to the complex
version of the result.

\section{The relative cohomology}
\label{sec:relative-cohomology}

We now start by calculating the relative BRST cohomology. Our approach
is very pedestrian. We determine the kernel of the differential at
every ghost number.  We start by describing the relative complex.

\subsection{The relative subcomplex}
\label{sec:relative-subcomplex}

Recall that the relative subcomplex is given by $\Crel^\bullet(p) :=
\sC^\bullet(p) \cap \ker b_0$. Thus, we can construct an explicit,
finite-dimensional basis of $\Crel^n(p)$ for each ghost number $n \in
\{0,1,\dots, 5\}$ as monomials of the modes $c_1, c_{n \leq -1}, C_{n
  \leq 1}, b_{n \leq -2}, B_{n \leq -2}, x^\mu_{n \leq -1}, \pi^\mu_{n
  \leq -1}$ acting on $\ket{p} = \lim_{z\to 0} W_p(z) \ket{0}$ with
ghost number $n$ and conformal weight zero. For instance, a basis for
$\Crel^1(p)$ is given by
\begin{equation} \label{eq:basis for C1rel in modes}
  C_0\ket{p},\ c_1 C_1 b_{-2}\ket{p},\ c_1 C_1 B_{-2}\ket{p},\ c_1 x^\mu_{-1}\ket{p} ,\ C_1  x^\mu_{-1}\ket{p},\ c_1 \pi^\mu_{-1}\ket{p}, C_1\pi^\mu_{-1}\ket{p},
\end{equation}
so a general element of $\Crel^1(p)$ is a linear combination of these.
Hence, we may deduce the action of $d$ on $\Crel^\bullet(p)$ by
calculating the action of $d$ on each basis element of
$\Crel^\bullet(p)$. This is more easily done using the language of
fields and operator product expansions instead of modes as written
above. The state-field correspondence sets up the following dictionary
(for more details presented in terms of notation compatible with this
paper, we refer the reader to \cite[Appendix
A.1]{Figueroa-OFarrill:2024wgs}):
\begin{equation}
  \label{eq:state-field-corr}
  \begin{aligned}
    c &\leftrightarrow c_1\\
    \d c &\leftrightarrow c_0\\
    \d^2 c &\leftrightarrow 2 c_{-1}\\
    \d^3 c &\leftrightarrow 6 c_{-2}\\
    b &\leftrightarrow b_{-2}
  \end{aligned}
  \qquad\qquad
  \begin{aligned}
    C &\leftrightarrow C_1\\
    \d C &\leftrightarrow C_0\\
    \d^2 C &\leftrightarrow 2 C_{-1}\\
    \d^3 C &\leftrightarrow 6 C_{-2}\\
    B &\leftrightarrow B_{-2}
  \end{aligned}
  \qquad\qquad
  \begin{aligned}
    \d X^\mu &\leftrightarrow x^\mu_{-1}\\
    \d^2 X^\mu &\leftrightarrow 2 x^\mu_{-2}\\
    \Pi^\mu &\leftrightarrow \pi^\mu_{-1}\\
    \d \Pi^\mu &\leftrightarrow \pi^\mu_{-2}.\\
  \end{aligned}
\end{equation}
In this language, the basis \eqref{eq:basis for C1rel in modes} for $\Crel^1(p)$ reads
\begin{equation}
    \d C W_p,\ cCb W_p,\ cCB W_p,\ c\d X^\mu W_p,\ C\d X^\mu W_p,\ c\d \Pi^\mu W_p,\ C\d\Pi^\mu W_p.
\end{equation}
Likewise, every basis state of $\Crel^\bullet(p)$ can be written as a
basis field instead, and the action of the BRST differential $d$ on
any basis state $\Psi$ gives a state $d \Psi$ that corresponds to the field appearing in the first-order pole of the OPE $j_\mathrm{BRST}(z)
\Psi(w)$. We compute this first-order pole for all the basis states of
$\Crel^\bullet(p)$ using the Mathematica package \texttt{OPEdefs}
written by Kris Thielemans
\cite{Thielemans:1991uw,Thielemans:1992mu,Thielemans:1994er} and
record our calculations in Appendix \ref{app:calculation of d on Crel}.

The action of the Lorentz algebra $\so(V)$ on the BRST complex is such
that
\begin{equation}
  [\Lambda(\omega), W_p]_1 = i \omega_{\mu\nu} p^\nu (X^\mu W_p),
\end{equation}
so that it annihilates $W_p$ precisely when $\omega_{\mu\nu} p^\nu =
0$, which says that $\omega \in \ext{2}p^\perp$ is in the image of
$\stab(p)$ under the isomorphism $\so(V) \xrightarrow{\Phi} \ext{2} V$
introduced in Section~\ref{sec:stabp-reps}.

Furthermore, since $[\Lambda(\omega),-]_1$ is a derivation over all
the $[-,-]_\ell$ brackets, we have that for all fields $A$,
\begin{equation}
  [\Lambda(\omega), [T,A]_2]_1 = [[\Lambda(\omega), T]_1, A]_2 + [T, [\Lambda(\omega),A]_1]_2,
\end{equation}
but since $\Lambda(\omega)$ is a conformal primary of weight $1$,
$[\Lambda(\omega),T]_1 = 0$ and hence the zero mode of
$\Lambda(\omega)$ commutes with $L_0 = [T,-]_2$.  It clearly also
commutes with $b_0 = [b,-]_2$ and hence (for $\omega \in
\ext{2}p^\perp$) it preserves the relative subcomplex and therefore it
defines  an action of $\stab(p)$ on the relative BRST cohomology.

\subsection{Numerical results}
\label{sec:numerical-results}

One can put the complex on Mathematica and solve the linear equations
to determine the dimension of the space of cocycles and, by using the
Rank Theorem, that of the coboundaries and hence of the cohomology.
Table~\ref{tab:numerical} is a summary of the results of these
computations.

\begin{table}[h!]
\centering
\caption{Summary of the calculation of $\dim\Hrel^\bullet(p)$}
\label{tab:numerical}
\begin{tabular}{>{$}c<{$}|*{4}{>{$}l<{$}}}
    \toprule\\
    n & \multicolumn{1}{c}{$\dim\Crel^n(p)$} & \multicolumn{1}{c}{$\dim\Zrel^n(p)$} & \multicolumn{1}{c}{$\dim\Brel^n(p)$} & \multicolumn{1}{c}{$\dim\Hrel^n(p)$}\\
    \midrule
    0 & 1 & \makecell[l]{\begin{cases} 1, & p=0\\ 0, & p \neq 0 \end{cases}} & 0 & \makecell[l]{\begin{cases}1,&p=0\\0,&p \neq 0\end{cases}}\\
    1 & 107 & \makecell[l]{\begin{cases}52,&p=0\\ 1,&p \neq 0\end{cases}} & \makecell[l]{\begin{cases}0,&p =0 \\ 1,&p\neq 0\end{cases}}& \makecell[l]{\begin{cases}52,&p=0\\0,&p \neq 0\end{cases}}\\
    2 & 1540 & \makecell[l]{\begin{cases}742,&p=0\\ 706,&p^2 =0\\ 106,&p^2 \neq 0\end{cases}} & \makecell[l]{\begin{cases}55,&p =0 \\ 106,&p\neq 0\end{cases}} & \makecell[l]{\begin{cases}677,&p=0\\ 600,&p^2 =0\\ 0,&p^2 \neq 0\end{cases}}\\
    3 & 1540  & \makecell[l]{\begin{cases}1485,&p=0\\ 1434,&p \neq 0\end{cases}} & \makecell[l]{\begin{cases}808,&p=0\\ 834,&p^2 =0 \\ 1434,&p^2 \neq 0\end{cases}} &  \makecell[l]{\begin{cases}677,&p=0\\ 600,&p^2 =0 \\ 0,&p^2 \neq 0\end{cases}}\\
    4 & 107 & \makecell[l]{\begin{cases}107,&p=0\\ 106,&p \neq 0\end{cases}} & \makecell[l]{\begin{cases}55,&p=0\\ 106,&p \neq 0\end{cases}} & \makecell[l]{\begin{cases}52,&p=0\\0,&p \neq 0\end{cases}}\\
    5 & 1 & 1 & \makecell[l]{\begin{cases}0,&p=0\\ 1,&p \neq 0\end{cases}} & \makecell[l]{\begin{cases}1,&p=0\\0,&p \neq 0\end{cases}}\\
    \bottomrule
  \end{tabular}
  \caption*{The condition $p^2 = 0$ tacitly assumes that $p\neq 0$.}
\end{table}

These calculations again show that Proposition~\ref{prop:on-shell}
holds and, as a further corollary of these calculations, we arrive at
the following result:

\begin{proposition}\label{prop:vector-space-duality-Hrel}
  $\Hrel^n(p) \cong \Hrel^{5-n}(p)$, where the isomorphism is one of complex vector spaces.
\end{proposition}

These numerical results determine the relative cohomology as a vector
space.  However, the vector space description of cohomology is
evidently unsatisfactory since we are left with too little information
to deduce any meaningful physics from it. To infer physics, recall
from Section \ref{sec:what-is-the-cohomology} that we must understand
the (relative) BRST cohomology as modules over $H := \Stab(p)$.

We remark that Table \ref{tab:numerical} already shows that the spectrum
must be larger than the space spanned by the graviton, dilaton and
Kalb--Ramond field.  As explained in
Section~\ref{sec:what-is-the-cohomology}, we view the cohomology as
inducing representations for (massless, due to
Proposition~\ref{prop:on-shell}) Poincaré representations.  The
Kalb--Ramond field is induced from $\ext{2}\VV^\top$ (which has
dimension $276$), whereas the graviton is induced by
$\sym{2}_0\VV^\top$ (which has dimension $299$) and the dilaton is
induced by the trace part $\CC \eta \subset \sym{2}\VV^\top$, which
is one-dimensional.  Adding the dimensions of these inducing
representations gives $576$, which falls short of the numerical result
of $600$ in Table~\ref{tab:numerical}.  The deficit is precisely a
transverse vector $\VV^\top$ and it would be tempting to think that
this is due to a massless vector, but as we will see in
Section~\ref{sec:phys-interpr}, this does not induce a unitary
representation at ghost number $2$.

Recall from Section \ref{sec:stabp-reps} that we must distinguish
between the two cases $p = 0$ and $p \neq 0$ when $p^2 = 0$ (as
dictated by Proposition \ref{prop:on-shell}). In the former case, we
have full Lorentz covariance (in particular, $\Stab(0)$ is the Lorentz
group) and the cohomology can be described as modules over the Lorentz
group $\SO(25,1)$.  However, the latter case, which is also the physically
interesting one, is more challenging due to the non-semisimplicity of
$H$, which is now isomorphic to $\ISO(24)$. To obtain some intuition,
we may first understand relative cohomology as a module over a maximal
compact subgroup $K \subset H$, where $K \cong \SO(V^\top) =
\SO(24)$. Thus, in this section, we go ghost number by ghost number
through the calculation of the $K$-module structure of relative
cohomology. This will also lay the groundwork for the more difficult
analysis of the relative cohomology at ghost numbers 2 and 3 as
$H$-modules carried out in Section~\ref{sec:phys-interpr}.

\subsection{Calculating $\Hrel^0(p)$}
\label{sec:calculating-hrel0p}

The space $\Crel^0(p)$ of relative $0$-cochains is one-dimensional and
spanned by $\vac$.  Using equation~\eqref{eq:dvac}, we see that the
space $\Zrel^0(p)$ of relative $0$-cocycles is given by
\begin{equation}
  \Zrel^0(p) \cong
  \begin{cases}
    \CC & p=0\\
    0 & p \neq 0,
  \end{cases}
\end{equation}
where $\CC$ stands for the one-dimensional complex scalar (i.e.,
trivial) representation of $\SO(25,1)$.  Thus we conclude that
\begin{equation}
  \Hrel^0(p) \cong
  \begin{cases}
    \CC & p=0\\
    0 & p \neq 0.
  \end{cases}
\end{equation}
From the Rank Theorem we also see that the space $\Brel^1(p)$ of
coboundaries at ghost number $1$ is given by
\begin{equation}
  \label{eq:B1rel-as-reps}
  \Brel^1(p) \cong
  \begin{cases}
    0 & p = 0 \\
    \CC & p \neq 0.
  \end{cases}
\end{equation}

\subsection{Calculating $\Hrel^1(p)$}
\label{sec:calculating-hrel1p}

We may write the general relative $1$-cochain $\Psi \in \Crel^1(p)$ as follows:
\begin{equation}
  \Psi = \phi^{(1)} \d C + \phi^{(2)} c C b + \phi^{(3)} c C B + A^{(4)}_\mu c \d X^\mu + A^{(5)}_\mu C \d X^\mu + A^{(6)}_\mu c \Pi^\mu + A^{(7)}_\mu C \Pi^\mu
\end{equation}
and the cocycle condition $d\Psi = 0$ reduces, after some simplification, to the following:
\begin{equation}
  \begin{aligned}
    \phi^{(2)} = \phi^{(3)} &=0\\
    A^{(5)}_\mu &= 0\\
    A^{(7)}_\mu &= - A^{(4)}_\mu
  \end{aligned}
  \qquad\qquad
  \begin{aligned}
    \phi^{(1)} &= - \tfrac{i}2 p^\mu A^{(4)}_\mu\\
    p_{(\mu} A^{(6)}_{\nu)} &= 0\\
    p_{[\mu} A^{(4)}_{\nu]} &= 0.
  \end{aligned}
\end{equation}

\subsubsection{Case $p=0$}
\label{sec:h1-p=0}

Here $\phi^{(1)} =0$ and $A_\mu^{(6)}$ and $A_\mu^{(7)} = -
A_\mu^{(4)}$ remain otherwise unconstrained, so that $\Zrel^1(0)$ is
spanned by
\begin{equation}\label{eq:Zrel10}
  c \Pi^\mu \qquad\text{and}\qquad c \d X^\mu - C \Pi^\mu.
\end{equation}
In other words, we have that
\begin{equation}
  \Crel^1(0) \cong 3 \CC \oplus 4 \VV  \qquad\text{and}\qquad \Zrel^1(0)
  \cong 2 \VV.
\end{equation}
Since from equation \eqref{eq:B1rel-as-reps}, $\Brel^1(0)=0$, we have
that $\Hrel^1(0) \cong 2 \VV$ and from Rank Theorem, $\Brel^2(0)\cong 3
\CC \oplus 2 \VV$.

\subsubsection{Case $p^2 = 0$ but $p\neq 0$}
\label{sec:h1-p2=0}

Here $A^{(4)}_\mu \propto p_\mu$ and hence $\phi^{(1)} =0$.  Similarly
$A^{(6)}_\mu = 0$ and $\Zrel^1(p)$ is spanned by
\begin{equation}
  p_\mu c \d X^\mu - p^\mu C \Pi_\mu,
\end{equation}
which agrees with the $1$-coboundaries $\Brel^1(p)$, as seen from the
expression for $dW_p$ in \eqref{eq:dvac} for $p^2 = 0$.  As
$K$-modules, we have that $\VV \cong 2\CC \oplus \VV^\top$.  Therefore
we have that as $K$-modules, for such $p$,
\begin{equation}
  \Crel^1(p) \cong 11 \CC \oplus 4 \VV^\top \qquad\text{and}\qquad
  \Zrel^1(p) \cong \CC.
\end{equation}
This allows us to conclude, from equation~\eqref{eq:B1rel-as-reps},
that
\begin{equation}
  \Hrel^1(p) \cong 0
\end{equation}
and, from the Rank Theorem, that
\begin{equation}
  \Brel^2(p) \cong 10 \CC \oplus 4 \VV^\top.
\end{equation}

\subsubsection{Summary}
\label{sec:hrel1-summary}

In summary, we see that
\begin{equation}
  \Hrel^1(p) \cong
  \begin{cases}
    2 \VV & p = 0\\
    0 & p \neq 0.
  \end{cases}
\end{equation}

\subsection{Calculating $\Hrel^2(p)$}
\label{sec:calculating-hrel2p}

The general relative $2$-cochain $\Psi \in \Crel^2(p)$ can be written as
\begin{multline}
  \Psi = \phi^{(1)} c C \d C  b + \phi^{(2)} c C \d C B + A^{(3)}_\mu c \d C \d X^\mu + A^{(4)}_\mu c \d C \Pi^\mu + A^{(5)}_\mu C \d C \d X^\mu + A^{(6)}_\mu C \d C \Pi^\mu\\
  + \phi^{(7)} c \d^2 c + \phi^{(8)} c \d^2 C + \phi^{(9)} C \d^2  c + \phi^{(10)} C \d^2 C + A^{(11)}_\mu c C \d^2 X^\mu + A^{(12)}_\mu c C \d \Pi^\mu\\
  + S^{(13)}_{\mu\nu} c C \d X^\mu \d X^\nu + T^{(14)}_{\mu\nu} c C \d X^\mu \Pi^\nu + S^{(15)}_{\mu\nu} c C \Pi^\mu \Pi^\nu,
\end{multline}
where $S^{(13)}_{\mu\nu}$ and $S^{(15)}_{\mu\nu}$ are symmetric, but $T^{(14)}_{\mu\nu}$ is a general $2$-tensor.

\subsubsection{Case $p=0$}
\label{sec:h2-p=0}

The cocycle conditions $d\Psi = 0$ allow us to solve for
\begin{equation}\label{eq:Zrel20}
  \begin{aligned}
    \phi^{(10)} &= -\tfrac32 \phi^{(1)}\\
    \phi^{(9)} &= \phi^{(8)} - \tfrac32 \phi^{(2)}\\
  \end{aligned}
  \qquad\qquad
  \begin{aligned}
    A^{(5)}_\mu &= 0\\
    A^{(6)}_\mu &= - A^{(3)}_\mu\\
    A^{(11)}_\mu &= A^{(3)}_\mu\\
    A^{(12)}_\mu &= A^{(4)}_\mu
  \end{aligned}
  \qquad\qquad
  \begin{aligned}
    S^{(13)}_{\mu\nu} &= \tfrac12 \eta_{\mu\nu} \phi^{(1)}\\
    T^{(14)}_{(\mu\nu)} &= \tfrac12 \eta_{\mu\nu} \phi^{(2)},\\
  \end{aligned}
\end{equation}
leaving free $\phi^{(1)}$, $\phi^{(2)}$, $A^{(3)}_\mu$, $A^{(4)}_\mu$,
$\phi^{(7)}$, $\phi^{(8)}$, $T^{(14)}_{[\mu\nu]}$ and $S^{(15)}_{\mu\nu}$.
Hence as representations of $\SO(25,1)$,
\begin{equation}
  \Zrel^2(0) \cong 5 \CC \oplus 2 \VV \oplus \sym{2}_0 \VV \oplus \ext{2} \VV.
\end{equation}
Since $\Brel^2(0) \cong 3\CC \oplus 2 \VV$, we see that
\begin{equation}
  \Hrel^2(0) \cong 2 \CC \oplus \sym{2}_0 \VV \oplus \ext{2} \VV.
\end{equation}
Since $\Crel^2(0) \cong 9 \CC \oplus 6 \VV \oplus 3 \sym{2}_0 \VV \oplus
\ext{2} \VV$, it follows from the Rank Theorem that
\begin{equation}
  \label{eq:Brel3-0-as-reps}
  \Brel^3(0) \cong 4 \CC \oplus 4 \VV \oplus 2 \sym{2}_0 \VV.
\end{equation}

\subsubsection{Case $p^2=0$ but $p\neq 0$}
\label{sec:h2-p2=0}

Since this is the most physically significant case, detailed
calculations are presented in Appendix \ref{app:H2rel-detailed-calc}.
The cocycle condition $d\Psi = 0$ translates into the following
equations (see Appendix \ref{app:H2rel-detailed-calc} for a detailed
derivation):
\begin{align} \label{eq:H2rel cocycle conditions}
  \begin{aligned}
    \phi^{(1)} &= 0\\
    A_\mu^{(5)} &= 0\\
    A_\mu^{(6)} &= A_\mu^{(3)} + i p_\mu \varphi\qquad (\exists \varphi)\\
    \phi^{(9)} &= \phi^{(8)} - \tfrac32 \phi^{(2)} - \tfrac{i}2 p \cdot A^{(4)}\\
    \phi^{(10)} &= \tfrac{i}2 p \cdot A^{(3)}\\
    A_\mu^{(11)} &= A_\mu^{(3)} + \tfrac{i}2 p_\mu \left(  \phi^{(2)} -\varphi \right) \\
    A_\mu^{(12)} &= A_\mu^{(4)} - i p_\mu \phi^{(7)} - i p^\nu S_{\mu\nu}^{(15)}\\
    S_{\mu\nu}^{(13)} &= i p_{(\mu} A_{\nu)}^{(3)} - \tfrac12 p_\mu p_\nu \varphi\\
    T_{(\mu\nu)}^{(14)} &= i p_{(\mu} A_{\nu)}^{(4)} + \tfrac12 \eta_{\mu\nu} \phi^{(2)}\\
    p^\nu T_{[\mu\nu]}^{(14)} &= p_\mu \left(\tfrac32 \phi^{(2)} + \tfrac{i}2 p \cdot A^{(4)} - 2 \phi^{(8)} - \varphi \right)\\
    p^\mu p^\nu S_{\mu\nu}^{(15)} &= 0.
  \end{aligned}
\end{align}
This results in free variables $\varphi$, $\phi^{(2)}$, $A_\mu^{(3)}$,
$A_\mu^{(4)}$, $\phi^{(7)}$, $\phi^{(8)}$, $T_{[\mu\nu]}^{(14)}$,
$S_{\mu\nu}^{(15)}$ subject to the last two equations above.

As mentioned at the beginning of this section, we view the space spanned by these free variables as a $K$-module. This makes it easier to determine and write down the space of cocycles by allowing us to choose a convenient
nonzero null momentum, but note that it breaks Lorentz covariance. We pick a Witt frame $(\be_+,\be_-,\be_i)$ with
$i=1,\dots,24$, and take the components of the momentum to satisfy
$p^+ = 1$ and $p^- = p^i = 0$ or, equivalently, $p_- = 1$, $p_+ = p_i
= 0$.  The two remaining cocycle conditions say that $S_{++}^{(15)}=0$,
$T_{[i+]}^{(14)} = 0$ and 
\begin{equation}
  T_{[-+]}^{(14)} = \tfrac32 \phi^{(2)} + \tfrac{i}2 p \cdot A^{(4)} - 2 \phi^{(8)} - \varphi,
\end{equation}
leaving the following free components of $S^{(15)}_{\mu\nu}$ and
$T^{(14)}_{[\mu\nu]}$:
\begin{equation}
  S^{(15)}_{+-},   S^{(15)}_{--},   S^{(15)}_{+i},   S^{(15)}_{-i},
  S^{(15)}_{ij}, T_{[-i]}^{(14)}, T^{(14)}_{[ij]}.
\end{equation}
The component $S^{(15)}_{ij}$ further decomposes into a trace part
(the transverse trace $\tr^\top S^{(15)}$) and a traceless part
denoted $S^{(15)}_{\left<ij\right>}$.  In summary, for such $p$,
\begin{equation}
  \Zrel^2(p) \cong 11 \CC \oplus 5 \VV^\top \oplus \sym{2}_0 \VV^\top \oplus \ext{2} \VV^\top
\end{equation}
as $K$-modules.  Since
\begin{equation}
  \Crel^2(p) \cong 31 \CC \oplus 14 \VV^\top \oplus 3 \sym{2}_0 \VV^\top \oplus \ext{2} \VV^\top,
\end{equation}
we see that
\begin{equation}
  \label{eq:B3rel-p-as-reps}
  \Brel^3(p) \cong 20 \CC \oplus 9 \VV^\top \oplus 2 \sym{2}_0 \VV^\top.
\end{equation}
From Section \ref{sec:h1-p2=0}, we know that $\Brel^2(p) \cong 10 \CC \oplus 4 \VV^\top$, which lets us deduce the
following.

\begin{proposition}\label{prop:Hrel2}
  For $p^2=0$, but $p\neq 0$,
  \begin{equation}
    \Hrel^2(p) \cong \CC \oplus \VV^\top \oplus \sym{2}_0 \VV^\top \oplus \ext{2}  \VV^\top.
  \end{equation}
\end{proposition}

\subsubsection{Summary}
\label{sec:hrel2-summary}

The relative cohomology at ghost number 2 is given by
\begin{equation}
  \Hrel^2(p) \cong
  \begin{cases}
    2 \CC \oplus \sym{2}_0 \VV \oplus \ext{2} \VV & p = 0 \\
    \CC \oplus \VV^\top \oplus \sym{2}_0 \VV^\top \oplus \ext{2} \VV^\top & p^2=0,~p\neq 0\\
    0 & p^2 \neq 0,
  \end{cases}
\end{equation}
where the first and second isomorphisms are of $\SO(25,1)$ and $K$-modules respectively.

We may choose cocycle representatives for these classes.  When $p^2 = 0$ but $p \neq 0$, the
cohomology is carried by the vertex operators
\begin{equation}
  \label{eq:fields-H2rel}
  \Psi = G_{\mu\nu} c C \left( \Pi^\mu \Pi^\nu - i p^{(\mu} \d \Pi^{\nu)} \right) \exp(i p \cdot X) + F_{\mu\nu} c C \d X^\mu \Pi^\nu \exp(i p\cdot X),
\end{equation}
where $G_{\mu\nu} = G_{\nu\mu} = S^{(15)}_{\mu\nu}$ and $F_{\mu\nu} =
- F_{\nu\mu} = T^{(14)}_{[\mu\nu]}$, subject to the equations
\begin{equation}
  \label{eq:field-eqns-H2rel}
  p^\mu p^\nu G_{\mu\nu} = 0 \qquad\text{and}\qquad  p^\mu F_{\mu\nu}
  = 0,
\end{equation}
and gauge symmetries
\begin{equation}
  \label{eq:gauge-symmetries-Hrel2}
  \delta G_{\mu\nu} = -\tfrac i6 \eta_{\mu\nu} p \cdot \chi + i p_{(\mu} \chi_{\nu)}  \qquad\text{and}\qquad
  \delta F_{\mu\nu} = i p_{[\mu} \omega_{\nu]},
\end{equation}
where $p \cdot \omega = 0$.
Crucially, $G_{\mu\nu}$ does not correspond to the graviton and dilaton as we know it from ordinary bosonic string theory. Furthermore, by looking at Proposition \ref{prop:Hrel2}, it may be tempting to claim that there exists a photon in the spectrum. However, this is not exactly the case. By analysing the $H$-module structure of $\Hrel^2(p)$ in Section \ref{sec:phys-interpr}, we elaborate on both these points in more detail. Specifically, we will see that both these peculiarities are a result of the non-unitarity of the spectrum.

\subsection{Calculating $\Hrel^3(p)$}
\label{sec:calculating-hrel3p}

The most general $\Psi \in \Crel^3(p)$ is given by
\begin{multline}
  \Psi = \phi^{(1)} c C \d^3 c + \phi^{(2)} c C \d^3 C + \phi^{(3)} c  \d C \d^2 c + \phi^{(4)} c \d C \d^2 C + \phi^{(5)} C \d C \d^2 c +  \phi^{(6)} C \d C \d^2 C\\
  + A^{(7)}_\mu c C \d^2 c \d X^\mu + A^{(8)}_\mu c C \d^2 C \d X^\mu +   A^{(9)}_\mu c C \d^2 c \Pi^\mu + A^{(10)}_\mu c C \d^2 C \Pi^\mu +   A^{(11)}_\mu c C \d C \d^2 X^\mu\\
  +   A^{(12)}_\mu c C \d C \d  \Pi^\mu + S^{(13)}_{\mu\nu} c C \d C \d X^\mu \d X^\nu + T^{(14)}_{\mu\nu} c C  \d C \d X^\mu \Pi^\nu + S^{(15)}_{\mu\nu} c C \d C \Pi^\mu \Pi^\nu,
\end{multline}
with $S^{(13)}$ and $S^{(15)}$ symmetric tensors.  The cocycle
condition $d\Psi = 0$ breaks up in to the following components:
\begin{align*}
  - \phi^{(4)} - \phi^{(5)} + \tfrac{i}2 p \cdot A^{(7)} + \tfrac{i}2 p \cdot A^{(10)}&= 0 & \tag{$c C \d^2 c \d^2 C$}\\
  2 \phi^{(2)} - \phi^{(4)} + \phi^{(5)} - \tfrac{i}3 p \cdot A^{(12)} - \tfrac16 \tr T^{(14)}&= 0 & \tag{$c C \d C \d^3 c$}\\
  \phi^{(6)} + \tfrac{i}3 p \cdot A^{(11)} + \tfrac16 \tr S^{(13)} &= 0 & \tag{$c C \d C \d^3 C$}\\
  i p_\mu \phi^{(5)} + A^{(8)}_\mu - A^{(11)}_\mu - \tfrac{i}2 p^\nu T^{(14)}_{\mu\nu}&= 0 & \tag{$c C \d C \d^2 c \d X^\mu$}\\
  i p_\mu \phi^{(3)} - A^{(7)}_\mu + A^{(10)}_\mu - A^{(12)}_\mu - i p^\nu S^{(15)}_{\mu\nu}&= 0 & \tag{$c C \d C \d^2c \Pi_\mu$}\\
  i p_\mu \phi^{(6)} + i p^\nu S^{(13)}_{\mu\nu}&= 0 & \tag{$c C \d C \d^2 C \d X^\mu$}\\
  i p_\mu \phi^{(4)} - A^{(8)}_\mu + A^{(11)}_\mu + \tfrac{i}2 p^\nu T^{(14)}_{\nu\mu} &= 0. & \tag{$c C \d C \d^2 C \Pi_\mu$}
\end{align*}
We may use the first five equations to solve for $\phi^{(5)}, \phi^{(2)}, \phi^{(6)}, A^{(11)}_\mu$ and $A^{(12)}_\mu$, respectively:
\begin{align}
  \label{eq:Z3rel-equations-partial}
  \begin{aligned}
    \phi^{(5)} &= - \phi^{(4)} + \tfrac{i}2 p \cdot A^{(7)} + \tfrac{i}2 p \cdot A^{(10)}\\
    \phisup{2} &= \phisup{4} - \tfrac{5i}{12} p \cdot \Asup{7} - \tfrac i{12} p \cdot \Asup{10} + \tfrac 16 p^\mu p^\nu \Ssup{15}_{\mu\nu} + \tfrac1{12} \tr \Tsup{14}\\
    \phi^{(6)} &=  -\tfrac{i}3 p \cdot  A^{(8)} - \tfrac16 \tr S^{(13)} - \tfrac16 p^\mu p^\nu T^{(14)}_{\mu\nu}\\
    A^{(11)}_\mu &= - i p_\mu \phi^{(4)} - \tfrac12 p_\mu p \cdot A^{(7)} + A^{(8)}_\mu - \tfrac{1}2 p_\mu p \cdot A^{(10)} - \tfrac{i}2 p^\nu T^{(14)}_{\mu\nu}\\
    A^{(12)}_\mu &=  i p_\mu \phi^{(3)} - A^{(7)}_\mu + A^{(10)}_\mu - i p^\nu S^{(15)}_{\mu\nu}.
  \end{aligned}
\end{align}
This leaves free $\phi^{(1)},\phi^{(3)},\phi^{(4)},A^{(7)}_\mu, A^{(8)}_\mu, A^{(9)}_\mu, A^{(10)}_\mu, S^{(13)}_{\mu\nu},T^{(14)}_{\mu\nu}, S^{(15)}_{\mu\nu}$, subject to the two last equations above: namely,
\begin{align}
  \label{eq:Z3rel-remaining-equations}
  \begin{aligned}
     p^\nu S^{(13)}_{\mu\nu} - \tfrac16 p_\mu \tr S^{(13)} &=
     \tfrac{i}3 p_\mu p \cdot  A^{(8)} + \tfrac16 p_\mu p^\lambda p^\nu T^{(14)}_{\lambda\nu}\\
      p^\nu T^{(14)}_{[\mu\nu]} &=\tfrac{i}2 p_\mu (p \cdot A^{(7)} + p \cdot A^{(10)}).
  \end{aligned}
\end{align}

\subsubsection{Case $p=0$}
\label{sec:case-p=0}

If $p=0$, then these equations are trivially satisfied and we are left with the following free components spanning $\Zrel^3(0)$:
\begin{equation}
  \phi^{(1)},\phi^{(3)},\phi^{(4)},A^{(7)}_\mu, A^{(8)}_\mu, A^{(9)}_\mu, A^{(10)}_\mu, S^{(13)}_{\mu\nu},T^{(14)}_{\mu\nu}, S^{(15)}_{\mu\nu}.
\end{equation}
In terms of representations of $\SO(25,1)$,
\begin{equation}
  \label{eq:Z3rel-0-as-reps}
  \Zrel^3(0) \cong 6 \CC \oplus 4 \VV \oplus 3 \sym{2}_0 \VV \oplus \ext{2} \VV.
\end{equation}
Taking equation~\eqref{eq:Brel3-0-as-reps} into account, we see that
$\Hrel^3(0) \cong \Hrel^2(0)$, as $\SO(25,1)$-modules and, using the
Rank Theorem, we deduce that
\begin{equation}
  \label{eq:B4rel-0-as-reps}
  \Brel^4(0) \cong 3 \CC \oplus 2 \VV.
\end{equation}

\subsubsection{Case $p^2=0$ but $p\neq 0$}
\label{sec:case-p2=0-but}

As we did for $\Hrel^2(p)$, we break $H$-equivariance by choosing an explicit null momentum $p^+ =1$ with $p^- = p^i = 0$ to study this case as an $K$-module, deferring the analysis of the $H$-module structure of $\Hrel^3(p)$ to Section \ref{sec:rep-theory-hrel3}. With this choice of momentum, the second equation breaks up as
\begin{equation}
    T^{(14)}_{[i +]}  = 0  \qquad\text{and}\qquad T^{(14)}_{[-+]} = \tfrac{i}2  (p \cdot A^{(10)} - p \cdot A^{(7)}),
\end{equation}
leaving free $T^{(14)}_{[ij]}$ and $T^{(14)}_{[-i]}$.  The first
equation breaks up as
\begin{equation}
  S^{(13)}_{++} = S^{(13)}_{+i} = 0 \qquad\text{and}\qquad S^{(13)}_{+-} = \tfrac{i}2 p \cdot A^{(8)} + \tfrac14 \tr^\top S^{(13)} + \tfrac14 p^\mu p^\nu T^{(14)}_{\mu\nu},
\end{equation}
leaving free $S^{(13)}_{ij}, S^{(13)}_{i-}, S^{(13)}_{--}$.  In
summary we have the following free components:
\begin{equation}
  \phi^{(1)},\phi^{(3)},\phi^{(4)},A^{(7)}_\mu, A^{(8)}_\mu, A^{(9)}_\mu, A^{(10)}_\mu, S^{(13)}_{ij}, S^{(13)}_{i-},S^{(13)}_{--}, T^{(14)}_{(\mu\nu)}, T^{(14)}_{[ij]}, T^{(14)}_{[-i]}, S^{(15)}_{\mu\nu}.
\end{equation}
In terms of representations of $K$,
\begin{equation}
  \label{eq:Z3rel-as-reps}
  \Zrel^3(p) \cong 21 \CC \oplus 10 \VV^\top \oplus 3 \sym{2}_0 \VV^\top \oplus \ext{2} \VV^\top.
\end{equation}
Taking equation~\eqref{eq:B3rel-p-as-reps} into account, we see that
$\Hrel^3(p) \cong \Hrel^2(p)$.  Using the Rank Theorem, we also see
that
\begin{equation}
  \label{eq:B4rel-p-as-reps}
  \Brel^4(p) \cong 10 \CC \oplus 4 \VV^\top.
\end{equation}

In summary, we have that for all $p$, $\Hrel^3(p) \cong \Hrel^2(p)$. When $p^2 = 0$ but $p\neq 0$, we stress that this is an isomorphism of $K$-modules. As we will see in Section \ref{sec:phys-interpr}, this is not quite true as $H$-modules. 

\subsection{Calculating $\Hrel^4(p)$}
\label{sec:calculating-hrel4p}

The general relative $4$-cocycle $\Psi \in \Crel^4(p)$ is given by
\begin{multline}
  \Psi = \phi^{(1)} c C \d^2 c \d^2 C + \phi^{(2)} c C \d C \d^3 c +  \phi^{(3)} c C \d C \d^3 C\\
  + A^{(4)}_\mu c C \d C \d^2 c \d X^\mu + A^{(5)}_\mu c C \d C \d^2 C \d X^\mu + A^{(6)}_\mu c C \d C \d^2 c \Pi^\mu + A^{(7)}_\mu c C \d C \d^2 C \Pi^\mu.
\end{multline}
The cocycle condition $d\Psi = 0$ reduces to
\begin{equation}
  p \cdot A^{(4)} = p \cdot A^{(7)},
\end{equation}
which is identically satisfied if $p=0$, in which case $\Zrel^4(0) =
\Crel^4(0)$, implying already that $\Brel^5(0) = 0$.  If $p^2=0$, but
$p\neq 0$, then $\Zrel^4(p) \cong \Crel^4(p) \ominus \CC$.  In
summary, we have that $\Crel^4(p) \cong \Crel^1(p)$ and, in addition,
\begin{equation}
  \label{eq:Z4rel-as-reps}
  \Zrel^4(p) \cong
  \begin{cases}
    3 \CC \oplus 4 \VV & p=0\\
    10\CC \oplus 4 \VV^\top & p^2=0,~p\neq 0,
  \end{cases}
\end{equation}
so that
\begin{equation}
  \label{eq:B5rel-as-reps}
  \Brel^5(p) \cong
  \begin{cases}
    0 & p = 0\\
    \CC & p^2=0,~p\neq 0.
  \end{cases}
\end{equation}
Since, as calculated above,
\begin{equation}
  \Brel^4(p) \cong
  \begin{cases}
    3 \CC \oplus 2 \VV  & p = 0\\
    10 \CC \oplus 4 \VV^\top & p^2=0,~p\neq 0,
  \end{cases}
\end{equation}
we conclude that $\Hrel^4(p) \cong \Hrel^1(p)$ for all $p$.

\subsection{Calculating $\Hrel^5(p)$}
\label{sec:calculating-hrel5p}

Here $\Crel^5(p) \cong \CC$ and $\Zrel^5(p) = \Crel^5(p)$, so that
from equation~\eqref{eq:B5rel-as-reps}, we have that $\Hrel^5(p) \cong
\Hrel^0(p)$ for all $p$.

\subsection{Summary}
\label{sec:hrel-summary}

We conclude this section with a summary of the $K$-module structure of $\Hrel^\bullet(p)$ . 

\begin{proposition}\label{prop:Hrel}
  The relative cohomology exhibits Poincaré duality $\Hrel^n(p) \cong
  \Hrel^{5-n}(p)$ as $\SO(25,1)$-modules when $p=0$ and as $K$-modules when $p \neq 0$. 
  Hence, by the above calculations,
  \begin{equation}
    \Hrel^n(0) \cong
    \begin{cases}
      \CC & n = 0,5\\
      2 \VV & n = 1,4\\
      2 \CC \oplus \sym{2}_0 \VV \oplus \ext{2} \VV & n = 2,3,
    \end{cases}
  \end{equation}
  as $\SO(25,1)$-modules
  and for $p^2 = 0$ but $p\neq 0$,
  \begin{equation}
    \Hrel^n(p) \cong
    \begin{cases}
      0 & n = 0,1,4,5\\
      \CC \oplus \VV^\top \oplus \sym{2}_0 \VV^\top \oplus \ext{2} \VV^\top & n = 2,3.
    \end{cases}
  \end{equation}
  as $K$-modules.
\end{proposition}

\section{The BRST cohomology}
\label{sec:brst-cohomology}

As was mentioned at the end of Section~\ref{sec:bosonic-ambi}, the
short exact sequence \eqref{eq:ses} of complexes induces a long exact
sequence in cohomology:
\begin{equation}
  \label{eq:les}
  \begin{tikzcd}
    \cdots \arrow[r] & \Hrel^{n-2}(p) \arrow[r] & \Hrel^n(p) \arrow[r] & \sH^n(p) \arrow[r] & \Hrel^{n-1}(p)
    \arrow[r] & \Hrel^{n+1}(p) \arrow[r] & \cdots
  \end{tikzcd}
\end{equation}
where $\sH^n(p)$ and $\Hrel^n(p)$ are the cohomologies of the absolute and
relative complexes, respectively, at momentum $p$.  The maps
$\Hrel^n(p) \to \sH^n(p)$ and $\sH^n(p) \to \Hrel^{n-1}(p)$ are
induced by the inclusion and $b_0$ respectively, whereas the
connecting homomorphism $\Delta: \Hrel^{n-1}(p) \to \Hrel^{n+1}(p)$ is given
by $\Delta = [d,c_0]$, whose action on a relative cocycle $\zeta \in \Zrel^\bullet(p)$ is
given by
\begin{equation}
  \Delta \zeta = (4 c_2 c_{-2} + 2 c_1 c_{-1}) \zeta.
\end{equation}
Note in particular that the long exact sequence \eqref{eq:les} is $H$-equivariant.

Since $\Hrel^{n<0}(p)=0$, taking $n=0$ in the long exact
sequence~\eqref{eq:les} gives that $\sH^0(p) \cong \Hrel^0(p)$, and
since $\Hrel^{n>5}(p)=0$, taking $n=6$ in the long exact
sequence~\eqref{eq:les} gives that $\sH^6(p) \cong \Hrel^5(p)$.

\begin{corollary}
  From Proposition~\ref{prop:on-shell} and the long exact sequence
  \eqref{eq:les}, we see that $\sH^\bullet(p) = 0$ for $p^2 \neq 0$.
\end{corollary}

\subsection{Case $p^2=0$ but $p\neq 0$}
\label{sec:nonzero-massless}

First, we summarise the $H$-module isomorphisms between relative and
absolute cohomology that we can infer at this point.

\begin{proposition}
  \label{prop:Summary-H}
  For $p^2=0$ but $p \neq 0$ and as $H$-modules,
  \begin{equation}
    \sH^n(p) \cong
    \begin{cases}
      0 & n = 0,1,5,6\\
      \Hrel^2(p) &   n=2\\
      \Hrel^3(p) & n=4,
    \end{cases}
  \end{equation}
  and $\sH^3(p)$ is an extension of $\Hrel^2(p)$ by $\Hrel^3(p)$.
\end{proposition}

\begin{proof}
  For $p^2 = 0$ but $p \neq 0$, Proposition~\ref{prop:Hrel} tells us
  that $\Hrel^n(p) = 0$ for $n = 0,\,1,\,4,$ and $5$. These zeroes
  reduce the long exact sequence \eqref{eq:les} to multiple short
  exact sequences. Firstly, we have
  \begin{equation}
    \begin{tikzcd}
      0 \arrow[r] & \sH^1(p) \arrow[r] & 0
    \end{tikzcd}
    \qquad\text{and}\qquad
    \begin{tikzcd}
      0 \arrow[r] & \sH^5(p) \arrow[r] & 0
    \end{tikzcd}
  \end{equation}
  so that $\sH^1(p) = \sH^5(p) = 0$.  We then have
  \begin{equation}
    \begin{tikzcd}
      0 \arrow[r] & \Hrel^2(p) \arrow[r] & \sH^2(p) \arrow[r] & 0 
    \end{tikzcd}
    \qquad\text{and}\qquad
    \begin{tikzcd}
      0 \arrow[r] & \sH^4(p) \arrow[r] & \Hrel^3(p) \arrow[r] & 0 
    \end{tikzcd}
  \end{equation}
  so that $\sH^2(p) \cong \Hrel^2(p)$  and $\sH^4(p) \cong \Hrel^3(p)$.
  Finally we have
  \begin{equation} \label{eq:Hrel3-H3-Hrel2-SES}
    \begin{tikzcd}
      0 \arrow[r] & \Hrel^3(p) \arrow[r] & \sH^3(p) \arrow[r] &
      \Hrel^2(p) \arrow[r] & 0,
    \end{tikzcd}
  \end{equation}
  which shows that $\sH^3(p)$ is an extension of $\Hrel^2(p)$ by
  $\Hrel^3(p)$.
\end{proof}

\begin{remark}
  Extensions of $\Hrel^2(p)$ by $\Hrel^3(p)$ are classified
  cohomologically.  Working at the level of the Lie algebra $\h$, the
  extension defines a class in $H^1(\h,\Hom(\Hrel^2(p), \Hrel^3(p)))$.
  From the structure of $\h$, as a direct product
  $\h \cong \fk \ltimes \fa$, with $\fk$ simple and
  $\fa$ abelian, it follows that $H^1(\h,\Hom(\Hrel^2(p), \Hrel^3(p)))
  \cong H^1(\fa,\Hom(\Hrel^2(p), \Hrel^3(p)))^\fk$.  We have not
  calculated the extension class, but in Appendix~\ref{app:splitting}
  we prove that $\dim H^1(\fa,\Hom(\Hrel^2(p), \Hrel^3(p)))^\fk = 3$,
  thereby concluding that we cannot deduce a priori that the extension
  is trivial.
\end{remark}

All $K$-module extensions split, so making use of the $K$-module
structure of relative cohomology at ghost numbers 2 and 3 as given by
Proposition~\ref{prop:Hrel} leads to the following corollary.

\begin{corollary}\label{cor:H-on-shell}  For $p^2=0$ but $p \neq 0$, we have
  \begin{equation}
    \sH^n(p) \cong
    \begin{cases}
      0 & n = 0,1,5,6\\
      \CC \oplus \VV^\top \oplus \sym{2}_0 \VV^\top \oplus \ext{2} \VV^\top &   n=2,4\\
      2 \CC \oplus 2 \VV^\top \oplus 2 \sym{2}_0 \VV^\top \oplus 2 \ext{2} \VV^\top &   n=3,
    \end{cases}
  \end{equation}
  as $K$-modules.
\end{corollary}

\subsection{Case $p=0$}
\label{sec:zero-momentum}

If $p=0$, $\sH^0(0) \cong \sH^6(0) \cong \CC$, but there are not
enough zeroes in the long exact sequence to be able to deduce the
cohomology.  One needs a more careful look at the connecting
homomorphism.  This is in principle possible, but assuming
(conjecturally) Poincaré duality for the absolute cohomology, we may
bootstrap the calculation of $\sH^\bullet(0)$ simply from the
knowledge of $\sH^n(0)$ for $n=0,1,2$.  Indeed, noticing that the
Euler--Poincaré principle relates $\sH^3(0)$ to $\sH^n(0)$ for
$n=0,1,2$, it is only necessary to calculate $\sH^1(0)$ and $\sH^2(0)$
to be able to conjecturally calculate $\sH^\bullet(0)$.

In more detail, the Euler--Poincaré principle says the following.  Let
$D$ be one of the isotypical representations occurring in the
cohomology; that is, $D$ is one of or one of $\CC$, $\VV$, $\sym{2}_0 \VV$
or $\ext{2} \VV$ in the case of $\SO(25,1)$-representations ($p=0$) or
one of $\CC$, $\VV^\top$, $\sym{2}_0 \VV^\top$ or $\ext{2} \VV^\top$ in the
case of $\SO(24)$-representations ($p^2=0$, but $p \neq 0$).  We
define $h_n(D)$ to be the multiplicity of $D$ in $\sH^n(p)$ and
$d_n(D)$ to be the multiplicity of $D$ in $\sC^n(p)$. Then the
Euler--Poincaré principle, which is a consequence of the Rank Theorem
applied to the differential, says that
\begin{equation}
  \sum_{n=0}^6 (-1)^n h_n(D) = \sum_{n=0}^6 (-1)^n d_n(D).
\end{equation}
It is easy to check that the RHS is zero and hence so is the LHS.
Poincaré duality (thus far conjectural) says that $h_n(D) =
h_{6-n}(D)$ and therefore
\begin{equation}
  h_3(D) = 2 (h_2(D) - h_1(D) + h_0(D)).
\end{equation}
Corollary~\ref{cor:H-on-shell} shows that this holds for $\sH^\bullet(p)$ for
$p^2=0$ and $p\neq 0$.  We will now calculate $\sH^1(0)$ and $\sH^2(0)$.

We observe that under $\sC^1(0) = \Crel^1(0) \oplus c_0 \Crel^0(0)$, a
cochain $\Psi \in \sC^1(0)$ can be written as $\Psi = \psi + c_0
\zeta$, where $\psi \in \Crel^1(0)$ and $\zeta \in \Crel^0(0)$.  The
cocycle condition $d\Psi = 0$ decomposes into two:
\begin{equation}
  d\zeta = 0 \qquad\text{and}\qquad d\psi + \Delta\zeta = 0,
\end{equation}
where $\Delta$ is the connecting homomorphism on cocycles.  As we saw
above, for $p=0$, $\Zrel^0(0) = \CC \ket{0}$ and $\Delta \ket{0} = 0$.  Therefore,
\begin{equation}
  \sZ^1(0) \cong \Zrel^1(0) \oplus \Zrel^0(0) \cong \CC \oplus 2 \VV.
\end{equation}
Since $\sC^1(0) \cong \Crel^1(0) \oplus \Crel^0(0) \cong 4 \CC \oplus
4 \VV$, it follows that
\begin{equation}
  \sB^2(0) \cong 3 \CC \oplus 2 \VV.
\end{equation}
Finally, since $\sB^1(0) = 0$, it follows that
\begin{equation}
  \sH^1(0) \cong \CC \oplus 2 \VV
\end{equation}
as a representation of $\SO(25,1)$.

Now under $\sC^2(0) = \Crel^2(0) \oplus c_0 \Crel^1(0)$, we again have
that $\Psi  = \psi + c_0 \zeta \in \sZ^2(0)$ if and only if $\zeta \in
\Zrel^1(0)$ and $d \psi + \Delta \zeta = 0$.  From
equation~\eqref{eq:Zrel10}, we see that $\zeta \in \Zrel^1(0)$ is
given by
\begin{equation}
  \zeta = \chi^{(1)}_\mu c \Pi^\mu + \chi^{(2)}_\mu\left( c \d X^\mu - C \Pi^\mu \right),
\end{equation}
so that one finds
\begin{equation}
  \Delta \zeta = \chi^{(2)}_\mu c C \d c  \Pi^\mu.
\end{equation}
This only affects one of the cocycle conditions in
equation~\eqref{eq:Zrel20}; namely,
\begin{equation}
  A^{(12)}_\mu = - A^{(4)}_\mu - \chi^{(2)}_\mu, 
\end{equation}
but does not alter the combinatorics.  Therefore we conclude that
\begin{equation}
  \sZ^2(0) \cong \Zrel^2(0) \oplus \Zrel^1(0) \cong 5 \CC \oplus 4 \VV
  \oplus \sym{2}_0 \VV \oplus \ext{2} \VV.
\end{equation}
Since $\sB^2(0) = 3 \CC \oplus 2 \VV$, it follows that
\begin{equation}
  \sH^2(0) \cong 2 \CC \oplus 2 \VV \oplus \sym{2}_0 \VV \oplus \ext{2} \VV.
\end{equation}

Finally, we arrive at the following corollary of
Proposition~\ref{prop:Hrel}, which remains conjectural subject to the
Poincaré duality of the absolute BRST cohomology at zero momentum.

\begin{conjecture}
  \label{sec:conj:zero-momentum}
  The BRST cohomology at momentum $p=0$ is given by
  \begin{equation}
    \sH^n(0) \cong
    \begin{cases}
      \CC & n=0,6\\
      \CC \oplus 2 \VV & n=1,5\\
      2 \CC \oplus 2 \VV \oplus \sym{2}_0 \VV \oplus \ext{2} \VV & n=2,4\\
      4 \CC \oplus 2 \sym{2}_0 \VV \oplus 2 \ext{2} \VV & n=3\\
    \end{cases}
  \end{equation}
  as a representation of $\SO(25,1)$.
\end{conjecture}

\section{Discussion of the spectrum}
\label{sec:phys-interpr}

In this section we discuss the spectrum in representation-theoretic
terms.  Since the cohomology resides in the massless sector, we expect
that we should interpret our results in terms of massless
representations of the Poincaré group.

Massless unitary irreducible representations (UIRs) of the Poincaré
group can all be obtained via the method of induced representations of
Wigner and Mackey (see, e.g.,
\cite[Appendix~A]{Figueroa-OFarrill:2023qty} and references therein).
In a nutshell, once we pick a nonzero $p$ with $p^2 = 0$, there is an
equivalence between the massless UIRs of the Poincaré group and the
UIRs of $\Stab(p)$ in the Lorentz group. The latter are called the
\emph{inducing representations}. We may also induce non-unitary
representations of the Poincaré group from finite-dimensional
representations of $\Stab(p)$, such as $\sH^\bullet(p)$. Indeed, we
let $\eO_p$ denote the orbit of $p \in V^*$ under the proper
orthochronous Lorentz group, which for a massless vector is, say, the
future lightcone. This is a homogeneous space of (the identity
component of) the Poincaré group $P$; although the translation
subgroup $T$ acts trivially. The representation $\sH^\bullet(p)$ of
$\Stab(p)$ can be made into a representation of
$S := \Stab(p) \ltimes T$ by having $T$ act via the unitary character
$\chi_p$ defined by $p$:
\begin{equation}
  \chi_p(\exp t) = e^{i \left<p,t\right>}
\end{equation}
for all $t \in T$.  This data defines a homogeneous vector bundle $P
\times_S  \sH^\bullet$ over $\eO_p$ whose sections each carry a
representation of the Poincaré group.  A natural question is whether
the representation induced from $\sH^\bullet$ is unitary.

Let us recall our notation: $H = \Stab(p)$, $\h = \stab(p)$ and
$K \subset H$ is a maximal compact subgroup whose Lie algebra
$\fk \subset \h$ splits the sequence~\eqref{eq:SES-stab-h-semidirect}.
A finite-dimensional unitary representation of $H$ must be such that
the normal subgroup with Lie algebra the abelian ideal $\fa$, being
non-compact, acts trivially, hence the representation factors through
a unitary representation of $K$. We will show below that this fails to
be the case: we will show in Section~\ref{sec:non-unitarity-hrel2p}
that $\Hrel^2(p)$ is not a unitary $H$-module and the $H$-module
isomorphism $\sH^2(p) \cong \Hrel^2(p)$ says that the same is true for
$\sH^2(p)$. In Section~\ref{sec:rep-theory-hrel3} we will show that
the same holds for $\Hrel^3(p)$. Indeed in
Section~\ref{sec:poinc-dual-betw} we show that
$\Hrel^3(p) \cong \Hrel^2(p)^*$ (Poincaré duality) and hence
$\Hrel^3(p)$ is not unitary. The same can be proved directly from
either of the isomorphisms~\eqref{eq:hrel3-as-h-mod-v1} or
\eqref{eq:hrel3-as-h-mod-v2}. Since $\sH^3(p)$ is an extension of
$\Hrel^2(p)$ by $\Hrel^3(p)$, it has a non-unitary submodule
isomorphic to $\Hrel^3(p)$ and hence is itself non-unitary.

\subsection{$\Hrel^2(p)$ as a $\Stab(p)$-module: non-unitarity}
\label{sec:rep-theory-hrel2}

As shown in Appendix~\ref{app:H2rel-detailed-calc}, the relative
cohomology $\Hrel^2(p)$ for $p^2 =0$ (but $p\neq 0$) has
representative cocycles
\begin{equation}
  G_{\mu\nu} c C (\Pi^\mu \Pi^\nu -i p^{(\mu} \d \Pi^{\nu)}) e^{i   p\cdot X}  \qquad\text{and}\qquad F_{\mu\nu} c C \d X^\mu \Pi^\nu  e^{i p \cdot X},
\end{equation}
where $G_{\mu\nu}$ is symmetric and obeys $p^\mu p^\nu G_{\mu\nu}= 0$
and $F_{\mu\nu}$ is skew-symmetric and obeys $p^\mu F_{\mu\nu}= 0$.
These cocycles are coboundaries if
\begin{equation}
  G_{\mu\nu} = \tfrac{i}6 \eta_{\mu\nu} p \cdot \chi - i p_{(\mu}  \chi_{\nu)}\qquad\text{and}\qquad F_{\mu\nu} = i p_{[\mu} \omega_{\nu]},
\end{equation}
where $\omega_\mu$ obeys $p \cdot \omega = 0$.  We quickly remark that
there exists a slightly different approach to studying $\Hrel^2(p)$ as
a $H$-module which involves using the above gauge freedom to gauge
away $\tr G$. This leads to the traceless part of $G_{\mu\nu}$
admitting the same gauge transformations as the graviton, as was
already shown by Berkovits and Lize \cite{Berkovits:2018jvm}. This is
rederived in Section \ref{sec:more-refined-version}.

\subsubsection{Action of $\h$ on cocycles}
\label{sec:action-h-2-cocycles}

Let us check that the action of $\h = \stab(p)$ on the cocycle coefficients
is the expected one.  We will only do this in this case for
illustration, but of course it is a general result.

The action of the zero mode of $\Lambda(\omega) = \frac12
\omega_{\mu\nu} (X^\mu \Pi^\nu)$, for $\omega_{\mu\nu} p^\nu = 0$,  on
these cocycles is given explicitly by
\begin{equation}
  \begin{split}
    [\Lambda(\omega), G_{\mu\nu} c C \left(\Pi^\mu \Pi^\nu - i p^{(\mu} \d \Pi^{\nu)}  \right) e^{i p \cdot X}]_1
    &= \left(\omega_\mu{}^\rho G_{\rho\nu} + \omega_\nu{}^\rho G_{\rho\mu}\right) c C \left(\Pi^\mu \Pi^\nu - i p^{(\mu} \d \Pi^{\nu)}  \right) e^{i p \cdot X}\\
    &= \delta_\omega G_{\mu\nu} c C \left(\Pi^\mu \Pi^\nu - i p^{(\mu} \d \Pi^{\nu)}  \right) e^{i p \cdot X}
  \end{split}
\end{equation}
and
\begin{equation}
  \begin{split}
    [\Lambda(\omega), F_{\mu\nu} c C \d X^\mu \Pi^\nu e^{i p \cdot X}]_1
    &= (\omega_\mu{}^\rho F_{\rho\nu} - \omega_\nu{}^\rho F_{\rho\mu})  c C \d X^\mu \Pi^\nu e^{i p \cdot X}\\
    &= \delta_\omega F_{\mu\nu} c C \d X^\mu \Pi^\nu e^{i p \cdot X},
  \end{split}
\end{equation}
which define $\delta_\omega G_{\mu\nu}$ and
$\delta_\omega F_{\mu\nu}$.  Notice that if $p^\mu F_{\mu\nu}= 0$,
then also $p^\mu \delta_\omega F_{\mu\nu} = 0$ and if
$F_{\mu\nu} = i p_{[\mu} \chi_{\nu]}$, then
$\delta_\omega F_{\mu\nu} = i p_{[\mu} \omega_{\nu]}{}^\rho
\chi_\rho$.  Similarly, $p^\mu p^\nu \delta_\omega G_{\mu\nu} = 0$
without having to impose $p^\mu p^\nu G_{\mu\nu} = 0$ and if
$G_{\mu\nu} = \tfrac{i}6 \eta_{\mu\nu} p \cdot \chi - i p_{(\mu}
\chi_{\nu)}$ then $\delta_\omega G_{\mu\nu}$ takes the same form with
$\chi_\mu$ replaced by $\omega_\mu{}^\rho \chi_\rho$.  These facts
simply restate that that $\delta_\omega$ is a chain map.  We see that
the infinitesimal $\stab(p)$ action on the cocycles is the standard
one:
\begin{equation}
  \delta_\omega G_{\mu\nu} = \omega_\mu{}^\rho G_{\rho\nu} +
  \omega_\nu{}^\rho G_{\rho\mu} \qquad\text{and}\qquad \delta_\omega
  F_{\mu\nu} = \omega_\mu{}^\rho F_{\rho\nu} - \omega_\nu{}^\rho
  F_{\rho\mu}.
\end{equation}

\subsubsection{A first version of $\Hrel^2(p)$ as an $H$-module}
\label{sec:hrel2-as-H-module-v1}

We now begin to identify $\Hrel^2(p)$ as an $H$-module, for
nonzero massless momentum $p$.  Since $p^\mu p^\nu G_{\mu\nu} = 0$,
the symmetric cocycle $G_{\mu\nu} c C (\Pi^\mu \Pi^\nu - i p^{(\mu} \d
\Pi^{\nu)}) e^{i p\cdot X}$ belongs to the $H$-module $p^\perp \odot
\VV \subset \sym{2}\VV$.  The coboundary $\delta G_{\mu\nu}$
belongs to the $H$-module $L_p \odot \VV$, so that this subspace of
$\Hrel^2(p)$ is isomorphic to the $324$-dimensional $H$-module
\begin{equation}
  \label{eq:h2rel-symmetric}
  \dfrac{p^\perp \odot \VV}{L_p \odot \VV}.
\end{equation}

Since $p^\mu F_{\mu\nu} = 0$, the skewsymmetric cocycle $F_{\mu\nu} c C
\d X^\mu \Pi^\nu e^{i p\cdot X}$ belongs to the $H$-submodule
$\ext{2}p^\perp \subset \ext{2} \VV$.  The cocycle is a coboundary if
it of the form $i p_{[\mu} \omega_{\nu]}$ with $\omega \in p^\perp$.
This means that it lies in the $H$-submodule $L_p \wedge p^\perp$.
This subspace of the cohomology corresponds to the $276$-dimensional
$H$-module
\begin{equation}
  \label{eq:h2rel-skewsymmetric}
  \dfrac{\ext{2}p^\perp}{L_p \wedge p^\perp} \cong \ext{2}\VV^\top,
\end{equation}
where the isomorphism was proved in
Lemma~\ref{lem:iso-h2rel-skewsymmetric}.  In summary, we have proved
that as an $H$-module, the second relative cohomology is given by
\begin{equation}
  \label{eq:Hrel2-as-H-mod}
  \Hrel^2(p) \cong \ext{2}\VV^\top \oplus \dfrac{p^\perp \odot \VV}{L_p \odot \VV}.
\end{equation}

\begin{proposition}
  \label{prop:H2rel-SES-duality}
  There is a short exact sequence of $H$-modules
  \begin{equation}
    \label{eq:H2rel-SES-duality}
    \begin{tikzcd}
      0 \arrow[r] & \ext{2}\VV^\top \oplus \sym{2}\VV^\top  \arrow[r] & \Hrel^2(p) \arrow[r] & \VV^\top \arrow[r] & 0.
    \end{tikzcd}
  \end{equation}
\end{proposition}

The proof will follow after the following two lemmas, which we will
reuse later.

\begin{lemma}
  \label{lem:SES-1-for-H2rel}
  There is a short exact sequence of $H$-modules
  \begin{equation}
    \begin{tikzcd}
      0 \arrow[r] & \sym{2}p^\perp \arrow[r] & p^\perp \odot \VV \arrow[r] & p^\perp \arrow[r] & 0,
    \end{tikzcd}
  \end{equation}
  and hence an $H$-module isomorphism
  \begin{equation}
    p^\perp \cong \dfrac{p^\perp\odot \VV}{\sym{2}p^\perp}.
  \end{equation}
\end{lemma}

\begin{proof}
  Contraction with $p$ gives an $H$-equivariant map $\imath_p \colon
  \sym{2} \VV \to \VV$, whose kernel is $\sym{2} p^\perp$.  The
  restriction of this map to $p^\perp \odot \VV \subset \sym{2} \VV$
  is also an $H$-equivariant map $p^\perp \odot \VV \to p^\perp$,
  whose kernel is still $\sym{2} p^\perp \subset p^\perp \odot \VV$.
  Furthermore it is surjective: choose $v \in \VV$ with
  $\left<p,v\right>=1$.  Then $w \in p^\perp$ can be written as
  $w = \imath_p(v \odot w)$.
\end{proof}

\begin{lemma}
  \label{lem:SES-2-for-H2rel}
  There is a short exact sequence of $H$-modules
  \begin{equation}
    \begin{tikzcd}
      0 \arrow[r] & L_p \odot p^\perp \arrow[r] & L_p \odot \VV \arrow[r] & L_p \arrow[r] & 0,
    \end{tikzcd}
  \end{equation}
  and hence an $H$-module isomorphism
  \begin{equation}
    L_p \cong \dfrac{L_p\odot \VV}{L_p \odot p^\perp}.
  \end{equation}
\end{lemma}

\begin{proof}
  Restricting $\imath_p \colon \sym{2} \VV \to \VV$ now to $L_p \odot
  \VV \subset \sym{2} \VV$, we have an $H$-equivariant map $L_p \odot
  \VV \to L_p$, whose kernel is $\left( L_p \odot \VV\right) \cap
  \sym{2} p^\perp = L_p \odot p^\perp$.  Again it is surjective since
  $L_p$ is one-dimensional and this map is not identically zero: there
  exist $v \in \VV$ with $\left<p,v\right> \neq 0$.
\end{proof}

\begin{proof}[Proof of the Proposition]
  The two short exact sequences in Lemmas~\ref{lem:SES-1-for-H2rel} and
  \ref{lem:SES-2-for-H2rel} fit together as the leftmost two columns
  in the following commutative diagram of $H$-modules:
  \begin{equation}
    \label{eq:3x3-comm-diag}
    \begin{tikzcd}
      & 0 \arrow[d] & 0 \arrow[d] & 0 \arrow[d] & \\
      0 \arrow[r] & L_p \odot p^\perp \arrow[d] \arrow[r] & \sym{2} p^\perp \arrow[r] \arrow[d] & \sym{2} \VV^\top \arrow[r] \arrow[d] & 0\\
      0 \arrow[r] & L_p \odot \VV \arrow[d] \arrow[r] & p^\perp \odot \VV \arrow[r] \arrow[d] & \dfrac{p^\perp \odot \VV}{L_p \odot \VV} \arrow[r] \arrow[d] & 0\\
      0 \arrow[r] & L_p \arrow[d] \arrow[r] & p^\perp \arrow[r]  \arrow[d]  & \VV^\top \arrow[r] \arrow[d] & 0\\
      & 0  & 0 & 0 &\\
    \end{tikzcd}
  \end{equation}
  where the rightmost column contains the canonical quotients.
  Exactness of the rightmost column shows that
  \begin{equation}
    \dfrac{(p^\perp \odot \VV)/(L_p \odot \VV)}{\sym{2}\VV^\top} \cong  \VV^\top.
  \end{equation}
  Putting the rightmost column together with isomorphism~\eqref{eq:h2rel-skewsymmetric},
  we obtain the following short exact sequence:
  \begin{equation}
    \label{eq:SES-Hrel2-unitary}
    \begin{tikzcd}
      0 \arrow[r] & \ext{2} \VV^\top \oplus \sym{2} \VV^\top \arrow[r]
      & \dfrac{\ext{2}p^\perp}{L_p \wedge p^\perp} \oplus \dfrac{p^\perp \odot \VV}{L_p \odot \VV} \arrow[r] & \VV^\top \arrow[r] & 0,
    \end{tikzcd}
  \end{equation}
  which, taking \eqref{eq:Hrel2-as-H-mod} into account, is precisely the sequence~\eqref{eq:H2rel-SES-duality}.
\end{proof}

\subsubsection{A refined version of $\Hrel^2(p)$ as an  $H$-module}
\label{sec:more-refined-version}

We may refine the description of $\Hrel^2(p)$ as an $H$-module by
decomposing the symmetric cocycle $G_{\mu\nu}$ into traceless and
trace parts:
\begin{equation}
  G_{\mu\nu} = G_{\left<\mu\nu\right>} + \tfrac1{26} \eta_{\mu\nu} \tr
  G, \quad\text{where}\quad \tr G = \eta^{\mu\nu} G_{\mu\nu},
\end{equation}
and we we have introduced the notation $G_{\left<\mu\nu\right>}$ for
traceless symmetric tensors:
$\eta^{\mu\nu}G_{\left<\mu\nu\right>} = 0$.  It follows that
$G_{\left<\mu\nu\right>}$ is a coboundary if it is of the form
$G_{\left<\mu\nu\right>} = -i p_{\left<\mu\right.}
\chi_{\left.\nu\right>}$ and the trace part is a coboundary if
$\tr G = i \tfrac{10}3 p \cdot \chi$.  This last equation can always
be satisfied by choosing $\chi$ appropriately (in particular, we can always choose $\chi$ such that $p \cdot \chi = -\tfrac{3}{10}i \tr G$), so that we can restrict
for the purposes of computing the cohomology to traceless
$G_{\mu\nu}$. In this case, we have $p \cdot \chi = 0$ together with the familiar gauge transformation of the graviton for $G_{\left<\mu\nu\right>}$. As mentioned earlier, this is in line with the findings of Berkovits and Lize, with the caveat that our ``field redefinition'' to obtain this familiar form is not the same as theirs (c.f. \cite[Section 2.3]{Berkovits:2018jvm}).

The traceless cocycles then live in the $H$-module
$(p^\perp \odot \VV)_0$, whereas the traceless coboundaries live in the
$H$-module $(L_p \odot \VV)_0 = L_p \odot p^\perp$.  This part of the
cohomology corresponds to the $324$-dimensional $H$-module
\begin{equation}
  \label{eq:h2rel-symmetric-traceless}
  \dfrac{\left( p^\perp \odot \VV \right)_0}{L_p \odot p^\perp}.
\end{equation}

We have a version of Lemma~\ref{lem:SES-1-for-H2rel} for the traceless
subspace.

\begin{lemma}
  \label{lem:SES-1-for-H2rel-traceless}
  There is a short exact sequence of $H$-modules
  \begin{equation}
    \begin{tikzcd}
      0 \arrow[r] & \sym{2}_0p^\perp \arrow[r] & \left( p^\perp \odot \VV \right)_0 \arrow[r] & p^\perp \arrow[r] & 0,
    \end{tikzcd}
  \end{equation}
  and hence an $H$-module isomorphism
  \begin{equation}
    p^\perp \cong \dfrac{(p^\perp\odot \VV)_0}{\sym{2}_0p^\perp}.
  \end{equation}
\end{lemma}

\begin{proof}
  We restrict the domain of the contraction $\imath_p \colon p^\perp \odot \VV \to
  p^\perp$ to the traceless subspace, resulting in $\imath_p \colon
  (p^\perp \odot \VV)_0 \to p^\perp$, which is still surjective: for any $u
  \in p^\perp$, choose $v \in V$ such that $\eta(u,v) = 0$ and
  $\left<p,v\right> =1$, so that $\imath_p (u \odot v) = u$.  The
  kernel of the map is the intersection of $\sym{2}p^\perp \cap
  (p^\perp \odot \VV)_0 = \sym{2}_0 p^\perp$.
\end{proof}

We may now restrict to traceless tensors in the top two rows of the
commutative diagram~\eqref{eq:3x3-comm-diag} and, using that $L_p
\odot p^\perp = (L_p \odot \VV)_0$ is already traceless and
Lemma~\ref{lem:SES-1-for-H2rel-traceless}, we arrive at the new diagram:
\begin{equation}
  \begin{tikzcd}
    & & 0 \arrow[d] & 0 \arrow[d] & \\
    0 \arrow[r] & L_p \odot p^\perp \arrow[d, equal] \arrow[r] & \sym{2}_0 p^\perp \arrow[r] \arrow[d] & \sym{2}_0 \VV^\top \arrow[r] \arrow[d] & 0\\
    0 \arrow[r] & (L_p \odot \VV)_0 \arrow[r] & (p^\perp \odot \VV)_0 \arrow[r] \arrow[d] & \dfrac{(p^\perp \odot \VV)_0}{(L_p \odot \VV)_0} \arrow[r] \arrow[d] & 0\\
    & & p^\perp \arrow[r, equal] \arrow[d] & p^\perp \arrow[d] & \\
    & & 0 & 0 &\\
  \end{tikzcd}
\end{equation}
and hence exactness of the rightmost column now gives
\begin{equation}
  \dfrac{(p^\perp \odot \VV)_0/(L_p \odot \VV)_0}{\sym{2}_0\VV^\top}  \cong p^\perp.
\end{equation}
That column, together with the
isomorphism~\eqref{eq:h2rel-skewsymmetric}, gives the following short
exact sequence of $H$-modules:
\begin{equation}
  \label{eq:SES-H2rel-v2}
  \begin{tikzcd}
    0 \arrow[r] & \ext{2}\VV^\top \oplus \sym{2}_0 \VV^\top \arrow[r] & \Hrel^2(p) \arrow[r] & p^\perp \arrow[r] & 0.
  \end{tikzcd}
\end{equation}
Note, however, that equation \eqref{eq:SES-H2rel-v2} does \emph{not}
imply that the maximal unitary $H$-submodule in $\Hrel^2(p)$ is
$\ext{2}\VV^\top \oplus \sym{2}_0 \VV^\top$, as one might be tempted
to conclude.

\subsubsection{Non-unitarity of $\Hrel^2(p)$}
\label{sec:non-unitarity-hrel2p}

Having determined $\Hrel^2(p)$ as an $H$-module, it remains to
investigate whether it is a unitary module.  Since $H$ is non-compact,
a finite-dimensional unitary module has to factor through a maximal
compact subgroup and hence the normal subgroup corresponding to the
ideal $\fa$ must act trivially.

The real $H$-module $V^\top = p^\circ/\ell_p$ is unitary since $\fa$
sends $p^\circ$ to $\ell_p$ and hence the action is trivial in the
quotient.  The $H$-invariant positive-definite inner product is
induced by $\eta$ was defined in
equation~\eqref{eq:euclidean-IP-on-V-transverse}.

Similarly, $\ext{2}V^\top$ and $\sym{2}V^\top$ are also unitary
representations, where the positive-definite $H$-invariant inner
product is induced from that on $V^\top$.  For $\ext{2}V^ \top$, it is
given by the Gram determinant
\begin{equation}
  \label{eq:gram-det}
  \bar\eta(\bar u_1 \wedge \bar v_1, \bar u_2 \wedge \bar v_2) =
  \det
  \begin{pmatrix}
    \eta(u_1,u_2) & \eta(u_1,v_2) \\
    \eta(v_1,u_2) & \eta(v_1,v_2)
  \end{pmatrix}
\end{equation}
and for $\sym{2}V^\top$ it is given by the analogous
\begin{equation}
  \bar\eta(\bar u_1 \odot \bar v_1, \bar u_2 \odot \bar v_2)=
  \eta(u_1,u_2) \eta(v_1,v_2) + \eta(u_1,v_2) \eta(u_2,v_1).
\end{equation}
Both of these inner products are clearly well-defined and
positive-definite.  The $H$-invariance follows from the following
facts.  First of all, we have the $H$-module isomorphism
\begin{equation}
  \ext{2}V^\top \cong \dfrac{\ext{2}p^\circ}{\ell_p \wedge p^\circ}
\end{equation}
and we see that $\fa$ sends $\ext{2}p^\circ$ to
$\ell_p \wedge p^\circ$, so it acts trivially in the quotient.
Similarly, we have the $H$-module isomorphism
\begin{equation}
  \sym{2}V^\top \cong \dfrac{\sym{2}p^\circ}{\ell_p \odot p^\circ}
\end{equation}
and again we see that  $\fa$ sends $\sym{2}p^\circ$ to
$\ell_p \odot p^\circ$, so it acts trivially in the quotient.

The complexification of a real orthogonal $H$-module
is a complex unitary $H$-module.  For example, $\VV^\top = \CC
\otimes_\RR V^\top$.  We define a hermitian inner product $h$ on
$\VV^\top$ by
\begin{equation}
  h( z_1 \otimes v_1, z_2 \otimes v_2) = \overline{z}_1 z_2 \bar\eta(v_1,v_2)
\end{equation}
for all $z_1,z_2 \in \CC$ and $v_1,v_2 \in V^\top$, and extending
additively.  It follows that if $H$ preserves the real bilinear inner
product $\bar\eta$ on $V^\top$, then it also preserves the complex
hermitian inner product $h$ on $\VV^\top$ provided that we extend the
representation complex linearly; so that $g \in H$ acts on $z \otimes
v$ by $g \cdot (z \otimes v) = z \otimes (g \cdot v)$.  The same holds
for $\ext{2}\VV^\top$ and $\sym{2}\VV^\top$, which are also unitary
$H$-modules.

\begin{proposition} \label{prop:Hrel2_not_unitary}
  $\Hrel^2(p)$ is not a unitary representation of $H$. Its maximal
  unitary submodule is isomorphic to $\ext{2}\VV^\top \oplus
  \sym{2}\VV^\top$.
\end{proposition}

\begin{proof}
  We have seen that $\ext{2}\VV^\top$ is a unitary $H$-module.  The
  other part of the cohomology is
  \begin{equation}
    \label{eq:h2rel-h-mods}
    \dfrac{p^\perp \odot \VV}{L_p \odot \VV}.
  \end{equation}
  The action of the ideal $\fa \subset \h$ on
  $p^\perp \odot \VV$ is given, for $u\in p^\circ$, $v \in p^\perp$
  and $w \in \VV$, by
  \begin{equation}
    (p^\sharp \curlywedge u) (v \odot w) = \eta(u,v) p^\sharp \odot w +
    \eta(u,w) v \odot p^\sharp - \eta(p^\sharp,w) v \odot w.
  \end{equation}
  The first two terms belong to $L_p \odot \VV$, but the last term
  belongs to $\sym{2} p^\perp$.  From
  equation~\eqref{eq:h2rel-h-mods}, the coboundaries correspond to
  the subrepresentation $L_p\odot \VV \subset p^\perp \odot
  \VV$. Since $\sym{2} p^\perp \not\subset L_p \odot \VV$, we conclude
  that $\fa$ does not act trivially on the
  cohomology and hence the cohomology is \emph{not} a unitary
  representation of $H$.
  
  To determine the maximal unitary submodule we must find the largest
  submodule of the cocycles which is mapped to coboundaries under
  $\fa$. It follows from the complex versions
  of the Hasse diagram~\eqref{eq:VsymV-as-H-mod} and the action
  \eqref{eq:action-of-ideal-on-sym2} of
  $\fa$ on the submodules of
  $\sym{2}\VV$ that the largest submodule of the
  module $p^\perp \odot \VV$ of $2$-cocycles which is mapped under
  $\fa$ to the module $L_p \odot \VV$ of
  $2$-coboundaries is $L_p \odot \VV + \sym{2}p^\perp$.  Its image in
  the cohomology is
  \begin{equation}
    \dfrac{L_p \odot \VV + \sym{2}p^\perp}{L_p \odot \VV} \cong
    \dfrac{\sym{2}p^\perp}{(L_p \odot \VV) \cap \sym{2}p^\perp} =
    \dfrac{\sym{2}p^\perp}{L_p \odot p^\perp} \cong \sym{2}\VV^\top,
  \end{equation}
  where the first isomorphism is the Second Isomorphism Theorem for
  modules \cite[Theorem~4(2), Section~10.2]{DummitFoote} and the last
  isomorphism is the complex version of the first isomorphism in
  Lemma~\ref{lem:iso-h2rel-symmetric-submodule}.  This shows that the
  maximal unitary submodule of $\Hrel^2(p)$ is isomorphic to
  $\ext{2}\VV^\top \oplus \sym{2}\VV^\top$.
\end{proof}

Since $\Hrel^2(p)$ is not unitary, neither is $\sH^2(p) \cong
\Hrel^2(p)$ by Proposition~\ref{prop:Summary-H}. If we identify
the absolute cohomology at ghost number 2 with the physical spectrum
of the ambitwistor string, as is usually the case, Proposition
\ref{prop:Hrel2_not_unitary} proves the non-unitarity of the spectrum.
This is in keeping with the results of Berkovits and Lize
\cite{Berkovits:2018jvm}. However, Proposition~\ref{prop:Summary-H}
tells us that the absolute cohomology at ghost number 4 is also
non-trivial and not isomorphic as $H$-modules to that at ghost number
2. Hence, there is no reason to rule out the possibility that $\sH^4(p)$
described correctly the ambitwistor string spectrum.  We therefore
study $\sH^4(p) \cong \Hrel^3(p)$ as a $H$-module in the next
subsection, but in Section \ref{sec:phys-interpr-spectr}, we explain
why we still think that $\sH^2(p)$ describes correctly the spectrum
and what information $\sH^4(p)$ is giving us.

\subsection{$\Hrel^3(p)$ as a $\Stab(p)$-module: non-unitarity and
  Poincaré duality}
\label{sec:rep-theory-hrel3}

In this section we interpret $\Hrel^3(p)$ as an $H$-module.  We show
that $\Hrel^2(p)$ and $\Hrel^3(p)$ are dual $H$-modules.

\subsubsection{Poincaré duality between $\Hrel^2(p)$ and $\Hrel^3(p)$}
\label{sec:poinc-dual-betw}

As shown in Appendix~\ref{app:hrel3-details}, $\Hrel^3(p)$ for $p^2 =
0$ (but $p\neq 0$) has representative cocycles
\begin{equation}
  G_{\mu\nu} c C \d C \d X^\mu \d X^\nu e^{i p \cdot X} - \tfrac1{20} \tr
  G \left( c C \d C \d X \cdot \d X e^{i p\cdot X} - C \d C \d^2 C e^ {i p \cdot X}\right),
\end{equation}
and
\begin{equation}
  F_{\mu\nu} c C \d C \d X^\mu \Pi^\nu e^{i p \cdot X},
\end{equation}
where $G_{\mu\nu} = G_{\nu\mu}$ obeys $p^\nu G_{\mu\nu}= 0$ and
$F_{\mu\nu} = -F_{\nu\mu}$ obeys $p^\nu F_{\mu\nu} = 0$.  The
corresponding coboundaries are
\begin{equation}
  \delta G_{\mu\nu} = p_\mu p_\nu \vartheta \qquad\text{and}\qquad
  \delta F_{\mu\nu} = i p_{[\mu} \omega_{\nu]},
\end{equation}
where $p \cdot \omega = 0$.

In terms of $H$-modules, $G \in \sym{2}\VV$ belongs to the kernel
of $\imath_p : \sym{2}\VV \to \VV$ and hence $G \in \sym{2} p^\perp$,
whereas $\delta G \in L_p \odot L_p$.  Similarly, $F \in \ext{2}\VV$ belongs to
the kernel of $\imath_p : \sym{2} \VV \to p^\perp$, so that $F \in
\ext{2}p^\perp$, whereas $\delta F \in L_p \wedge p^\perp$.  In other
words, we have an isomorphism of $H$-modules
\begin{equation}
  \label{eq:hrel3-as-h-mod-v1}
  \Hrel^3(p) \cong \dfrac{\ext{2}p^\perp}{L_p \wedge p^\perp}\oplus
\dfrac{\sym{2}p^\perp}{L_p \odot L_p}.
\end{equation}

\begin{proposition} \label{prop:Hrel3_not_unitary}
  $\Hrel^3(p)$ is not a unitary $H$-module and its maximal
  unitary submodule is isomorphic to $\ext{2}\VV^\top \oplus \VV^\top$.
\end{proposition}

\begin{proof}
  The non-unitarity follows from the fact that
  \begin{equation}
    \fa\colon \sym{2}p^\perp \to L_p \odot p^\perp \not\subset L_p \odot L_p,
  \end{equation}
  and hence $\fa$ does not act trivially on
  $\dfrac{\sym{2}p^\perp}{L_p \odot L_p}$.  The maximal unitary
  submodule contains $\dfrac{\ext{2}p^\perp}{L_p \wedge p^\perp} \cong
  \ext{2} \VV^\top$ together with the image in
  cohomology of the largest submodule of $\sym{2}p^\perp$ which is
  mapped to $L_p \odot L_p$ under the action of $\ell_p \curlywedge
  p^\circ$.  From the complex version of the Hasse
  diagram~\eqref{eq:VsymV-as-H-mod} it follows that this submodule is $L_p \odot p^\perp
  \subset \sym{2}p^\perp$, whose image in cohomology is
  \begin{equation}
    \dfrac{L_p \odot p^\perp}{L_p \odot L_p} \cong \VV^\top,
  \end{equation}
  where the isomorphism follows from the complex version of
  Proposition~\ref{prop:assoc-graded-symmetric-modules}.
\end{proof}

This result also follows as a consequence of Poincaré duality.

\begin{proposition}[Poincaré duality]
  \label{prop:poincare-duality}
  There is an $H$-module isomorphism $\Hrel^2(p) \cong \Hrel^3(p)^*$.
\end{proposition}

\begin{proof}
  This is equivalent to the existence of a non-degenerate
  $H$-equivariant pairing
  \begin{equation}
    \Hrel^2(p) \times \Hrel^3(p) \to \CC
  \end{equation}
  which we will now exhibit.  We have that
  \begin{equation}
    \Hrel^2(p) \cong \dfrac{\ext{2} p^\perp}{L_p \wedge p^\perp} \oplus \dfrac{\sym{2} p^\perp}{L_p \odot L_p}
  \end{equation}
  whereas
  \begin{equation}
    \Hrel^3(p) \cong \dfrac{\ext{2} p^\perp}{L_p \wedge p^\perp} \oplus \dfrac{p^\perp \odot \VV}{L_p \odot \VV}.
  \end{equation}
  It follows from Lemma~\ref{lem:iso-h2rel-skewsymmetric} that
  \begin{equation}
    \dfrac{\ext{2}p^\perp}{L_p \wedge p^\perp} \cong \ext{2} \VV^\top,
  \end{equation}
  which inherits a positive-definite $H$-invariant inner product from
  the euclidean inner product on $\VV^\top$ via the Gram
  determinant~\eqref{eq:gram-det}.  So we must now exhibit a
  non-degenerate pairing between
  \begin{equation}
     \dfrac{\sym{2} p^\perp}{L_p \odot L_p} \qquad\text{and}\qquad
     \dfrac{p^\perp \odot \VV}{L_p \odot \VV}.
  \end{equation}
  The complex-bilinear inner product on $\sym{2}\VV$ induced by $\eta$
  via the symmetric analogue of the Gram determinant gives a dual
  pairing
  \begin{equation}
    \sym{2}\VV \times \sym{2}\VV \to \CC,
  \end{equation}
  which we may restrict to a bilinear map
  \begin{equation}
    \sym{2} p^\perp \times p^\perp \odot \VV \to \CC.
  \end{equation}
  Explicitly, if $w_1,w_2,w_3 \in p^\perp$ and $v \in \VV$, this map
  sends
  \begin{equation}
    (w_1 \odot w_2, w_3 \odot v) \mapsto \eta(w_1,w_3) \eta(w_2,v) + \eta(w_2,w_3) \eta(w_1,v).
  \end{equation}
  Clearly if $w_1, w_2 \in L_p$, then this is zero for all $w_3 \in
  p^\perp$ and $v \in \VV$ and, similarly, if $w_3 \in L_p$, then this
  is zero for all $w_1,w_2 \in p^\perp$ and any $v \in \VV$.  Therefore
  we see that the above bilinear map descends to a bilinear map
  \begin{equation}
    \left<-,-\right> : \dfrac{\sym{2} p^\perp}{L_p \odot L_p} \times \dfrac{p^\perp \odot \VV}{L_p \odot \VV} \to \CC.
  \end{equation}
  We can show that this pairing is non-degenerate by explicitly
  writing it down in a basis.  A Witt frame $\be_+,\be_i,
  \be_-$ for $\VV$ gives a natural basis $\be_+ \odot \be_+, \be_+
  \odot \be_i$, $\be_i \odot \be_j$ (for $i\leq j$) for $\sym{2}p^\perp$
  and $\be_+\odot \be_+, \be_+ \odot \be_i, \be_+\odot \be_-, \be_i
  \odot \be_j, \be_i \odot \be_-$ (for $i\leq j$) for $p^\perp \odot
  \VV$. Letting bars denote the projection to the relevant quotient, a
  basis for $\sym{2}p^\perp/L_p \odot L_p$ is given by
  $\overline{\be_+ \odot \be_i}, \overline{\be_i \odot \be_j}$ (for 
  $i\leq j$), whereas a basis for $(p^\perp \odot \VV)/(L_p \odot
  \VV)$ is given by $\overline{\be_- \odot \be_i}, \overline{\be_i
    \odot \be_j}$ (for $i\leq j$).  The above pairing on these basis
  elements is given by
  \begin{equation}
    \begin{split}
      \left<\overline{\be_+ \odot \be_i}, \overline{\be_- \odot
          \be_j}\right> &= \eta(\be_+,\be_-) \eta(\be_i, \be_j) = \delta_{ij}\\
      \left<\overline{\be_i \odot \be_j}, \overline{\be_k \odot
          \be_\ell}\right> &= \eta(\be_i, \be_k) \eta(\be_j,\be_\ell)
      +  \eta(\be_i, \be_\ell) \eta(\be_j,\be_k) = \delta_{ik}\delta_{j\ell} +  \delta_{i\ell} \delta_{jk},
  \end{split}
\end{equation}
which is clearly non-degenerate.  We conclude that $\Hrel^2(p)$ and
$\Hrel^3(p)$ are dual $H$-modules.
\end{proof}

\subsubsection{Another cocycle representative}
\label{sec:other-cocycle-reps}

In Appendix~\ref{app:hrel3-details} we exhibit an alternative cocycle
representative for the symmetric cohomology in $\Hrel^3(p)$; namely,
\begin{equation}
  \label{eq:hrel3-symmetric-alt-cocycle}
  G_{\mu\nu} c C \d C \d X^\mu \d X^\nu e^{i p \cdot X} + A_\mu c C  \d (\d C \d X^\mu) e^{i p\cdot X}
\end{equation}
subject to $p^\nu G_{\mu\nu} = 0$ and $\tfrac12 \tr G + i p \cdot A =
0$ with coboundaries $\delta G_{\mu\nu} = i p_{(\mu} \omega_{\nu)}$
and $\delta A_\mu = \omega_\mu -i \alpha p_\mu$, with $p\cdot \omega
= 0$.  It is convenient to define $\chi \in p^\perp$ by $\chi_\mu :=
\omega_\mu - i \alpha p_\mu$.  This then has $\delta A_\mu = \chi_\mu$
and $\delta G_{\mu\nu} = i p_{(\mu} \chi_{\nu)} + \alpha p_\mu p_\nu$.

Representation-theoretically, the pair $(G, A)$ belongs
to the $H$-module $\sym{2}p^\perp \oplus \VV$.  Let $\tau \colon
\sym{2}p^\perp \oplus \VV \to \CC$ be the $H$-equivariant linear map
such that $\tau(G,A) = \tfrac12 \tr G + i p \cdot A$.  The expression
in~\eqref{eq:hrel3-symmetric-alt-cocycle} is a cocycle precisely when
$(G,A) \in \ker \tau$.  Let $\lambda \colon p^\perp \oplus L_p \to \sym{2}
p^\perp \oplus \VV$ be the $H$-equivariant linear map sending
$(\chi,\alpha p^\sharp) \mapsto (i p^\sharp \odot\chi + \alpha
p^\sharp \odot p^\sharp , \chi)$.  A cocycle $(G,A)$  is a coboundary
precisely when it belongs to $\im \lambda$.  Since $\tau \circ \lambda
= 0$, these two maps define a complex
\begin{equation}
  \begin{tikzcd}
    0 \arrow[r] & p^\perp \oplus L_p \arrow[r,"\lambda"] & \sym{2} p^\perp \oplus
    \VV \arrow[r, "\tau"] & \CC \arrow[r] & 0,
  \end{tikzcd}
\end{equation}
whose cohomology $\ker \tau/\im \lambda$ is isomorphic to the
symmetric part of $\Hrel^3(p)$.  We notice that
\begin{equation}
  \im \lambda \subset \sym{2}_0 p^\perp \oplus p^\perp \subset \ker \tau.
\end{equation}
We decompose $G$ into traceless and trace:
\begin{equation}
\sym{2}\VV \ni  G_{\mu\nu} = G_{\left<\mu\nu\right>} + \tfrac1{26}
\eta_{\mu\nu} \tr G \in \sym{2}_0 \VV \oplus \CC \eta
\end{equation}
and hence the symmetric part of $\Hrel^3(p)$ is given by
\begin{equation}
  \dfrac{\sym{2}_0 p^\perp \oplus p^\perp}{\lambda(p^\perp \oplus L_p)} \oplus
  \ker\left( \CC\eta \oplus \VV/p^\perp \xrightarrow{\tau'} \CC \right),
\end{equation}
where $\tau'(\tfrac1{26} \eta_{\mu\nu} \tr G, A \mod p^\perp) = \tfrac12 \tr
G + i p \cdot A$ is induced by the old map $\tau$, by modifying its
domain.  If $(G,A) \in \sym{2}_0 p^\perp \oplus p^\perp$, then we can
add the coboundary $\lambda(-A,0)$ to bring it to the form
$(G - i p^\sharp \odot A, 0) \in \sym{2}_0p^\perp$.  This still leaves
coboundaries in $L_p \odot L_p$.  The kernel of $\CC\eta \oplus
\VV/p^\perp \xrightarrow{\tau'} \CC$ is a trivial $H$-module since so
are the codomain and domain of $\tau'$.  Therefore we conclude with a
slightly different description of $\Hrel^3(p)$ as an $H$-module:
namely,
\begin{equation}
  \label{eq:hrel3-as-h-mod-v2}
  \Hrel^3(p) \cong \ext{2}\VV^\top \oplus \dfrac{\sym{2}_0
    p^\perp}{L_p \odot L_p} \oplus \CC.
\end{equation}
Of course, since $L_p \odot L_p \subset \sym{2}_0 p^\perp$, we have
that
\begin{equation}
  \dfrac{\sym{2} p^\perp}{L_p \odot L_p} =   \dfrac{\sym{2}_0 p^\perp
    \oplus \CC \bar\eta}{L_p \odot L_p}\cong  \dfrac{\sym{2}_0
    p^\perp}{L_p \odot L_p} \oplus \CC,
\end{equation}
showing that the two descriptions of $\Hrel^3(p)$ given by
equations~\eqref{eq:hrel3-as-h-mod-v2} and
\eqref{eq:hrel3-as-h-mod-v1} are of course isomorphic.

\subsection{Physical interpretation of the spectrum}
\label{sec:phys-interpr-spectr}

Let us summarise the key results from which we try to infer some physics,
keeping in mind that $p^2 = 0$, $p \neq 0$.
First, Corollary \ref{cor:H-on-shell} tells us that as $K$-modules,
the full BRST cohomology is given by
\begin{equation}
  \sH^n(p) \cong
  \begin{cases}
    0 & n = 0,1,5,6\\
    \CC \oplus \VV^\top \oplus \sym{2}_0 \VV^\top \oplus \ext{2} \VV^\top &   n=2,4\\
    2 \CC \oplus 2 \VV^\top \oplus 2 \sym{2}_0 \VV^\top \oplus 2 \ext{2} \VV^\top &   n=3.
  \end{cases}
\end{equation}
From the above, it is natural to interpret the cohomology at ghost
number $2$ to be the physical spectrum of the ambitwistor string, with
the cohomology at ghost number $3$ simply double-counting it, and ghost
number $4$ simply duplicating it. However, when studying BRST cohomology
as a $K$-module, we lose all information on how the non-compact part
of $H$ acts on the cohomology, and thus on the states in the spectrum.
Thus, to deduce the $H$-module structure of $\sH^\bullet(p)$, we
proceeded to study $\Hrel^2(p)$ and $\Hrel^3(p)$ as $H$-modules, as
in Proposition \ref{prop:Summary-H}. However, there are two key
differences compared to the $K$-module setting:
\begin{enumerate}
\item $\sH^4(p) \cong \sH^2(p)$ as $K$-modules, but $\sH^4(p)\cong
  \sH^2(p)^*$ as $H$-modules. This is an immediate consequence of the
  fact that $\Hrel^2(p) \cong \Hrel^3(p)^*$ as $H$-modules, while
  $\Hrel^2(p) \cong \Hrel^3(p)$ as $K$-modules.

\item $\sH^3(p) \cong 2\sH^2(p)$ as $K$-modules, but as $H$-modules,
  we have only shown that $\sH^3(p)$ is an extension of $\Hrel^2(p)$
  by $\Hrel^3(p)$.
\end{enumerate}
We understand the first difference as a generalised relationship
between $\sH^2(p)$ and $\sH^4(p)$ that accounts for the non-unitarity
(and thus non-self-duality) of $\sH^2(p)\cong\Hrel^2(p)$. But this
begs the following question: Which ghost number ($2$ or $4$) of 
the full BRST cohomology contains the physical states? 
Given that the ambitwistor string is a closed string,
we suggest that it is more natural to consider $\sH^2(p)$ as the
correct description of the spectrum of the ambitwistor string. Hence,
it is sensible to expect massless (c.f. Proposition \ref{prop:on-shell})
closed string states in the spectrum. But Proposition \ref{prop:Hrel2_not_unitary} 
tells us that the inducing representations for these: the graviton ($\sym{2}_0\VV^\top$), 
the Kalb--Ramond field ($\ext{2}\VV^\top$)
and the dilaton (the trace part of $\sym{2}\VV^\top$),
form the maximal unitary submodule $\ext{2}\VV^\top \oplus \sym{2}\VV^\top$
contained in $\sH^2(p)\cong \Hrel^2(p)$.  Furthermore,  taking
$\sH^2(p)$ to be the correct description  of the ambitwistor spectrum
extends the widely accepted notion that the spectrum contains the
massless closed bosonic string states. This does, however, imply that
the spectrum is strictly larger than just the space spanned by the
massless closed bosonic state, and that the spectrum is indeed
non-unitary, in line with the findings of \cite{Berkovits:2018jvm}. On
the other hand, Proposition \ref{prop:Hrel3_not_unitary} implies that
the maximal unitary submodule of $\sH^4(p) \cong \Hrel^3(p)$
corresponds to the inducing representation for the Poincaré UIRs
describing the Kalb--Ramond field ($\ext{2}\VV^T$) and the photon
($\VV^\top$). Not only is the latter an open string state, but the
inducing representation for the ``graviton+dilaton'' is not a
subrepresentation of BRST cohomology at ghost number $4$ at all.
For these reasons, we regard the BRST cohomology at ghost number $2$ 
to be the correct description of the physical spectrum of the bosonic
ambitwistor string.

Closely related is the second difference compared to the $K$-module
setting. Notice that the double-counting that occurs in $\sH^3(p)$ in
this case is a direct consequence of the unitarity and self-duality of
each of the $K$-modules that appear in $\sH^2(p)$. Specifically, the
short exact sequence \eqref{eq:Hrel3-H3-Hrel2-SES} splits when viewed
as $K$-modules and we are able to conclude that
\begin{equation}
  \sH^3(p) \cong \Hrel^2(p) \oplus \Hrel^3(p) \cong \sH^2(p) \oplus \sH^4(p).
\end{equation}
Thus, the fact that $\sH^3(p) \cong 2\sH^2(p)$ is a corollary of the
unitarity (and self-duality) of $\sH^2(p)$ which leads to the
$K$-module isomorphism $\sH^4(p) \cong \sH^2(p)$. Consider now the
$H$-module setting, where $\sH^4(p) \cong \sH^2(p)^*$ instead and
non-unitarity of $\sH^2(p)$ implies that $\sH^2(p) \ncong
\sH^2(p)^*$. The most naive and intuitive generalisation of the
double-counting present in the $K$-module case leads to the following
conjecture, for which we must confess to having no other evidence:
\begin{conjecture}
    The short exact sequence \eqref{eq:Hrel3-H3-Hrel2-SES} splits as a $H$-module, and thus,
    \begin{equation}
        \sH^3(p) \cong \Hrel^2(p) \oplus \Hrel^2(p)^* \cong \sH^2(p) \oplus \sH^2(p)^*.
    \end{equation}
\end{conjecture}

\section{Conclusions and outlook}
\label{sec:conclusions-outlook}

In this paper, we have revisited the BRST cohomology of the bosonic ambitwistor string 
from an algebraic lens, using the groundwork laid in \cite{Figueroa-OFarrill:2024wgs}, as the semi-infinite cohomology of its worldsheet symmetry algebra (i.e., BMS$_3$)
with values in the BMS$_3$-module spanned by its matter sector. 
Our main results
involve the determination of the (relative) cohomology at momentum $p$
as a representation of the stabiliser $\Stab(p)$ of that momentum, as
befits the interpretation of the cohomology as inducing
representations for representations of the Poincaré group.  This
allows us to prove a number of structural results and in particular
understand the precise non-unitarity of the representation as well as
the relation between the cohomologies at different ghost numbers.  In
summary, we interpret $\sH^2(p)$ as the physical spectrum of the
bosonic ambitwistor string and thus we conclude that the spectrum is
non-unitary and contains the massless level of the closed bosonic
string as the maximal unitary submodule.  An immediate algebraic
consequence of this non-unitarity is that $\sH^4(p) \cong \sH^2(p)^*$
as $H$-modules. This generalises the expected relation
$\sH^4(p) \cong \sH^2(p)$ to the setting of non-unitary (and thus
non-self-dual) $H$-modules. Extending this generalisation naively, we
conjecture that $\sH^3(p) \cong \sH^2(p) \oplus \sH^2(p)^*$.

There are a number of open problems we would like to solve:
\begin{itemize}
\item Elucidating whether or not $\sH^3(p)$ is a trivial extension of
  $\Hrel^2(p)$ by $\Hrel^3(p)$.
\item Constructing an algebraic proof of Poincaré duality for both the relative and
  absolute cohomologies. In principle, this could be done by
  constructing a Hodge star operator using the volume forms $cC\d C
  \d^2 c \d^2 C \in \Crel^5(p)$ an $c C \d c \d C \d^2 c \d^2 C \in
  \sC^6(p)$, but in practice, it is quite tedious.
\item Applying our algebraic approach to calculating the spectrum of the
  Type II ambitwistor string.
\end{itemize}
We may turn to some of these problems in the future.

\section*{Acknowledgments}

We would like to thank Tim Adamo for insightful discussions on the bosonic 
ambitwistor string and for careful readings of an earlier draft of
this manuscript and of the present draft.  In addition, JMF would like
to thank Emil Have and Niels Obers for many enlightening and
entertaining discussions about non-lorentzian string theory and an
enjoyable collaboration leading up to \cite{Figueroa-OFarrill:2025njv}.

\begin{appendices}
\section{The action of the BRST differential on $\Crel^\bullet(p)$}
\label{app:calculation of d on Crel}
In this appendix, we give explicit expressions for the basis fields of $\Crel^\bullet(p)$ and the action of the BRST differential on each of them. We divide our calculations by ghost number and distinguish the basis fields via the representation of the Lorentz group that they carry.

$\Crel^0(p)$ is one-dimensional and spanned by $W_p$, and its BRST differential is
given by
\begin{equation}
  dW_p = i p_\mu c \d X^\mu W_p - i p^\mu C \Pi_\mu W_p + \tfrac12 p^2 \d C W_p.
\end{equation}
From now on, we will omit $W_p$ and tacitly assume that all field expressions
end in a $W_p$.  In these abbreviated conventions, the above equation
is then simply
\begin{equation}
  \label{eq:dvac}
  dW_p = i p_\mu c \d X^\mu - i p^\mu C \Pi_\mu + \tfrac12 p^2 \d C.
\end{equation}

$\Crel^1(p)$ is spanned by a basis of scalar and vector monomials, with the following action of $d$ on them:
\begin{align}
  \begin{split} \label{eq:d-on-scalars-gh-no-1}
    d(\d C) &= c \d^2 C + C \d^2 c + i p_\mu c \d C \d X^\mu - i p^\mu C \d C \Pi_\mu \\
    d(c C b) &= 2 c C \d C B + c C \d X^\mu \Pi_\mu + i p_\mu c C \d^2 X^\mu + \tfrac32 c \d^2 C - \tfrac32 C \d^2 c + \tfrac12 p^2 c C \d C b\\
    d(c C B) &= - \tfrac12 \eta^{\mu\nu} c C \Pi_\mu \Pi_\nu - i p^\mu c C \d \Pi_\mu + \tfrac32 c \d^2 c + \tfrac12 p^2 c C \d C B.
  \end{split}\\[1ex]
  \begin{split}     \label{eq:d-on-vectors-gh-no-1}
    d(c \d X^\mu) &= c C \d \Pi^\mu + i p^\nu c C \d X^\mu \Pi_\nu + c \d C \Pi^\mu + \tfrac{i}2 p^\mu c \d^2 C - \tfrac12 p^2 c \d C \d X^\mu \\
    d(C \d X^\mu) &= i p_\nu c C \d X^\nu \d X^\mu + c C \d^2 X^\mu + c \d C \d X^\mu + C \d C \Pi^\mu + \tfrac{i}2 p^\mu C \d^2 C - \tfrac12 p^2 C \d C \d X^\mu\\
    d(c \Pi_\mu)&= i p^\nu c C \Pi_\mu \Pi_\nu -\tfrac{i}2 p_\mu c \d^2 c - \tfrac12 p^2 c \d C \Pi_\mu\\
    d(C \Pi_\mu) &= i p_\nu c C \d X^\nu \Pi_\mu + c C \d \Pi_\mu + c \d C \Pi_\mu - \tfrac{i}2 p_\mu C \d^2 c - \tfrac12 p^2 C \d C \Pi_\mu.
  \end{split}
\end{align}
There are no tensors at ghost number $1$.  

$\Crel^2(p)$ admits a basis consisting of six scalars,\footnote{At this and higher
  ghost number there are additional scalars obtained by taking traces
  of symmetric tensors.  Since $d$ commutes with taking trace, the
  action of $d$ on them can be read from that on tensors by taking
  trace and so we will not list them separately here.} six vectors, and three tensors, two of which are symmetric. The BRST differential acts on these basis vectors as follows:
\begin{align}
    \begin{split} \label{eq:d-on-scalars-gh-no-2}
    d (c \d^2 c) &= -\tfrac12 p^2 c \d C \d^2 c + i p^\mu c C \d^2 c \Pi_\mu\\
    d (c \d^2 C) &= -\tfrac12 p^2 c \d C \d^2 C + i p^\mu c C \d^2 C \Pi_\mu - c\d C \d^2c - c C \d^3c\\
    d (C \d^2 c) &= -\tfrac12 p^2 C \d C \d^2 c + i p_\mu c C \d^2 c \d X^\mu + c \d C \d^2 c + c C \d^3 c\\
    d (C \d^2 C) &= -\tfrac12 p^2 C \d C \d^2 C + i p_\mu c C \d^2 C \d X^\mu + c \d C \d^2 C + c C \d^3 C - C \d C \d^2 c\\
    d (c C \d C b) &= -c C \d C \d X^\mu \Pi_\mu - i p_\mu c C \d C \d^2X^\mu + \tfrac32 c \d C \d^2 C - \tfrac23 c C \d^3C - \tfrac32 C \d C \d^2 c\\
    d (c C \d C B) &= \tfrac12 \eta^{\mu\nu} c C \d C \Pi_\mu \Pi_\nu + i p^\mu c C \d C \d\Pi_\mu + \tfrac32 c \d C \d^2 c - \tfrac23 c C \d^3 c,
  \end{split}\\[1ex]
  \begin{split} \label{eq:d-on-vectors-gh-no-2}
    d ( c C \d^2 X^\mu) &= \tfrac12 p^2 c C \d C \d^2 X^\mu - 2 c C \d C \d \Pi^\mu - c C \d^2 C \Pi^\mu - \tfrac{i}3 p^\mu c C \d^3 C + c C \d^2 c \d X^\mu\\
    d ( c C \d \Pi_\mu) &= \tfrac12 p^2 c C \d C \d \Pi_\mu + c C \d^2 c \Pi_\mu + \tfrac{i}3 p_\mu c C \d^3 c\\
    d ( c \d C \d X^\mu) &=  c C \d C \d \Pi^\mu + i p^\nu c C \d C \d X^\mu \Pi_\nu - \tfrac{i}2 p^\mu c \d C \d^2 C - c C \d^2 c \d X^\mu\\
    d ( C \d C \d X^\mu) &= i p_\nu c C \d C \d X^\mu \d X^\nu - \tfrac{i}2 p^\mu C \d C \d^2 C + c C \d C \d^2 X^\mu + c C \d^2 C \d X^\mu\\
    d ( c \d C \Pi_\mu) &= i p^\nu c  C \d C \Pi_\mu \Pi_\nu - c C \d^2 c \Pi_\mu + \tfrac{i}2 p_\mu c \d C \d^2 c\\
    d ( C \d C \Pi_\mu) &= i p_\nu c C \d C \d X^\nu \Pi_\mu + c C \d C \d \Pi_\mu + c C \d^2 C \Pi_\mu + \tfrac{i}2 p_\mu C \d C \d^2 c,
  \end{split}\\[1ex]
  \begin{split} \label{eq:d-on-tensors-gh-no-2}
    d(c C \d X^\mu \d X^\nu) &= \tfrac12 p^2 c C \d C \d X^\mu \d X^\nu - i c C \d^2 C p^{(\mu} \d X^{\nu)} - 2 c C \d C \d X^{(\mu} \Pi^{\nu)} - \tfrac16 \eta^{\mu\nu} c C \d^3 C\\
    d(c C \d X^\mu \Pi_\nu) &= -c C \d C \Pi^\mu \Pi_\nu + \tfrac12 p^2 c \d C \d X^\mu \Pi_\nu - \tfrac{i}2 p^\mu c C \d^2 C \Pi_\nu + \tfrac{i}2 p_\nu c C \d^2 c \d X^\mu + \tfrac16 \delta^\mu_\nu c C \d^3 c\\
    d(c C \Pi_\mu \Pi_\nu) &= i c C \d^2 c p_{(\mu} \Pi_{\nu)} + \tfrac12 p^2 c C \d C \Pi_\mu \Pi_\nu.
  \end{split}
\end{align}

$\Crel^3(p)$ is also spanned by six scalars, six vectors, and three tensors with two of them being symmetric.
\begin{align}      
  \begin{split} \label{eq:d-on-scalars-gh-no-3}
    d(c C \d^3 c) &= \tfrac12 p^2 c C \d C \d^3 c\\
    d(c C \d^3 C) &= \tfrac12 p^2 c C \d C \d^3 C + 2 c C \d C \d^3 c\\
    d(c \d C \d^2 c) &= i p^\mu c C \d C \d^2 c \Pi_\mu\\
    d(c \d C \d^2 C) &= i p^\mu c C \d C \d^2 C \Pi_\mu - c C \d^2 c \d^2 C - c C \d C \d^3 c\\
    d(C \d C \d^2 c) &= i p_\mu c C \d C \d^2 c \d X^\mu - c C \d^2 c \d^2 C + c C \d C \d^3 c\\
    d(C \d C \d^2 C) &= i p_\mu c C \d C \d^2 C \d X^\mu + c C \d C \d^3 C,    
  \end{split}\\[1ex]
  \begin{split}     \label{eq:d-on-vectors-gh-no-3}
    d (c C \d^2 c \d X^\mu) &= - c C \d C \d^2 c \Pi^\mu + \tfrac12 p^2 c C \d C \d^2 c \d X^\mu + \tfrac{i}2 p^\mu c C \d^2 c \d^2 C\\
    d (c C \d^2 C \d X^\mu) &= - c C \d C \d^2 C \Pi^\mu + \tfrac12 p^2 c C \d C \d^2 C \d X^\mu + c C \d C \d^2 c \d X^\mu\\
    d (c C \d^2 c \Pi_\mu) &= \tfrac12 p^2 c C \d C \d^2 c \Pi_\mu\\
    d (c C \d^2 C \Pi_\mu) &= \tfrac12 p^2 c C \d C \d^2 C \Pi_\mu + c C \d C \d^2 c \Pi_\mu + \tfrac{i}2 p_\mu c C \d^2 c \d^2 C\\
    d (c C \d C \d^2 X^\mu) &= c C \d C \d^2 C \Pi^\mu + \tfrac{i}3 p^\mu c C \d C \d^3 C - c C \d C \d^2 c \d X^\mu\\
    d (c C \d C \d \Pi_\mu) &= - c C \d C \d^2 c \Pi_\mu - \tfrac{i}3 p_\mu c C \d C \d^3 c,
  \end{split}\\[1ex]
  \begin{split}     \label{eq:d-on-tensors-gh-no-3}
    d (c C \d C \d X^\mu \d X^\nu) &= i c C \d C \d^2 C p^{(\mu} \d X^{\nu)} + \tfrac16 \eta^{\mu\nu} c C \d C \d^3 C\\
    d (c C \d C \d X^\mu \Pi_\nu) &= \tfrac{i}2 p^\mu c C \d C \d^2 C \Pi_\nu - \tfrac{i}2 p_\nu c C \d C \d^2 c \d X^\mu - \tfrac16 \delta^\mu_\nu c C \d C \d^3 c\\
    d (c C \d C \Pi_\mu \Pi_\nu) &= -i c C \d C \d^2 c p_{(\mu} \Pi_{\nu)}.
  \end{split}
\end{align}

$\Crel^4(p)$ is spanned by three scalars and four vectors, just like $\Crel^1(p)$.
\begin{align}
  \begin{split} \label{eq:d-on-scalars-gh-no-4}
    d (c C \d^2 c \d^2 C) &= \tfrac12 p^2 c C \d C \d^2 c \d^2 C\\
    d (c C \d C \d^3 c) &= 0\\
    d (c C \d C \d^3 C) &= 0,
  \end{split} \\[1ex]
  \begin{split} \label{eq:d-on-vectors-gh-no-4}
    d(c C \d C \d^2 c \d X^\mu) &= - \tfrac{i}2 p^\mu c C \d C \d^2 c \d^2 C\\
    d(c C \d C \d^2 C \d X^\mu) &= 0\\
    d(c C \d C \d^2 c \Pi_\mu) &= 0\\
    d(c C \d C \d^2 C \Pi_\mu) &= -\tfrac{i}2 p_\mu c C \d C \d^2 c \d^2 C.
  \end{split}
\end{align}

Finally, $\Crel^5(p)$ is spanned the scalar $c C \d C \d^2 c
\d^2 C$, which is, of course, a cocycle. At this point, it is worth
nothing that $\Crel^{n}(p) \cong \Crel^{5-n}(p)$ as
$\Stab(p)$-modules, which is strongly suggestive of Poincaré duality
holding in cohomology too.

\section{Detailed calculation of $\Hrel^2(p)$ for $p^2=0,\ p\neq 0$}
 \label{app:H2rel-detailed-calc}

Here, we present the calculation of $\Hrel^2(p)$ as a $K$-module in great pedagogical
detail. Throughout this section, the momentum $p$ is non-zero and
satisfies $p^2 = \eta_{\mu\nu} p^\mu p^\nu = 0$. Recall that a general
relative 2-cochain $\Psi \in \Crel^2(p)$ is given by
\begin{multline} \label{eq:gen_C2rel_elem}
  \Psi = \phi^{(1)} c C \d C  b + \phi^{(2)} c C \d C B + A^{(3)}_\mu c \d C \d X^\mu + A^{(4)}_\mu c \d C \Pi^\mu + A^{(5)}_\mu C \d C \d X^\mu + A^{(6)}_\mu C \d C \Pi^\mu\\
  + \phi^{(7)} c \d^2 c + \phi^{(8)} c \d^2 C + \phi^{(9)} C \d^2  c + \phi^{(10)} C \d^2 C + A^{(11)}_\mu c C \d^2 X^\mu + A^{(12)}_\mu c C \d \Pi^\mu\\
  + S^{(13)}_{\mu\nu} c C \d X^\mu \d X^\nu + T^{(14)}_{\mu\nu} c C \d X^\mu \Pi^\nu + S^{(15)}_{\mu\nu} c C \Pi^\mu \Pi^\nu.
\end{multline}
Also note that each term implicitly ends with a $W_p$ on the right
hand side, where $W_p$ is the field corresponding to the vacuum state
$\ket{p}$ with momentum $p^\mu$.

Acting the BRST differential (i.e., zero mode of the BRST current
\eqref{eq:BRST current}) $d$ on each term in $\Psi$ gives a linear
combination of the basis vectors of $\Crel^3(p)$. Hence, the cocycle
condition $d\Psi = 0$ sets the coefficients of each basis vector to
zero, thereby giving us the cocycle equations The resulting set of
equations on these coefficients are summarised in Table
\ref{tab:Z2rel_calc}.

\begin{table}[h]
\centering
\renewcommand{\arraystretch}{1.2}
\caption{Summary of the computation of $d\Psi$ using \eqref{eq:d-on-scalars-gh-no-2}, \eqref{eq:d-on-vectors-gh-no-2} and \eqref{eq:d-on-tensors-gh-no-2}. Hence, the cocycle condition $d\Psi = 0$ is satisfied when every equation equals zero.}
\label{tab:Z2rel_calc}
    \begin{tabular}{|c|c|c L|}
        \toprule
       \multicolumn{2}{|c|}{Basis vector of $\Crel^3(p)$} & \multicolumn{2}{|c|}{Coefficients}\\
       \midrule
     \multirowcell{6}{Scalars} & $cC \d^3 c$ & $-\phi^{(8)} + \phi^{(9)} - \tfrac{2}{3} \phisup{2} + \tfrac{i}{3} p \cdot A^{(12)} + \tfrac16 \tr T^{(14)}$ & eq:C3rel_scalar_coeff_1\\
                               & $cC \d^3 C$ & $\phi^{(10)} -\tfrac23 \phisup{1} - \tfrac i3 p \cdot A^{(11)} - \tfrac16 \tr S^{(13)}$ & eq:C3rel_scalar_coeff_2 \\
                               & $c \d C \d^2 c$ & $-\phisup{8} + \phisup{9} + \tfrac32 \phisup{2} + \tfrac i2 p \cdot \Asup{4}$ & eq:C3rel_scalar_coeff_3 \\
                               & $c \d C \d^2 C$ & $\phisup{10} + \tfrac32 \phisup{1} - \tfrac i2 p\cdot \Asup{3}$ & eq:C3rel_scalar_coeff_4 \\
                               & $C \d C \d^2 c$ & $-\phisup{10} - \tfrac32 \phisup{1} + \tfrac i2 p\cdot \Asup{6}$ & eq:C3rel_scalar_coeff_5 \\
                               & $C \d C \d^2 C$ & $-\tfrac i2 p \cdot A^{(5)}$ & eq:C3rel_scalar_coeff_6 \\ \midrule
     \multirowcell{6}{Vectors} & $cC \d^2 c \d X^\mu$ & $ip_\mu \phisup{9} + \Asup{11}_\mu - \Asup{3}_\mu + \tfrac i2 p^\nu T^{(14)}_{\mu\nu}$  & eq:C3rel_vector_coeff_1 \\
                               & $cC \d^2 C \d X^\mu$ & $ip_\mu \phisup{10} + \Asup{5}_\mu - ip^\nu S^{(13)}_{\mu\nu}$  & eq:C3rel_vector_coeff_2 \\
                               & $cC \d^2 c \Pi^\mu$ & $ip_\mu \phisup{7} + \Asup{12}_\mu - \Asup{4}_\mu + ip^\nu S^{(15)}_{\mu\nu}$  & eq:C3rel_vector_coeff_3 \\
                               & $cC \d^2 C \Pi^\mu$ & $ip_\mu \phisup{8} - \Asup{11}_\mu + \Asup{6}_\mu - \tfrac i2 p^\nu T^{(14)}_{\nu\mu}$  & eq:C3rel_vector_coeff_4 \\
                               & $cC \d C \d^2 X^\mu$ & $-ip_\mu \phisup{1} + \Asup{5}_\mu$  & eq:C3rel_vector_coeff_5 \\
                               & $cC \d C \d \Pi^\mu$ & $ip_\mu \phisup{2} - 2\Asup{11}_\mu + \Asup{3}_\mu + \Asup{6}_\mu$ & eq:C3rel_vector_coeff_6 \\ \midrule
     \multirowcell{3}{Tensors} & $cC \d C \d X^\mu \d X^\nu$ & $-ip_{(\mu} \Asup{5}_{\nu)}$ & eq:C3rel_tensor_coeff_1\\
                               & $cC \d C \d X^\mu \Pi^\nu$ & $-\eta_{\mu\nu} \phisup{1} + ip_\nu\Asup{3}_\mu + ip_\mu \Asup{6}_\nu - 2S^{(13)}_{\mu\nu}$ & eq:C3rel_tensor_coeff_2 \\
                               & $cC \d C \Pi^\mu \Pi^\nu$ & $\tfrac12\eta_{\mu\nu} \phisup{2} + ip_{(\mu}\Asup{4}_{\nu)} - T^{(14)}_{(\mu\nu)}$  & eq:C3rel_tensor_coeff_3 \\
     \bottomrule
    \end{tabular}
\end{table}
Using this set of equations, we may deduce some of the coefficients $\phi^{(1)}\dots S^{(15)}_{\mu\nu}$ in terms of others which will remain unconstrained or free. Let us do this step-by-step and obtain the equations from \eqref{eq:H2rel cocycle conditions} one by one, putting a box around each of them. 
\begin{enumerate}
    \item Equation \eqref{eq:C3rel_tensor_coeff_1} tells us that $\boxed{A^{(5)}_\mu = 0}$. Then \eqref{eq:C3rel_scalar_coeff_6} is automatically satisfied.
    \item Plugging this into equation \eqref{eq:C3rel_vector_coeff_5} implies that $\boxed{\phisup{1} = 0}$.
    \item The equations 
    Using $\phisup{1} = 0$ in equations \eqref{eq:C3rel_scalar_coeff_4} and \eqref{eq:C3rel_scalar_coeff_5}
    tells us that 
    \begin{enumerate}
        \item $p\cdot(\Asup{6} - \Asup{3}) = 0$
        \item $\boxed{\phisup{10} = \tfrac i2 p \cdot \Asup{6} = \tfrac i2 p \cdot \Asup{3}}$
    \end{enumerate}
    \item Equation \eqref{eq:C3rel_tensor_coeff_2}
    can be written as 
    \begin{equation*}
        \begin{split}
           -\eta_{\mu\nu} \phisup{1} + ip_{(\nu}\Asup{3}_{\mu)} - ip_{[\mu}\Asup{3}_{\nu]} + ip_{(\mu} \Asup{6}_{\nu)} + ip_{[\mu} \Asup{6}_{\nu]} - 2S^{(13)}_{\mu\nu} = 0 \\
           \iff -\eta_{\mu\nu} \phisup{1} + ip_{[\mu} (\Asup{6}_{\nu]} - \Asup{3}_{\nu]}) + ip_{(\mu} (\Asup{3}_{\nu)} + \Asup{6}_{\nu)})  - 2S^{(13)}_{\mu\nu} = 0 \\
           \iff  ip_{(\mu} (\Asup{3}_{\nu)} + \Asup{6}_{\nu)})  = 2S^{(13)}_{\mu\nu} \quad \text{and} \quad  ip_{[\mu} (\Asup{6}_{\nu]} - \Asup{3}_{\nu]}) = 0
        \end{split}
    \end{equation*}
    where we set $\phisup{1} = 0$ in the final step.
    The antisymmetric equation $ip_{[\mu} (\Asup{6}_{\nu]} - \Asup{3}_{\nu]}) = 0$ together with $p \cdot (\Asup{6} - \Asup{3}) = 0$ tells us that $\Asup{6}_\mu - \Asup{3}_\mu$ is not just orthogonal to $p^\mu$, but must also be proportional to $p^\mu$. That is, there exists some scalar field $\varphi$ such that $$\Asup{6}_\mu - \Asup{3}_\mu = i p_\mu \varphi \iff \boxed{\Asup{6}_\mu = i p_\mu \varphi + \Asup{3}_\mu}.$$
    Substituting this into the symmetric equation gives
   $$\boxed{S^{(13)}_{\mu\nu} = ip_{(\mu}\Asup{3}_{\nu)} - \tfrac12 p_\mu p_\nu \varphi.}$$
   Equation \eqref{eq:C3rel_vector_coeff_2} is automatically satisfied by the expressions for $S^{(13)}_{\mu\nu}$ and $\phisup{10}$ (recall that $\Asup{5}_\mu = 0$), so no new information is obtained from this equation.
    \item We then substitute $\Asup{6}_\mu = i p_\mu \varphi - \Asup{3}_\mu$ into equation \eqref{eq:C3rel_vector_coeff_6} to get
    $$\boxed{\Asup{11}_\mu =  \Asup{3}_\mu + \tfrac i2 p_\mu (\phisup{2} - \varphi).}$$
    This expression for $\Asup{11}_\mu$, together with the earlier formulas for $\phisup{10}$ and $S^{(13)}_{\mu\nu}$, automatically satisfies equation \eqref{eq:C3rel_scalar_coeff_2}.
    \item We make $\phisup{9}$ the subject of equation \eqref{eq:C3rel_scalar_coeff_3} to get
    $$\boxed{\phisup{9} = \phisup{8} - \tfrac32 \phisup{2} - \tfrac i2 p \cdot \Asup{4}.}$$
    \item Likewise, rearranging equations \eqref{eq:C3rel_tensor_coeff_3} and \eqref{eq:C3rel_vector_coeff_3} give
    $$\boxed{T^{(14)}_{(\mu\nu)} = \tfrac12\eta_{\mu\nu} \phisup{2} + ip_{(\mu}\Asup{4}_{\nu)}} \quad \text{and} \quad   \boxed{\Asup{12}_\mu = \Asup{4}_\mu - ip_\mu \phisup{7} - ip^\nu S^{(15)}_{\mu\nu}}$$
    respectively.
    \item Next, we add equations \eqref{eq:C3rel_vector_coeff_1} and \eqref{eq:C3rel_vector_coeff_4}
    (keeping in mind that each of these equal zero) to get
    \begin{equation*}
        i p^\nu T_{[\mu\nu]} + (\Asup{6}_\mu - \Asup{3}_\mu) + ip_\mu (\phisup{8} + \phisup{9}) = 0.
    \end{equation*}
    Using $\Asup{6}_\mu - \Asup{3}_\mu = i p_\mu \varphi$ and $\phisup{9} = \phisup{8} - \tfrac32 \phisup{2} - \tfrac i2 p \cdot \Asup{4}$, we may write the above as
    \begin{equation*}
    \begin{split}
        i p^\nu T_{[\mu\nu]} = -ip_\mu (2\phisup{8} - \tfrac32 \phisup{2} - \tfrac i2 p \cdot \Asup{4}) - i p_\mu \varphi \\
        \iff   \boxed{p^\nu T_{[\mu\nu]}^{(14)} = p_\mu \left(\tfrac32 \phi^{(2)} + \tfrac{i}2 p \cdot A^{(4)} - 2 \phi^{(8)} - \varphi \right).}
    \end{split}
    \end{equation*}
    \item Finally, we substitute 
    \begin{enumerate}
        \item $\phisup{9} = \phisup{8} - \tfrac32 \phisup{2} - \tfrac i2 p \cdot \Asup{4}$,
        \item $\Asup{12}_\mu = \Asup{4}_\mu - ip_\mu \phisup{7} - ip^\nu S^{(15)}_{\mu\nu}$, and
        \item $T^{(14)}_{(\mu\nu)} = \tfrac12\eta_{\mu\nu} \phisup{2} + ip_{(\mu}\Asup{4}_{\nu)} \implies \tr T^{(14)} = 13 \phisup{2} + i p \cdot \Asup{4}$
    \end{enumerate} into equation \eqref{eq:C3rel_scalar_coeff_1}, the last of the available equations from Table \ref{tab:Z2rel_calc}, to obtain our final cocycle condition from \eqref{eq:H2rel cocycle conditions}:
    \begin{equation*}
    \begin{gathered}
        -\phisup{8} + (\phisup{8} - \tfrac32 \phisup{2} - \tfrac i2 p \cdot \Asup{4}) - \tfrac23 \phisup{2} + \tfrac i3 p^\mu (\Asup{4}_\mu - ip^\nu S^{(15)}_{\mu\nu}) + \tfrac{13}{6}  \phisup{2} + \tfrac i6 p \cdot \Asup{4} = 0 \\
        \iff \boxed{p^\mu p^\nu S^{(15)}_{\mu\nu} = 0.}
    \end{gathered}
    \end{equation*}
\end{enumerate}
Thus, we have obtained all the equations of \eqref{eq:H2rel cocycle conditions}. 
Any $\Psi\in\Crel^2(p)$ given by \eqref{eq:gen_C2rel_elem} satisfying \eqref{eq:H2rel cocycle conditions} is a 2-cocycle. In other words, any $\Psi \in \Zrel^2(p)$) can be written in the form given by \eqref{eq:gen_C2rel_elem}, and it can be completely determined by the free variables 
\begin{equation} \label{eq:free_var_Z2rel}
    \varphi,\, \phi^{(2)},\, A_\mu^{(3)},\, A_\mu^{(4)},\, \phi^{(7)},\, \phi^{(8)},\, F_{\mu\nu} = T_{[\mu\nu]}^{(14)},\, G_{\mu\nu} = S^{(15)}_{\mu\nu},
\end{equation}
up to the constraints on $p^\nu F_{\mu\nu}$ and $p^\mu p^\nu G_{\mu\nu}$. 

Next, we determine which of these free variables span the space of coboundaries $\Brel^2(p)$ (i.e., which of these components can be gauged away/set to zero by addition of a suitable coboundary term).
To any $\Psi \in \Zrel^2(p)$, the addition of an arbitrary $\delta\Psi = d \Phi$ for any $\Phi \in \Crel^1(p)$ preserves the cocycles conditions (i.e., $\Psi + \delta \Psi \in \Zrel^2(p)$ for all such $\delta \Psi$) by virtue of $d^2=0$.
Such $\delta \Psi$ are thus called trivial cocycles. Physically, they correspond to gauge transformations (hence the notation used here).
The core idea behind computing cohomology is to rule out any such $\delta\Psi$. 
To compute this, we act $d$ on the most general $\Phi \in \Crel^1(p)$, given by 
\begin{equation}
  \Phi = \lambda^{(1)} \d C + \lambda^{(2)} c C b + \lambda^{(3)} c C B + \chi^{(4)}_\mu c \d X^\mu + \chi^{(5)}_\mu C \d X^\mu + \chi^{(6)}_\mu c \Pi^\mu + \chi^{(7)}_\mu C \Pi^\mu
\end{equation}
using \eqref{eq:d-on-scalars-gh-no-1} and \eqref{eq:d-on-vectors-gh-no-1}.
Then, we demand $d\Phi = \delta\Psi$. Analogous to how we computed $d\psi$ earlier, we group together the contributions to the basis elements of $\Crel^2$. The difference here is that instead of setting the resulting coefficients to be zero, we equate them to their corresponding variations, leading to the sought after gauge transformation conditions. This is summarised in Table \ref{tab:B2rel_calc}.
\begin{table}
    \centering
    \renewcommand{\arraystretch}{1.2}
    \caption{Summary of the computation of $d\Phi = \delta \Psi$. The gauge transformations obtained here give us the form of the most general state we can add to any given cocycle in a cohomologically trivial manner.}
    \begin{tabular}{|c|c|c L|}
        \toprule
         {Component of $\delta \Psi$} & {Variation} & \multicolumn{2}{|c|}{Contribution from $d\Phi$} \\
       \midrule
        $c C\d C b$  &  $\delta\phisup{1}$          & $0$ & eq:delta_phi_1\\
        $c C \d C B$  & $\delta\phisup{2}$           & $2 \lambda^{(2)}$ & eq:delta_phi_2 \\
        $c \d C \d X^\mu$  &  $\delta\Asup{3}_\mu$        & $i p_\mu \lambda^{(1)} + \chi^{(5)}_\mu$ & eq:delta_A_3 \\
        $c \d C \Pi^\mu$  &  $\delta\Asup{4}_\mu$        & $\chi^{(4)}_\mu + \chi^{(7)}_\mu$ & eq:delta_A_4 \\
        $C \d C \d X^\mu$  &  $\delta\Asup{5}_\mu$        & $0$ & eq:delta_A_5 \\
        $C \d C \Pi^\mu$  &  $\delta\Asup{6}_\mu$        & $-i p_\mu \lambda^{(1)} + \chi^{(5)}_\mu$ & eq:delta_A_6 \\ \midrule
        $c \d^2 c$  &  $\delta\phisup{7}$          & $\tfrac32 \lambda^{(3)} - \tfrac i2 p \cdot \chi^{(6)}$  & eq:delta_phi_7 \\
        $c \d^2 C$  &  $\delta\phisup{8}$          & $\lambda^{(1)} + \tfrac32 \lambda^{(2)} + \tfrac i2 p \cdot \chi^{(4)}$  & eq:delta_phi_8 \\
        $C \d^2 c$   &  $\delta\phisup{9}$          & $ \lambda^{(1)} - \tfrac32 \lambda^{(2)} - \tfrac i2 p \cdot \chi^{(7)}$  & eq:delta_phi_9 \\
        $C \d^2 C$   &  $\delta\phisup{10}$         & $\tfrac i2 p \cdot \chi^{(5)}$ & eq:delta_phi_10 \\
        $c C \d^2 X^\mu$   &  $\delta\Asup{11}_\mu$       & $ip_\mu \lambda^{(2)} + \chi^{(5)}_\mu$  & eq:delta_A_11 \\
        $c C \d \Pi^\mu $   &  $\delta\Asup{12}_\mu$       & $-i p_\mu \lambda^{(3)} + \chi^{(4)}_\mu+ \chi^{(7)}_\mu$ & eq:delta_A_12 \\ 
        \midrule
        $c C \d X^\mu \d X^\nu$  & $\delta S^{(13)}_{\mu\nu}$     & $i p_{(\mu}\chi^{(5)}_{\nu)} $ & eq:delta_S_13\\
        $c C \d X^\mu \Pi^\nu$  &  $\delta T^{(14)}_{\mu\nu}$    & $\eta_{\mu\nu} \lambda^{(2)} + i p_\nu \chi^{(4)}_\mu + i p_\mu \chi^{(7)}_\nu$ & eq:delta_T_14 \\
        $c C  \Pi^\mu  \Pi^\nu$  &  $\delta S^{(15)}_{\mu\nu}$   & $-\tfrac12 \eta_{\mu\nu} \lambda^{(3)} + i p_{(\mu} \chi^{(6)}_{\nu)}$  & eq:delta_S_15 \\  \midrule
    \multicolumn{2}{|c|}{$\delta \varphi$}  & $-2\lambda^{(1)}$  & eq:delta_varphi \\
    \multicolumn{2}{|c|}{$\delta T^{(14)}_{[\mu\nu]}$}    &  $ip_{[\mu}(\chi^{(4)}_{\nu]}+\chi^{(7)}_{\nu]})$  & eq:delta_T_14_antisym \\
    \bottomrule
    \end{tabular}
 \label{tab:B2rel_calc}
\end{table}

It is easy to check that the gauge transformations from Table \ref{tab:B2rel_calc} are compatible with \eqref{eq:H2rel cocycle conditions}, which confirms that the cocycle conditions are indeed gauge-invariant -- a useful sanity check. For example, consider the variation of $\Asup{12}_\mu$. The cocycle equations tell us that $\Asup{12}_\mu = \Asup{4}_\mu - ip_\mu \phisup{7} - ip^\nu S^{(15)}_{\mu\nu}$. Hence, it should hold true that $\delta \Asup{12}_\mu = \delta \Asup{4}_\mu - ip_\mu \delta \phisup{7} - ip^\nu \delta S^{(15)}_{\mu\nu}$. Using the expressions for $\delta \Asup{4}_\mu$, $\delta \phisup{7}$, and $\delta S_{\mu\nu}^{(15)}$ from Table \ref{tab:B2rel_calc}, we obtain that $\delta \Asup{12}_\mu =  -ip_\mu \lambda^{(3)} + \chi^{(4)}_\mu + \chi^{(7)}_\mu$, which is precisely the variation given by equation \eqref{eq:delta_A_12} in Table \ref{tab:B2rel_calc}, as expected.

Now, we use the available gauge freedom to set as many components of
the coefficients in \eqref{eq:gen_C2rel_elem} to zero as possible. In
particular, we want to do this for the free variables
\eqref{eq:free_var_Z2rel}.
\begin{enumerate}
    \item Equation \eqref{eq:delta_varphi} $\implies$ we can choose
      $\lambda^{(1)} =  \tfrac12 \varphi$, for any $\varphi$, to set
      $\varphi = 0$. In particular, $\lambda^{(1)} = 0$ once such an
      adjustment is made.
    \item Equation \eqref{eq:delta_phi_2} $\implies$ we can choose
      $\lambda^{(2)} = -\tfrac12 \phisup{2}$, for any $\phisup{2}$, to
      set $\phisup{2} = 0$. In particular, $\lambda^{(2)} = 0$ once
      such an adjustment is made.
    \item Equation \eqref{eq:delta_A_3} $\implies$ we can choose
      $\chi^{(5)}_\mu = -ip_\mu \lambda^{(1)} - \Asup{3}_\mu  = 0$,
      for any $\Asup{3}_\mu$ and $\lambda^{(1)}$, to set $\Asup{3}_\mu
      = 0$. Having already set $\varphi = 0$, we see that
      $\chi^{(5)}_\mu = \Asup{3}_\mu$, which is zero for once
      $\Asup{3}_\mu = 0$ is satisfied.
    \item Equation \eqref{eq:delta_A_4} $\implies$ we can choose
      $\chi^{(7)}_\mu = -\chi^{(4)}_\mu$ to set $\Asup{4}_\mu = 0$.
    \item Equation \eqref{eq:delta_phi_7} $\implies$ we can choose
      $\lambda^{(3)} = \tfrac i3 p \cdot \chi^{(6)} - \tfrac23
      \phisup{7}$, for any $\phisup{7}$, to set $\phisup{7} = 0$. As
      usual, $\lambda^{(3)} = \tfrac i3 p \cdot \chi^{(6)}$ once such
      an adjustment is made.
    \item Equation \eqref{eq:delta_phi_8}, together with the above
      gauge choices, implies that we can choose $p\cdot \chi^{(4)} =
      0$ to set $\phisup{8} = 0$.
\end{enumerate}
We are now in the partially gauge-fixed setting, as we still have the
freedom to use $\chi^{(6)}_\mu$ and $\chi^{(4)}_\mu$ subject to $\p
\cdot \chi^{(4)}_\mu$. As one may have noticed, these appear in the
gauge transformations of $G_{\mu\nu} = S^{(15)}_{\mu\nu}$ and
$F_{\mu\nu} = T^{(14)}_{[\mu\nu]}$, so we may use them to set some
components of $G_{\mu\nu}$ and $F_{\mu\nu}$ to zero.  However, to
avoid breaking $\Stab(p)$-invariance, it is useful to summarise our
calculations as follows.  The most general element in $\Hrel^2(p)$ can
be written as
\begin{equation}\label{eq:Zrel2-vertex-operators}
  \Psi = cC \left( G_{\mu\nu} \left(\Pi^\mu \Pi^\nu - i p^{(\nu} \d \Pi^{\mu)} \right)+ F_{\mu\nu} \d X^\mu \Pi^\nu \right) e^{i p\cdot X},
\end{equation}  
where $G_{\mu\nu}$ and $F_{\mu\nu}$ satisfy 
\begin{equation}
  \label{eq:Zrel2-conditions}
  p^\mu p^\nu G_{\mu\nu} = 0, \quad p^\mu F_{\mu\nu}  = 0
\end{equation}
and admit gauge transformations
\begin{equation}
  \label{eq:Brel2-conditions}
    \delta G_{\mu\nu} = -\tfrac i6 \eta_{\mu\nu} p \cdot \chi + i p_{(\mu} \chi_{\nu)}, \quad
    \delta F_{\mu\nu} = i p_{[\mu} \omega_{\nu]},
\end{equation}
where $\chi_\mu := \chi^{(6)}_\mu$, $\omega_\mu := 2 \chi^{(4)}_\mu$, and thus $p \cdot \omega = 0$.

From this point on, one can choose a Witt frame to simplify the calculation and gain some intuition by deducing the structure of $\Hrel^2(p)$ as a module over a maximal compact subgroup $K \cong \SO(V^\top) = \SO(24)$ of $H := \Stab(p) \cong \ISO(24)$. This is done in Section \ref{sec:calculating-hrel2p}. However, by doing so, we lose information about the action of the non-compact part of $\Stab(p)$ on the relative cohomology. Hence, at the end of the day, we must understand how all of $\Stab(p)$ acts on relative cohomology to eventually understand how the full Lorentz group acts on the ambitwistor string spectrum. This is done in Section \ref{sec:rep-theory-hrel2}. 

\section{Detailed calculation of $\Hrel^3(p)$ for $p^2=0$, $p\neq 0$}
\label{app:hrel3-details}

In this appendix we present a detailed calculation of $\Hrel^3(p)$ for
$p^2=0$ (and $p \neq 0$).  This will be useful in proving that
$\Hrel^2(p)$ and $\Hrel^3(p)$ are dual $\Stab(p)$-modules, as
expected by Poincaré duality, which is done at the end of
Section~\ref{sec:rep-theory-hrel3}.

Let us therefore fix $p\neq 0$ with $p^2 = 0$.  Recall that a general
relative $3$-cochain $\Psi \in \Crel^3(p)$ is given by
\begin{multline}
  \Psi = \phi^{(1)} c C \d^3 c + \phi^{(2)} c C \d^3 C + \phi^{(3)} c  \d C \d^2 c + \phi^{(4)} c \d C \d^2 C + \phi^{(5)} C \d C \d^2 c +  \phi^{(6)} C \d C \d^2 C\\
  + A^{(7)}_\mu c C \d^2 c \d X^\mu + A^{(8)}_\mu c C \d^2 C \d X^\mu +   A^{(9)}_\mu c C \d^2 c \Pi^\mu + A^{(10)}_\mu c C \d^2 C \Pi^\mu +   A^{(11)}_\mu c C \d C \d^2 X^\mu\\
  +   A^{(12)}_\mu c C \d C \d  \Pi^\mu + S^{(13)}_{\mu\nu} c C \d C \d X^\mu \d X^\nu + T^{(14)}_{\mu\nu} c C  \d C \d X^\mu \Pi^\nu + S^{(15)}_{\mu\nu} c C \d C \Pi^\mu \Pi^\nu,
\end{multline}
with $S^{(13)}$ and $S^{(15)}$ symmetric tensors and with the tacit
understanding that every term in the RHS has a $W_p = \exp(i p\cdot
X)$ on the right.  The action of the BRST differential on each of the
monomials in $\Psi$ is given in Appendix~\ref{app:calculation of d on
  Crel}, particularly
equations~\eqref{eq:d-on-scalars-gh-no-3}--\eqref{eq:d-on-tensors-gh-no-3}.
Setting $d\Psi = 0$ results in linear equations which are summarised
in Table~\ref{tab:Z3rel_calc}.
\begin{table}[h]
\centering
\renewcommand{\arraystretch}{1.2}
\caption{Summary of the computation of $d\Psi$ using
  \eqref{eq:d-on-scalars-gh-no-3}, \eqref{eq:d-on-vectors-gh-no-3} and
  \eqref{eq:d-on-tensors-gh-no-3}}
\label{tab:Z3rel_calc}
    \begin{tabular}{|c|c|c L|}
        \toprule
       \multicolumn{2}{|c|}{Basis vector of $\Crel^4(p)$} & \multicolumn{2}{|c|}{Coefficients}\\
       \midrule
     \multirowcell{3}{Scalars} & $cC \d^2 c \d^2 C$ & $-\phisup{4} - \phisup{5} + \tfrac i2 p \cdot \Asup{7} + \tfrac i2 p \cdot \Asup{10}$ & eq:C4rel_scalar_coeff_1\\
                               & $cC \d C\d^3 C$ & $2 \phisup{2} - \phisup{4} + \phisup{5} - \tfrac i3 p \cdot \Asup{12} - \tfrac 16 \tr \Tsup{14}$ & eq:C4rel_scalar_coeff_2 \\
                               & $cC \d C \d^3 c$ & $ \phisup{6} + \tfrac i3 p \cdot \Asup{11} + \tfrac 16 \tr \Ssup{13}$ & eq:C4rel_scalar_coeff_3 \\ \midrule
     \multirowcell{4}{Vectors} & $cC \d C \d^2 c \d X^\mu$ & $i p_\mu \phisup{5} + \Asup{8}_\mu - \Asup{11}_\mu - \tfrac i2 p^\nu \Tsup{14}_{\mu\nu}$  & eq:C4rel_vector_coeff_1 \\
                               & $cC \d C \d^2 C \d X^\mu$ & $i p_\mu \phisup{6} + i p^\nu \Ssup{13}_{\mu\nu} $  & eq:C4rel_vector_coeff_2 \\
                               & $cC \d C \d^2 c \Pi^\mu$ & $i p_\mu \phisup{3} - \Asup{7}_\mu+ \Asup{10}_\mu - \Asup{12}_\mu - p^\nu \Ssup{15}_{\mu\nu}  $  & eq:C4rel_vector_coeff_3 \\
                               & $cC \d C \d^2 C \Pi^\mu$ & $i p_\mu \phisup{4} - \Asup{8}_\mu + \Asup{11}_\mu + \tfrac i2 p^\nu \Tsup{14}_{\nu\mu}$  & eq:C4rel_vector_coeff_4 \\ \bottomrule
    \end{tabular}
\caption*{The cocycle condition $d\Psi = 0$ is equivalent to the
  vanishing of all the coefficients in the table.}
\end{table}

Using this set of equations, we may solve for some of the
coefficients, leaving the rest free subject to a reduced system of
equations.  In this way we obtain the
equations~\eqref{eq:Z3rel-equations-partial} and
\eqref{eq:Z3rel-remaining-equations}, which we identify by putting a
box around each of them.

The first three
equations~\eqref{eq:C4rel_scalar_coeff_1}--\eqref{eq:C4rel_scalar_coeff_3}
can be used to solve for $\phisup{5}$, $\phisup{2}$ and $\phisup{6}$.
We find \begin{equation}
  \boxed{\phisup{5} = - \phisup{4} + \tfrac i2 p \cdot \Asup{7} + \tfrac i2 p \cdot \Asup{10}},
\end{equation}
\begin{equation}
  \phisup{2} = \phisup{4} - \tfrac i4 p \cdot \Asup{7} - \tfrac i4 p \cdot \Asup{10} + \tfrac i6 p \cdot \Asup{12} + \tfrac1{12} \tr \Tsup{14}
\end{equation}
and
\begin{equation}
  \phisup{6} = - \tfrac i3 p \cdot \Asup{11} - \tfrac 16 \tr \Ssup{13}.
\end{equation}
Then equation~\eqref{eq:C4rel_vector_coeff_1} can be used to solve for $\Asup{11}$:
\begin{equation}
  \boxed{\Asup{11} = -i p_\mu \phisup{4} - \tfrac12 p_\mu p \cdot \Asup{7} + \Asup{8}_\mu - \tfrac12 p_\mu p \cdot \Asup{10} - \tfrac i2 p^\nu \Tsup{14}_{\mu\nu}},
\end{equation}
which in turn modifies $\phisup{6}$:
\begin{equation}
  \boxed{\phisup{6} = - \tfrac i3 p \cdot \Asup{8} - \tfrac 16 p^\mu p^\nu \Tsup{14}_{\mu\nu} - \tfrac 16 \tr \Ssup{13}}.
\end{equation}
Finally, we use equation~\eqref{eq:C4rel_vector_coeff_3} to solve for $\Asup{12}$:
\begin{equation}
  \boxed{\Asup{12}_\mu = i p_\mu \phisup{3} - \Asup{7}_\mu + \Asup{10}_\mu - i p^\nu \Ssup{15}_{\mu\nu}},
\end{equation}
which in turn modifies $\phisup{2}$:
\begin{equation}
  \boxed{\phisup{2} = \phisup{4} - \tfrac{5i}{12} p \cdot \Asup{7} - \tfrac i{12} p \cdot \Asup{10} + \tfrac 16 p^\mu p^\nu \Ssup{15}_{\mu\nu} + \tfrac1{12} \tr \Tsup{14}}.
\end{equation}
In summary, we are left with $\phisup{1}$, $\phisup{3}$, $\phisup{4}$, $\Asup{7}_\mu$, $\Asup{8}_\mu$, $\Asup{9}_\mu$, $\Asup{10}_\mu$, $\Ssup{13}_{\mu\nu}$, $\Tsup{14}_{\mu\nu}$ and $\Ssup{15}_{\mu\nu}$ subject to equation~\eqref{eq:C4rel_vector_coeff_2}
\begin{equation}
  \label{eq:S13-3-cocycle-condition}
  \boxed{p^\nu \Ssup{13}_{\mu\nu} - \tfrac16 p_\mu \tr \Ssup{13} = \tfrac i3 p_\mu p \cdot \Asup{8}+ \tfrac 16 p_\mu p^\lambda p^\nu \Tsup{14}_{\lambda\nu}}
\end{equation}
and equation~\eqref{eq:C4rel_vector_coeff_4}
\begin{equation}
  \label{eq:T14-3-cocycle-condition}
  \boxed{p^\nu \Tsup{14}_{[\mu\nu]} = \tfrac i2 p_\mu \left( p \cdot \Asup{7} + p \cdot \Asup{10} \right)}.
\end{equation}

We now determine the space $\Brel^3(p)$ of coboundaries.  That is the
image of the differential $d \colon \Crel^2(p) \to \Crel^3(p)$.  The
general relative $2$-cochain $\Phi$ is given by
\begin{multline}
  \Phi = \lambda^{(1)} c C \d C  b + \lambda^{(2)} c C \d C B + \chi^{(3)}_\mu c \d C \d X^\mu + \chi^{(4)}_\mu c \d C \Pi^\mu + \chi^{(5)}_\mu C \d C \d X^\mu + \chi^{(6)}_\mu C \d C \Pi^\mu\\
  + \lambda^{(7)} c \d^2 c + \lambda^{(8)} c \d^2 C + \lambda^{(9)} C
  \d^2  c + \lambda^{(10)} C \d^2 C + \chi^{(11)}_\mu c C \d^2 X^\mu + \chi^{(12)}_\mu c C \d \Pi^\mu\\
  + \Sigma^{(13)}_{\mu\nu} c C \d X^\mu \d X^\nu + \Theta^{(14)}_{\mu\nu} c C \d X^\mu \Pi^\nu + \Sigma^{(15)}_{\mu\nu} c C \Pi^\mu \Pi^\nu,
\end{multline}
where $\Sigma^{(13)}_{\mu\nu}$ and $\Sigma^{(15)}_{\mu\nu}$ are symmetric, but $\Theta^{(14)}_{\mu\nu}$ is a general $2$-tensor.
We denote the addition of the coboundary to a cocycle $\Psi$ by $\Psi
\mapsto \Psi + \delta\Psi$, where $\delta\Psi = d\Phi$.
Table~\ref{tab:B3rel_calc} summarises the calculation.

\begin{table}
    \centering
    \renewcommand{\arraystretch}{1.2}
    \caption{Summary of the computation of $d\Phi = \delta \Psi \in \Brel^3(p)$}
    \begin{tabular}{|>{$}c<{$}|>{$}c<{$}|>{$}c<{$} L|}
        \toprule
      \multicolumn{1}{|c|}{Component of $\delta \Psi$} & \multicolumn{1}{c|}{Variation} & \multicolumn{2}{c|}{Contribution from $d\Phi$}\\
       \midrule
        c C\d^3c  &  \delta\phisup{1}     &  -\lambda^{(8)} + \lambda^{(9)} - \tfrac{2}{3} \lambda^{(2)} + \tfrac i3 p \cdot \chi^{(12)} + \tfrac16 \tr \Theta^{(14)} & eq:B3-delta_phi_1\\
        c C \d^3 C  & \delta\phisup{2}     &  \lambda^{(10)} - \tfrac{2}{3} \lambda^{(1)} - \tfrac i3 p \cdot \chi^{(11)} - \tfrac16 \tr \Sigma^{(13)} & eq:B3-delta_phi_2 \\
        c \d C \d^2 c  &  \delta\phisup{3}   & -\lambda^{(8)} + \lambda^{(9)} + \tfrac32 \lambda^{(2)} + \tfrac i2 p \cdot \chi^{(4)} & eq:B3-delta_phi_3 \\
        c \d C \d^2 C  &  \delta\phisup{4}   & \lambda^{(10)} + \tfrac{3}{2} \lambda^{(1)} - \tfrac i2 p \cdot \chi^{(3)}  & eq:B3-delta_phi_4 \\
        C \d C \d^2 c  & \delta\phisup{5}   & -\lambda^{(10)} - \tfrac{3}{2} \lambda^{(1)} + \tfrac i2 p \cdot \chi^{(6)}  & eq:B3-delta_phi_5 \\
        C \d C \d^2 C  & \delta\phisup{6}   & -\tfrac i2 p \cdot \chi^{(5)}  & eq:B3-delta_phi_6 \\ \midrule
        c C \d^2 c \d X^\mu  & \delta\Asup{7}_\mu & i p_\mu \lambda^{(9)} - \chi_\mu^{(3)} + \chi_\mu^{(11)} + i p^\nu \Theta_{\mu\nu}^{(14)}  & eq:B3-delta_A_7 \\
        c C \d^2 C \d X^\mu  & \delta\Asup{8}_\mu & i p_\mu \lambda^{(10)} + \chi_\mu^{(5)} - i p^\nu \Sigma_{\mu\nu}^{(13)} & eq:B3-delta_A_8 \\
        c C \d^2 c \Pi^\mu   & \delta\Asup{9}_\mu & i p_\mu \lambda^{(7)} - \chi_\mu^{(4)} + \chi_\mu^{(12)} + i p^\nu \Sigma_{\mu\nu}^{(15)} & eq:B3-delta_A_9 \\
        c C \d^2 C \Pi^\mu   &  \delta\Asup{10}_\mu & i p_\mu \lambda^{(8)} + \chi_\mu^{(6)} - \chi_\mu^{(11)} - \tfrac i2 p^\nu \Theta_{\nu\mu}^{(14)} & eq:B3-delta_A_10 \\
        c C \d C \d^2 X^\mu   & \delta\Asup{11}_\mu & -i p_\mu \lambda^{(1)} + \chi_\mu^{(5)}  & eq:B3-delta_A_11 \\
        c C \d C \d \Pi^\mu   & \delta\Asup{12}_\mu & i p_\mu \lambda^{(2)} + \chi_\mu^{(3)} + \chi_\mu^{(6)} - 2 \chi_\mu^{(11)}  & eq:B3-delta_A_12 \\ 
        \midrule
        c C \d C \d X^\mu \d X^\nu  & \delta S^{(13)}_{\mu\nu}     & i p_{(\mu} \chi_{\nu)}^{(5)}   & eq:B3-delta_S_13\\
        c C \d C \d X^{(\mu} \Pi^{\nu)}  &  \delta T^{(14)}_{(\mu\nu)}    & -\eta_{\mu\nu} \lambda^{(1)} + i p_{(\mu} \chi_{\nu)}^{(3)} + i p_{(\mu} \chi_{\nu)}^{(6)} - 2 \Sigma_{\mu\nu}^{(13)}& eq:B3-delta_T_14_sym \\
        c C \d C \d X^{[\mu} \Pi^{\nu]}  &  \delta T^{(14)}_{[\mu\nu]}    & - i p_{[\mu} \chi_{\nu]}^{(3)} + i p_{[\mu} \chi_{\nu]}^{(6)} & eq:B3-delta_T_14_skew \\
      c C \d C \Pi^\mu  \Pi^\nu  &  \delta S^{(15)}_{\mu\nu}   & \tfrac12 \eta_{\mu\nu} \lambda^{(2)} + i p_{(\mu} \chi_{\nu)}^{(4)} - \Theta_{(\mu\nu)}^{(14)}  & eq:B3-delta_S_15 \\
      \bottomrule
    \end{tabular}
 \label{tab:B3rel_calc}
\end{table}

As a sanity check on our calculations, one can check that the cocycle
conditions~\eqref{eq:S13-3-cocycle-condition} and
\eqref{eq:T14-3-cocycle-condition} are satisfied by the relevant
coboundaries.

We will concentrate on trying to set to zero as many of the free
parameters in the space of cocycles.  For starters, from
equation~\eqref{eq:B3-delta_phi_1}, it follows that we can choose
\begin{equation}
 \lambda^{(2)}  = \tfrac32 (\lambda^{(9)}-\lambda^{(8)}) + \tfrac i2 p \cdot \chi^{(12)} + \tfrac14 \tr \Theta^{(14)} + \phisup{1}
\end{equation}
to set $\phi^{(1)} = 0$.  Similarly, from equation~\eqref{eq:B3-delta_phi_3}
we can choose
\begin{equation}
  \lambda^{(8)} - \lambda^{(9)}= \tfrac32 \lambda^{(2)} + \tfrac i2 p \cdot \chi^{(4)} + \phi^{(3)}
\end{equation}
to set $\phi^{(3)} = 0$.  These two equation are consistent and, once we back substitute, they give
\begin{equation}
  \begin{split}
    \lambda^{(2)} &= \tfrac 6{13} \left( -\tfrac i2 p \cdot \chi^{(4)} + \tfrac i3 p \cdot \chi^{(12)} + \tfrac16 \tr \Theta^{(14)} + \phisup{1} - \phisup{3}\right)\\
      \lambda^{(8)} - \lambda^{(9)} &= \tfrac 2{13} \left( i p \cdot \chi^{(4)} + \tfrac32 i p \cdot \chi^{(12)} + \tfrac 34 \tr \Theta^{(14)} + \tfrac 92 \phisup{1} + 2 \phisup{3}\right).
  \end{split}
\end{equation}
Now using equation~\eqref{eq:B3-delta_phi_4}, we can choose
\begin{equation}
  \lambda^{(1)} = \tfrac 23 \left( - \lambda^{(10)} + \tfrac i2 p \cdot \chi^{(3)} - \phisup{4}\right)
\end{equation}
to set $\phi^{(4)}=0$ and in this way set all free scalar cocycle parameters to zero.

We can set all free vector cocycle parameters to zero as well.  Using
equation~\eqref{eq:B3-delta_A_7}, we choose
\begin{equation}
  \chi_\mu^{(11)} = \chi_\mu^{(3)} - i p_\mu \lambda^{(9)} - i p^\nu  \Theta_{\mu\nu}^{(14)} - A_\mu^{(7)}
\end{equation}
to set $A_\mu^{(7)} = 0$.  We may substitute for $\lambda^{(9)}$, but
that expression does not involve any parameter we are going to
constraint here or in the sequel, so there is no need to do so.
Using equation~\eqref{eq:B3-delta_A_8}, we choose
\begin{equation}
  \chi_\mu^{(5)} = i p^\nu \Sigma_{\mu\nu}^{(13)} - i p_\mu  \lambda^{(10)} - A_\mu^{(8)}
\end{equation}
to set $A_\mu^{(8)}$.  We use equation~\eqref{eq:B3-delta_A_9} and choose
\begin{equation}
  \chi_\mu^{(12)} = \chi_\mu^{(4)} - i p_\mu \lambda^{(7)} - i p^\nu  \Sigma_{\mu\nu}^{(15)} - A_\mu^{(9)}
\end{equation}
to set $A_\mu^{(9)} = 0$.  This changes the expressions for
$\lambda^{(2)}$, $\lambda^{(8)}-\lambda^{(9)}$ and $\chi_\mu^{(11)}$
above, but again it does not constraint the parameters since these do
not occur in $\chi_\mu^{(12)}$.  Finally, we use
equation~\eqref{eq:B3-delta_A_10} and choose
\begin{equation}
  \chi_\mu^{(6)} = \chi_\mu^{(11)} + \tfrac i2 p^\nu  \Theta_{\nu\mu}^{(14)} - i p_\mu \lambda^{(8)} - A_\mu^{(10)}
\end{equation}
to set $A_\mu^{(10)} = 0$.  This sets all free vector cocycle
parameters to zero.

There are two symmetric tensor cocycles we can set to zero as well.
Using equation~\eqref{eq:B3-delta_T_14_sym}, we choose
\begin{equation}
  \Sigma_{\mu\nu}^{(13)} = -\tfrac12 \eta_{\mu\nu} \lambda^{(1)} +   \tfrac i2 p_{(\mu} \chi_{\nu)}^{(3)} + \tfrac i2 p_{(\mu}  \chi_{\nu)}^{(6)} + \tfrac12 T_{(\mu\nu)}^{(14)}
\end{equation}
to set $T_{(\mu\nu)}^{(14)}=0$.  Again we can substitute
the above expressions for $\lambda^{(1)}$ and $\chi_\mu^{(6)}$ here,
but there are no new terms in $\Sigma_{\mu\nu}^{(13)}$ appearing.
Finally, we use equation~\eqref{eq:B3-delta_S_15} and choose
\begin{equation}
  \label{eq:theta-14-symm}
  \Theta_{(\mu\nu)}^{(14)}= \tfrac 12 \eta_{\mu\nu} \lambda^{(2)} + i  p_{(\mu} \chi_{\nu)}^{(4)} + S_{\mu\nu}^{(15)}
\end{equation}
to set $S_{\mu\nu}^{(15)}=0$.  Here we have to be careful, because the
expression for $\lambda^{(2)}$ involves the trace of
$\Theta_{(\mu\nu)}^{(14)}$.  Indeed it follows from
equation~\eqref{eq:theta-14-symm} that
\begin{equation}
  p^\mu p^\nu \Theta_{\mu\nu}^{(14)} = p^\mu p^\nu S_{\mu\nu}^{(15)}
\end{equation}
and that
\begin{equation}
  \tr\Theta^{(14)} = 13 \lambda^{(2)} + i p \cdot \chi^{(4)} + \tr S^{(15)}.
\end{equation}
Substitution of the expression for $\lambda^{(2)}$ in this equation
results in a constraint involving coboundary parameters we have yet to
use:
\begin{equation}
  \tfrac 83 i p \cdot \chi^{(4)} + 2 p^\mu p^\nu  \Sigma_{\mu\nu}^{(15)} = 2 i p \cdot A^{(9)} - 6 \phisup{1} + 6  \phisup{3} - \tr S^{(15)}.
\end{equation}

Summarising this discussion thus far, having set to zero
$\phisup{1}$, $\phisup{3}$, $\phisup{4}$, $A_\mu^{(7)}$,
$A_\mu^{(8)}$, $A_\mu^{(9)}$, $A_\mu^{(10)}$, $T_{(\mu\nu)}^{(14)}$
and $S_{\mu\nu}^{(15)}$, we are left with $T_{[\mu\nu]}^{(14)}$ and
$S_{\mu\nu}^{(13)}$ subject to the cocycle conditions
\begin{equation}
  p^\nu \left( S_{\mu\nu}^{(13)}- \tfrac16 \eta_{\mu\nu} \tr S^{(13)}
  \right) = 0 \qquad\text{and}\qquad p^\nu T_{[\mu\nu]}^{(14)} = 0
\end{equation}
modulo the coboundaries
\begin{equation}
  \begin{split}
    \delta S_{\mu\nu}^{(13)} &= p_\mu p_\nu \left( \tfrac 32   \lambda^{(1)} - \tfrac i4 p \cdot \chi^{(3)} - \tfrac i4 p \cdot   \chi^{(11)}\right) \\
    \delta T_{[\mu\nu]}^{(14)} &= i p_{[\mu} \omega_{\nu]},
  \end{split}
\end{equation}
where $p \cdot \omega =0$ for consistency with the fact that
coboundaries are cocycles. (Indeed one calculates that $p\cdot \omega =
p \cdot \chi^{(6)}- p \cdot \chi^{(3)}$ and that this vanishes once we
have set to zero the above cocycle parameters.)   Notice that the
expression for $\delta S_{\mu\nu}^{(13)}$ implies that $\delta \tr
S^{(13)} = 0$  and hence we may define a symmetric tensor $G_{\mu\nu}
:= S_{\mu\nu}^{(13)} - \frac 16 \eta_{\mu\nu} \tr S^{(13)}$ and a
skew-symmetric tensor $F_{\mu\nu} = T_{[\mu\nu]}^{(14)}$ in terms of
which the cocycle conditions are
\begin{equation}
  \label{eq:3-relative-cocycles}
  p^\nu G_{\mu\nu} = 0 \qquad\text{and}\qquad p^\nu F_{\mu\nu} = 0
\end{equation}
modulo the coboundaries
\begin{equation}
  \label{eq:3-relative-coboundaries}
  \delta G_{\mu\nu} = p_\mu p_\nu \vartheta   \qquad\text{and}\qquad
  \delta F_{\mu\nu} = i p_{[\mu} \omega_{\nu]},
\end{equation}
with $p \cdot \omega = 0$.

As vertex operators for these cocycle representatives we can take
\begin{equation}
  \label{eq:Zrel3-vertex-operator-1}
  G_{\mu\nu} c C \d C \d X^\mu \d X^\nu e^{i p \cdot X} - \tfrac1{20} \tr
  G \left( c C \d C \d X \cdot \d X - C \d C \d^2 C \right) e^ {i p \cdot X},
\end{equation}
which is a cocycle precisely when $p^\nu G_{\mu\nu} = 0$, and
\begin{equation}
  \label{eq:Zrel3-vertex-operator-2}
  F_{\mu\nu} c C \d C \d X^\mu \Pi^\nu e^{i p \cdot X},
\end{equation}
which is a cocycle precisely when $p^\nu F_{\mu\nu} =0$.  We will
re-interpret these results representation-theoretically in
Section~\ref{sec:rep-theory-hrel3}.

There is another cocycle representative for the ``symmetric'' relative
cohomology: namely,
\begin{equation}\label{eq:Zrel3-alt-symm-vertex-op}
  \Psi = G_{\mu\nu} c C \d C \d X^\mu \d X^\nu e^{i p \cdot X} + A_\mu c C  \d (\d C \d X^\mu) e^{i p\cdot X}
\end{equation}
subject to $p^\nu G_{\mu\nu} = 0$ and $\tfrac12 \tr G + i p \cdot A =
0$ with coboundaries
\begin{equation}
  \delta \Psi = d \Lambda,
\end{equation}
where
\begin{equation}
  \Lambda = \omega_\mu C \d C \d X^\mu e^{i p \cdot X} + \alpha \left( c C \d C b - \tfrac32 C \d^2 C - \tfrac12 c C \d X \cdot \d X  \right) e^{i p \cdot X}.
\end{equation}
This translates into $\delta G_{\mu\nu} = i p_{(\mu} \omega_{\nu)}$
and $\delta A_\mu = \omega_\mu - i \alpha p_\mu$, with $p\cdot \omega
= 0$.  This representative has the virtue that there are no terms
purely with ghosts, but its representation-theoretic characterisation
is more involved and discussed in
Section~\ref{sec:other-cocycle-reps}.

\section{Splittings of an exact sequence}
\label{app:splitting}

The long exact sequence relating the relative and absolute
cohomologies results in a short exact sequence of $H$-modules
\begin{equation}
  \label{eq:SES-H3}
  \begin{tikzcd}
    0 \arrow[r] & \Hrel^3(p) \arrow[r] & \sH^3(p) \arrow[r] & \Hrel^2(p) \arrow[r] & 0,
  \end{tikzcd}
\end{equation}
where the first map $\Hrel^3(p) \to \sH^3(p)$ is induced by the
inclusion of $\Crel^3(p) \to \sC^3(p)$ and the second map $\sH^3(p)
\to \Hrel^2(p)$ is induced by $b_0 \colon \sC^3(p) \to \Crel^2(p)$,
which admits a splitting $c_0 \colon \Crel^2(p) \to \sC^3(p)$.  All
these maps are $H$-equivariant.  We are interested in determining
whether the sequence~\eqref{eq:SES-H3} splits as $H$-modules.

Let us simplify the notation and let $E = \Hrel^2(p)$ and
$F = \sH^3(p)$.  By Proposition~\ref{prop:poincare-duality},
$\Hrel^3(p) \cong E^*$ and hence we can rewrite the
sequence~\eqref{eq:SES-H3} as
\begin{equation}
  \label{eq:SES-H-mods}
  \begin{tikzcd}
    0 \arrow[r] & E^* \arrow[r,"i"] & F  \arrow[r,"\pi"] & E \arrow[r] & 0,
  \end{tikzcd}
\end{equation}
which says that $F$ is an extension of $E$ by $E^*$.

We will work at the level of the Lie algebra and hence this sequence
is one of $\h$-modules.   Such $\h$-module extensions are classified
cohomologically by the Chevalley--Eilenberg cohomology $H^1(\h,
\Hom(E,E^*))$.  Let us review this in the present case.  We have that
$\h \cong \fk \ltimes \fa$, where $\fk$ is simple $(\fk \cong
\so(24))$ and $\fa$ is abelian.  As $\fk$-modules, $\fa \cong V^\top$
and $\fk \cong \ext{2} V^\top$.  Both $E$ and $F$ are $\fk$-modules by
restriction and in particular $E \cong  E^*$ as $\fk$-modules.  Since
$\fk$ is simple, the sequence splits as $\fk$-modules.  Let $s \colon
E \to F$ denote a $\fk$-equivariant splitting: $\pi \circ s = \id_E$.
Any other splitting $s' \colon E \to F$ is such that for all $e \in
E$, $s'(e) - s(e) \in \ker \pi = \im i$.  In other words, there is
a linear map $\phi \colon E \to E^*$ such that for all $e \in E$,
\begin{equation}
  \label{eq:new-splitting}
  s'(e) = s(e) + i (\phi(e)).
\end{equation}
If both $s$ and $s'$ are $\fk$-equivariant, then so is $\phi$, so that
$\phi \in \Hom_{\fk}(E,E^*)$.

The obstruction for $s \colon E \to F$ to be $\h$-equivariant is
captured by the linear map $\beta \colon \h \to \Hom(E,E^*)$, sending
$X \mapsto \beta_X$ and defined by
\begin{equation}
  i (\beta_X(e)) = X \cdot s(e) - s(X \cdot e),
\end{equation}
for all $X \in \h$, $e \in E$.

\begin{lemma}
  The map $\beta$ is $\fk$-equivariant.
\end{lemma}

\begin{proof}
  Let $Y \in \fk$.  Then $(Y \cdot \beta)_X = Y \cdot \beta_X -
  \beta_{[Y,X]}$ and $(Y \cdot \beta_X)(e) = Y \cdot \beta_X(e) -
  \beta_X(Y\cdot e)$.  We then simply calculate
  \begin{align*}
    i \left( (Y \cdot \beta)_X(e)  \right) &= i\left( Y \cdot \beta_X(e) - \beta_X (Y\cdot e)  - \beta_{[Y,X]}(e)\right)\\
       &= Y \cdot i (\beta_X(e)) - i (\beta_X(Y\cdot e)) - i (\beta_{[Y,X]}(e))\\
       &= Y \cdot (X \cdot s(e)-s(X\cdot e)) - X \cdot s(Y \cdot e) + s(X \cdot Y\cdot e) - [Y,X]\cdot s(e) + s([Y,X]\cdot e)\\
       &= X \cdot Y \cdot s(e) - Y \cdot s (X\cdot e) - X \cdot s (Y\cdot e) + s(Y\cdot X \cdot e) & \tag{since $[X,Y] = X \cdot Y - Y \cdot X$}\\
       &= X \cdot s(Y\cdot e) - s(Y \cdot X \cdot e) - X \cdot s(Y \cdot e) + s (Y\cdot X \cdot e) & \tag{since $s$ is $\fk$-equivariant}\\
       &= 0.
  \end{align*}
\end{proof}

Furthermore, since $s$ is $\fk$-equivariant, it follows that $\beta_X = 0$ if $X \in \fk$.  In other words, $\beta \in C^1(\h,\fk,\Hom(E,E^*))$ is a $1$-cochain in the relative complex $C^\bullet(\h,\fk,\Hom(E,E^*))$ which is the $\fk$-invariant subcomplex $C^\bullet(\fa,\Hom(E,E^*))^\fk$.

\begin{lemma}
  The map $\beta$ is a cocycle: $\beta \in Z^1(\fa, \Hom(E,E^*))^\fk$.
\end{lemma}

\begin{proof}
  Let $X,Y \in \fa$.  Then since $\fa$ is abelian,
  \begin{equation}
    (d\beta)(X,Y) = X \cdot \beta_Y - Y \cdot \beta_X.
  \end{equation}
  Then we calculate
  \begin{align*}
    i \left( (X \cdot \beta_Y)(e) - (Y \cdot \beta_X)(e)\right) &= i \left(  X \cdot \beta_Y (e) - \beta_Y(X \cdot e)\right) - (X \leftrightarrow Y)\\
                           &= X \cdot i (\beta_Y(e)) - i \beta_Y(X \cdot e) -  (X \leftrightarrow Y)\\
                           &= X \cdot (Y \cdot s(e) - s(Y \cdot e)) - Y \cdot s (X \cdot e) + s(Y \cdot X \cdot e) -  (X \leftrightarrow Y)\\
                           &= (X \cdot Y) \cdot s(e) - X \cdot s(Y \cdot e) - Y \cdot s(X \cdot e) + S(Y \cdot X e) - (X \leftrightarrow Y)\\
                           &= 0. & \tag{since $X \cdot Y = Y \cdot X$}
  \end{align*}
\end{proof}

Modifying the $\fk$-equivariant splitting $s \mapsto s'$ defined in
equation~\eqref{eq:new-splitting}, the cocycle $\beta$ changes by a
coboundary: $\beta \mapsto \beta' = \beta + d \phi$, where $\phi \in
\Hom_{\fk}(E,E^*) = C^0(\fa,\Hom(E,E^*))^\fk$.  Indeed,
\begin{align*}
  i (\beta'_X (e)) &= X \cdot s'(e) - s'(X \cdot e) \\
                  &= X \cdot (s(e) + i (\phi(e))) - s(X \cdot e) - i (\phi (X \cdot e))\\
                  &= i (\beta_X(e) + X \cdot \phi(e) - \phi(X \cdot e)),
\end{align*}
so that
\begin{equation}
  \beta'_X = \beta_X + X \cdot \phi = \beta_X + d\phi(X).
\end{equation}
The sequence~\eqref{eq:SES-H-mods} splits if the class of $\beta$ in
$H^1(\fa,\Hom(E,E^*))^\fk$ vanishes.

This would require exhibiting an explicit $\fk$-equivariant splitting
$s$ and then asking whether it can be modified by $\phi \in
\Hom_{\fk}(E,E^*)$  so that the new splitting $s'$ is
also $\fa$-equivariant.  Alternatively, we could try to see whether
$H^1(\fa,\Hom(E,E^*))^\fk = 0$, in which case we could deduce that
there does exist a splitting.

As $\fk$-modules,
\begin{equation}
  E^* \cong E \cong \CC \oplus \VV^\top \oplus \sym{2}_0 \VV^\top \oplus \ext{2} \VV^\top,
\end{equation}
which is the complexification of $\RR \oplus V^\top \oplus \sym{2}_0
V^\top \oplus \ext{2} V^\top$, whereas $\fa \cong V^\top$.  Then
\begin{equation}
  \begin{split}
    \Hom_{\fk}(\fa, \Hom(E,E^*)) &\cong \Hom_{\fk}(V^\top, \CC  \otimes_{\RR}(\RR \oplus V^\top \oplus \sym{2}_0 V^\top \oplus  \ext{2} V^\top)^{\otimes 2})\\
    & \cong \CC \otimes_\RR \Hom_{\fk} (V^\top, (\RR \oplus V^\top \oplus \sym{2}_0 V^\top \oplus  \ext{2} V^\top)^{\otimes 2}).
  \end{split}
\end{equation}
A calculation with LiE \cite{LiE} shows that this space is isomorphic
to $\CC^6$ and we will exhibit an explicit basis below.  On the other
hand, since $E$ and $E^*$ consists of four pairwise non-isomorphic
$\fk$-modules occurring with multiplicity $1$, Schur's Lemma says that
\begin{equation}
  \Hom_{\fk}(E,E^*) \cong \CC \otimes_\RR \left((\RR \oplus V^\top
    \oplus \sym{2}_0 V^\top \oplus  \ext{2} V^\top)^{\otimes
      2}\right)^\fk \cong \CC^4,
\end{equation}
spanned by the identity maps for each of the isotypical factors $\RR$,
$V^\top$, $\sym{2}_0 V^\top$ and $\ext{2}V^\top$.  Another calculation
with LiE shows that
\begin{equation}
  C^2(\fa, \Hom(E,E^*))^\fk \cong \CC \otimes_\RR
  \Hom_{\fk}(\ext{2}V^\top,  (\RR \oplus V^\top \oplus \sym{2}_0
  V^\top \oplus  \ext{2} V^\top)^{\otimes 2}) \cong \CC^7.
\end{equation}
So the $\fk$-invariant complex $C^\bullet(\fa, \Hom(E,E^*))^\fk$
starts like
\begin{equation}
  \begin{tikzcd}
    \CC^4 \arrow[r,"\d"] & \CC^6  \arrow[r,"\d"] & \CC^7,
  \end{tikzcd}
\end{equation}
which is the complexification of a complex which starts like
\begin{equation}
  \begin{tikzcd}
    \RR^4 \arrow[r,"\d"] & \RR^6  \arrow[r,"\d"] & \RR^7,
  \end{tikzcd}
\end{equation}
and which we need to study.

Let us determine a basis for these spaces.  We will choose a Witt
frame $(\be_+,\be_i, \be_-)$ for $V$ and choose $p$ so that $p^\sharp
= \be_+$.  In other words, $p^+ = 1$ and $p^- = p^i = 0$.  Dually,
$p_- = 1$ and $p_+ = p_i = 0$.  Then $\fa$ is spanned by $L_{+i} :=
\be_+ \curlywedge \be_i$.

As shown in Appendix~\ref{app:hrel3-details}, the $3$-cocycles have
components $G_{\mu\nu} = G_{\nu\mu}$ and $F_{\mu\nu} = - F_{\nu\mu}$
such that $p^\mu G_{\mu\nu} =0$ and $p^\mu F_{\mu\nu}= 0$ and the
coboundaries are $\delta G_{\mu\nu} = p_\mu p_\nu \vartheta$ and
$\delta F_{\mu\nu} = i p_{[\mu} \omega_{\nu]}$ with
$p \cdot \omega = 0$.  For our choice of $p$, the cocycle equations
are $G_{++} = G_{+-} = G_{+i} =0$ and $F_{+-} = F_{+i} = 0$.  The
coboundaries are
\begin{equation}
  \begin{aligned}
    \delta G_{++} &= 0\\
    \delta G_{+i} &= 0\\
    \delta G_{ij} &= 0\\
    \delta G_{+-} &= 0\\
    \delta G_{-i} &= 0\\
  \end{aligned}
  \qquad\qquad
  \begin{aligned}
    \delta G_{--} &= \vartheta\\
    \delta F_{+i} &= 0\\
    \delta F_{+-} &= 0\\
    \delta F_{ij} &= 0\\
    \delta F_{-i} &= \tfrac i2 \omega_i.\\
  \end{aligned}
\end{equation}
Therefore a basis for the cohomology $\Hrel^3(p)$ is given by
\begin{equation}
  G_{\left<ij\right>}, \quad \tr G^ \top, \quad G_{-i} \quad\text{and}\quad F_{ij}.
\end{equation}
The action of $\fa$ on the cohomology is obtained as before by acting
on the cocycles and dropping any coboundary terms:
\begin{equation}
  \begin{split}
    L_{+k} \cdot G _{\left<ij\right>} &= 0\\
    L_{+k} \cdot \tr G^\top &= 0\\
    L_{+k} \cdot G_{-i} & = G_{ik}\\
    L_{+k} \cdot F_{ij} &= 0,
  \end{split}
\end{equation}

Now we need to determine bases for the first three spaces of cochains
in the Chevalley--Eilenberg complex $C^\bullet(\fa, \Hom(E, E^*))^{\fk}$.
We have the natural $\fa$-module isomorphism $\Hom(E,E^*) \cong E^*
\otimes E^*$, so we will use the latter module instead in our
calculations.

As mentioned above, the space $\Hom_{\fk}(E, E^*) \cong (E^*\otimes
E^*)^\fk$ is $4$-dimensional, with basis
\begin{equation}
  \begin{split}
    c^{(0)}_1 &=\tr G^\top \otimes \tr G^\top\\
    c^{(0)}_2 &= \delta^{ij} G_{-i} \otimes G_{-j}\\
    c^{(0)}_3 &= \delta^{ik}\delta^{j\ell} G_{\left<ij\right>} \otimes G_{\left<k\ell\right>}\\
    c^{(0)}_4 &= \delta^{ik}\delta^{j\ell} F_{ij} \otimes F_{k\ell}.
  \end{split}
\end{equation}

Since $\fa \cong V^\top$, a basis for the $6$-dimensional space
$\Hom_{\fk}(\fa, E^* \otimes E^*)$ is given by:
\begin{equation}
  \begin{aligned}
    c^{(1)}_1 &= \lambda^i \otimes G_{-i} \otimes \tr G^\top\\
    c^{(1)}_2 &= \lambda^i \otimes \delta^{jk} G_{\left<ij\right>} \otimes G_{-k}\\
    c^{(1)}_3 &= \lambda^i \otimes \tr G^\top \otimes G_{-i}
  \end{aligned}
  \qquad\qquad
  \begin{aligned}
    c^{(1)}_4 &= \lambda^i \otimes \delta^{jk} F_{ij} \otimes G_{-k}\\
    c^{(1)}_5 &= \lambda^i \otimes \delta^{jk} G_{-k} \otimes G_{\left<ij\right>}\\
    c^{(1)}_6 &= \lambda^i \otimes \delta^{jk} G_{-k} \otimes F_{ij},
  \end{aligned}
\end{equation}
where we have introduced $\lambda^i$ as basis for $\fa^*$ canonically
dual to the basis $L_{+i}$ for $\fa$.

Finally, since $\ext{2}\fa \cong \ext{2}V^\top$, we must look for
copies of $\ext{2}V^\top$ in $E^* \otimes E^*$.  There are $7$ such
copies and hence there are $7$ basis elements: namely,
\begin{equation}
  \begin{split}
    c^{(2)}_1 &= \tfrac12 \lambda^i \wedge \lambda^j \otimes \tr G^\top \otimes F_{ij}\\
    c^{(2)}_2 &= \tfrac12 \lambda^i \wedge \lambda^j \otimes F_{ij} \otimes \tr G^\top\\
    c^{(2)}_3 &= \tfrac12 \lambda^i \wedge \lambda^j \otimes G_{-i} \otimes G_{-j} \\
    c^{(2)}_4 &= \tfrac12 \lambda^i \wedge \lambda^j \otimes  \delta^{k\ell} G_{\left<ik\right>} \otimes  G_{\left<j\ell\right>}\\
    c^{(2)}_5 &= \tfrac12 \lambda^i \wedge \lambda^j \otimes  \delta^{k\ell} G_{\left<ik\right>} \otimes  F_{j\ell}\\
    c^{(2)}_6 &= \tfrac12 \lambda^i \wedge \lambda^j \otimes \delta^{k\ell} F_{ik} \otimes  G_{\left<j\ell\right>}\\
    c^{(2)}_7 &= \tfrac12 \lambda^i \wedge \lambda^j \otimes \delta^{k\ell} F_{ik} \otimes  F_{j\ell}.
  \end{split}
\end{equation}

A calculation of the Chevalley--Eilenberg differential $\d \colon
C^0(\fa,E^* \otimes E^*)^{\fk}\to C^1(\fa,E^* \otimes E^*)^{\fk}$
shows that $\d c^{(0)}_a = 0$ for $a =1,3,4$ and that
\begin{equation}
  \d c^{(0)}_2 = c^{(1)}_2 + c^{(1)}_5 + \tfrac 1{24} \left( c^{(1)}_1 + c^{(1)}_3 \right).
\end{equation}

A calculation of the Chevalley--Eilenberg differential
$\d \colon C^1(\fa,E^* \otimes E^*)^{\fk}\to C^2(\fa,E^* \otimes
E^*)^{\fk}$ already shows that $\d c^{(1)}_a = 0$ for $a =1,2,3,5$,
whence $\dim H^1(\fa,E^* \otimes E^*)^{\fk} \geq 3$, but a closer look
reveals that actually $\dim H^1(\fa,E^* \otimes E^*)^{\fk} = 3$.

In summary, we cannot conclude that the extension class is zero a priori and
we would need to compute the extension class explicitly. This requires
exhibiting a $\fk$-equivariant splitting $\Hrel^2(p) \to \sH^3(p)$ and
hence requires calculating $\sH^3(p)$, which we have not done in this
paper.  Perhaps we return to this calculation in the future.

\end{appendices}

\bibliographystyle{utphys}
\bibliography{AmbiCoh}

\end{document}